\documentclass[11pt]{amsart}
\def\isdraft{1}

\usepackage[dvipsnames]{xcolor}
\usepackage{bbm}
\usepackage{graphicx}
\usepackage{bbm}
\usepackage{amsaddr} % remove in amsbook, JAIR, IOS-Press
\usepackage{mathtools}
\usepackage{amsmath,amssymb,amsfonts,dsfont}
\usepackage{prettyref} % achtung: position dieses kommandos scheint essentiell zu sein
\usepackage[utf8]{inputenc}
\usepackage{enumitem}
\usepackage[hyphens]{url} % remove in JAIR
\usepackage{tikz-cd}
\usepackage{tikz}
\usetikzlibrary{positioning}
\usetikzlibrary{arrows}
\usepackage{bussproofs}
\usepackage[mathscr]{euscript}
\usepackage{hyperref}
\usepackage{amsthm}
\usepackage{theapa}
\usepackage{a4wide}
\usepackage[color=black,textcolor=white\if\isdraft0,disable\fi]{todonotes}
\usepackage{stackengine}
\usepackage[normalem]{ulem}
\usepackage{stmaryrd}
\usepackage{wrapfig}
\usepackage{pdfpages}
\usepackage{marginnote}
\usepackage{float}
\graphicspath{{fig/}}
\usetikzlibrary{arrows,automata,topaths,matrix,positioning,fit}
\usepgflibrary{shapes.geometric}
\usetikzlibrary{shapes.geometric}
\tikzset{every state/.style={minimum size=0pt}}
\usepackage{footnote}
\usepackage{float}
\usepackage{tabularx}
\usepackage{txfonts}

\newtheorem{theorem}{Theorem}

\newtheorem{corollary}[theorem]{Corollary}
\newtheorem{fact}[theorem]{Fact}
\newtheorem{lemma}[theorem]{Lemma}
\newtheorem{proposition}[theorem]{Proposition}

\theoremstyle{definition} % TPLP remove

\newtheorem{convention}[theorem]{Convention}
\newtheorem{definition}[theorem]{Definition}

\newtheorem{example}[theorem]{Example}

\newtheorem{notation}[theorem]{Notation}

\newtheorem{problem}[theorem]{Problem}
\newtheorem{pseudocode}[theorem]{Pseudocode}

\newtheorem{remark}[theorem]{Remark}

\newrefformat{§}{§\ref{#1}}
\newrefformat{algorithm}{Algorithm \ref{#1}}
\newrefformat{a}{Answer \ref{#1}}
\newrefformat{appendix}{Appendix \ref{#1}}
\newrefformat{claim}{Claim \ref{#1}}
\newrefformat{conclusion}{Conclusion \ref{#1}}
\newrefformat{convention}{Convention \ref{#1}}
\newrefformat{conjecture}{Conjecture \ref{#1}}
\newrefformat{c}{Corollary \ref{#1}}
\newrefformat{equ}{(\ref{#1})}
\newrefformat{e}{Example \ref{#1}}
\newrefformat{exe}{Exercise \ref{#1}}
\newrefformat{d}{Definition \ref{#1}}
\newrefformat{done}{Done \ref{#1}}
\newrefformat{f}{Fact \ref{#1}}
\newrefformat{fig}{Figure \ref{#1}}
\newrefformat{h}{Hypothesis \ref{#1}}
\newrefformat{idea}{Idea \ref{#1}}
\newrefformat{i}{Item \ref{#1}}
\newrefformat{l}{Lemma \ref{#1}}
\newrefformat{note}{Note \ref{#1}}
\newrefformat{n}{Notation \ref{#1}}
\newrefformat{o}{Observation \ref{#1}}
\newrefformat{problem}{Problem \ref{#1}}
\newrefformat{p}{Proposition \ref{#1}}
\newrefformat{pseudocode}{Pseudocode \ref{#1}}
\newrefformat{Q}{Q. \ref{#1}}
\newrefformat{r}{Remark \ref{#1}}
\newrefformat{table}{Table \ref{#1}}
\newrefformat{t}{Theorem \ref{#1}}
\newrefformat{Todo}{Todo \ref{#1}}
\newrefformat{w}{Warning \ref{#1}}

\newcommand{\righttherefore}{:\joinrel\cdot\,}

\title{
    Analogical Proportions 
}
\author{
    Christian Anti\'c
}
\address{
    christian.antic@icloud.com\\
    Vienna, Austria
}

\begin{document}
\begin{abstract} 
    Analogy-making is at the core of human and artificial intelligence and creativity with applications to such diverse tasks as proving mathematical theorems and building mathematical theories, common sense reasoning, learning, language acquisition, and story telling. This paper introduces from first principles an abstract algebraic framework of analogical proportions of the form `$a$ is to $b$ what $c$ is to $d$' in the general setting of universal algebra. This enables us to compare mathematical objects possibly across different domains in a uniform way which is crucial for AI-systems. It turns out that our notion of analogical proportions has appealing mathematical properties. As we construct our model from first principles using only elementary concepts of universal algebra, and since our model questions some basic properties of analogical proportions presupposed in the literature, to convince the reader of the plausibility of our model we show that it can be naturally embedded into first-order logic via model-theoretic types and prove from that perspective that analogical proportions are compatible with structure-preserving mappings. This provides conceptual evidence for its applicability. In a broader sense, this paper is a first step towards a theory of analogical reasoning and learning systems with potential applications to fundamental AI-problems like common sense reasoning and computational learning and creativity.
\end{abstract}
\maketitle

\section{Introduction}

Analogy-making is at the core of human and artificial intelligence and creativity with applications to such diverse tasks as proving mathematical theorems and building mathematical theories, common sense reasoning, learning, language acquisition, and story telling \cite<e.g.>{Boden98,Gust08,Hofstadter01,Hofstadter13,Krieger03,Polya54,Sowa03,Winston80,Wos93}. This paper introduces from first principles an abstract algebraic framework of analogical proportions of the form `$a$ is to $b$ what $c$ is to $d$' in the general setting of universal algebra. This enables us to compare mathematical objects possibly across {\em different} domains in a uniform way which is crucial for AI-systems. The main idea is simple and is illustrated in the following example.% to define solutions to analogical equations in terms of maximal sets of algebraic justifications which enables us to find analogous elements in an `unknown' target domain given `known' elements in the source domain.

\begin{example}\label{exa:exa} Imagine two domains, one consisting of positive integers $1,2,\ldots$ and the other made up of words $ab,ba\ldots$ et cetera. The analogical equation
\begin{align}\label{equ:24abz} 
    2:4::ab:x
\end{align} is asking for some word $x$ (here $x$ is a variable) which is to $ab$ what $4$ is to $2$. What can be said about the relationship between 2 and 4? One simple observation is that 4 is the square of 2. Now, by analogy, what is the `square' of $ab$? If we interpret `multiplication' of words as concatenation---a natural choice---then $(ab)^2$ is the word $abab$, which is a plausible solution to \prettyref{equ:24abz}. We can state this more formally as follows. Let $s(z):=z$ and $t(z):=z^2$ be two terms. We have 
\begin{align}\label{equ:2=s(2)} 
    2=s(2),\quad 4=t(2), \quad\text{and}\quad ab=s(ab).
\end{align} By continuing the pattern in \prettyref{equ:2=s(2)}, what could $x$ in \prettyref{equ:24abz} be equal to? In \prettyref{equ:2=s(2)}, we see that transforming $2$ into $4$ means transforming $s(2)$ into $t(2)$. Now what does it mean to transform $ab$ `in the same way' or `analogously'? The obvious answer is to transform $s(ab)$ into $t(ab)=abab$ computed before. As simple as this line of reasoning may seem, it cannot be formalized by some current models of analogical proportions which restrict themselves to proportions between objects of a single domain \cite<cf.>{Stroppa06,Miclet08} and we will return to this specific analogical proportion in a more formal manner in \prettyref{exa:2_4_ab_z}.
\end{example}

It is important to emphasize that we do not want our model to capture only obvious analogical proportions as in the example above. To the contrary, in the more interesting and `creative' cases, analogical proportions provide a tool for deriving {\em unexpected} and often hypothetical conclusions, which if necessary can afterwards be checked for plausibility in a specific context. This process is similar to non-monotonic reasoning, where this kind of `guess and check' paradigm is common \cite<cf.>{Eiter09}. The next example illustrates the idea of formalizing aspects of creativity via unexpected analogical proportions.

\begin{example}\label{exa:2_0_3_z} Consider the analogical equation over the integers given by
\begin{align*} 
    2:0::3:x. 
\end{align*} An intuitive solution is given, for example, by $x=1$ justified via $0=2-2$ and $1=3-2$. However, there is a less obvious but still reasonably justifiable solution to the above equation: $x=1000$! To see why, take a look at the following analogical proportion:
\begin{align*} 
    [\underline{-3}+5]:[1000\cdot(\underline{-3})+3000]::[\underline{-2}+5]:[1000\cdot(\underline{-2})+3000].
\end{align*} Observe that we transform $\underline{-3}+5=2$ into $\underline{-2}+5=3$ by replacing $-3$ with $-2$, and {\em analogously} we transform $1000\cdot(\underline{-3})+3000=0$ into $1000\cdot(\underline{-2})+3000=1000$ again by replacing $-3$ with $-2$. This (together with some technical argument; see \prettyref{exa:0_0_2_40}) justifies the proportion 
\begin{align*} 
    2:0::3:1000.
\end{align*} We believe that finding such hidden transformations between seemingly unrelated objects is crucial for formalizing creativity \cite<cf.>{Boden98}. We will come back to this proportion in \prettyref{exa:0_0_2_40}.
\end{example}

The rest of the paper is devoted to formalizing and studying reasoning patterns as in the examples above within the abstract algebraic setting of universal algebra. The aim of this paper is to introduce our model of analogical proportions---which to the best of our knowledge is novel---in its full generality. The core idea is formulated in \prettyref{def:models} and despite its conceptual simplicity (it consists of three parts, where the second and third parts are symmetrical variants of the first) it has interesting consequences with mathematically appealing proofs, which we plan to explore further in the future. %More precisely, this theoretical paper is declarative in nature---it clarifies {\em what} is meant by an analogical proportion across two (different) domains without saying much about {\em how} analogical proportions are to be (algorithmically) computed, which is highly non-trivial in the general case and beyond the scope of this paper. 
Since `plausible analogical proportion' is an informal concept, we cannot hope to formally prove the soundness and completeness of our framework---the best we can do is to prove that some desirable proportions are derivable within our framework (e.g. Theorems \ref{thm:FPT}, \ref{thm:axioms}, \ref{thm:FIT}, \ref{thm:SIT}) and that some implausible proportions cannot be derived (e.g. \prettyref{thm:axioms} and Examples \ref{exa:eta_hom} and \ref{exa:aaa0}). We compare our framework with two recently introduced frameworks of analogical proportions from the literature \cite{Miclet08,Stroppa06}---introduced for applications to artificial intelligence and machine learning, specifically for natural language processing and handwritten character recognition---within the concrete domains of sets and numbers, and in each case we either disagree with the notion from the literature justified by some counter-example (\prettyref{exa:aaa0}) or we can show that our model yields strictly more justifiable solutions (\prettyref{exa:ab0Z}), which provides evidence for its applicability. Finally, in the 2-element boolean setting, we argue in \prettyref{sec:Related_Work} that our model coincides with Klein's \cite{Klein82} model and that it reasonably disagrees with Miclet and Prade's \cite{Miclet09} framework, which is remarkable as our framework is not geared towards the boolean domain.

\citeA{Lepage03} proposes four axioms---namely symmetry $a:b::c:d\Leftrightarrow c:d::a:b$, central permutation $a:b::c:d\Leftrightarrow a:c::b:d$, strong inner reflexivity $a:a::c:d\Rightarrow d=c$, and strong reflexivity $a:b::a:d\Rightarrow d=b$---as a guideline for formal models of analogical proportions. To be more precise, \citeA{Lepage03} introduces his axioms in the linguistic setting of words and although his axioms appear reasonable in the word domain (but see \prettyref{problem:Lepage}), the following counter-examples show that they cannot be straightforwardly applied to the general case. %For instance, consider the domain $\mathbb Z$ of integers ordered by $<$ without any other structure on $\mathbb Z$---here $0<1$ and $-1<2$ reasonably imply $0:1::-1:2$. However, one can argue that the permuted proportion $0:-1::1:2$ fails since $0>-1$ whereas $1<2$. 
Strong reflexivity fails, for instance, in cases where the relation of $a$ to $b$ and to $d$ is identical, for some distinct elements $b$ and $d$. Strong inner reflexivity fails, for instance, if the relation of $a$ to itself is similar to the relation of $c$ to $d$. In our framework, by making the underlying structures of an analogical proportion explicit, it turns out that except for symmetry none of Lepage's axioms holds in the general case, justified by counter-examples (\prettyref{thm:axioms}). This has critical consequences as his axioms are assumed by many authors \cite<e.g.>{Barbot19,Miclet08} to hold beyond the word domain. % For instance, central permutation and strong reflexivity---a property that is satisfied in every structure---imply strong determinism, which again should not be treated as an axiom (\prettyref{thm:axioms}).
We adapt Lepage's list of axioms by including symmetry, and by adding inner symmetry $a:b::c:d\Leftrightarrow b:a::d:c$, inner reflexivity $a:a::c:c$, and reflexivity $a:b::a:b$, and determinism $a:a::a:d \quad\Leftrightarrow\quad a=d$ to the list. Notice that inner reflexivity and reflexivity are weak forms of Lepage's strong inner reflexivity and strong reflexivity axioms, respectively, whereas inner symmetry is a variant of Lepage's symmetry axiom which requires symmetry to hold within the respective structures. Moreover, we consider the properties of commutativity $a:b::b:a$, transitivity $a:b::c:d\;\&\;c:d::e:f \Rightarrow a:b::e:f$, inner transitivity $a:b::c:d\;\&\;b:e::d:f \Rightarrow a:e::c:f$, and central transitivity $a:b::b:c\;\&\;b:c::c:d \Rightarrow a:b::c:d$. We prove that the inner symmetry, inner reflexivity, reflexivity, and determinism axioms are satisfied within our framework, whereas commutativity, transitivity, inner transitivity, and central transitivity fail in general (\prettyref{thm:axioms}). This shows that the property of being in analogical proportion is a {\em local} property (\prettyref{rem:local}). This is in contrast to category theory \cite<cf.>{Awodey10}---the algebraic field for formalizing mathematical analogies---where the transitivity of arrows leads to a form of `connectedness' which in general is not present in analogical proportions (see the discussion in \prettyref{sec:Category_Theory}). % The interesting fact that everyday human analogy-making is instantaneous although the brain...

Interestingly, analogical proportions turn out to be {\em non-monotonic} in the sense that expanding the underlying structure of an analogical proportion may prevent its derivation (\prettyref{thm:axioms}). This may have interesting connections to non-monotonic reasoning, which itself is crucial for common sense reasoning and which has been prominently formalized within the field of answer set programming \cite{Gelfond91} \cite<cf.>{Brewka11}). % It is interesting to study in which structures central permutation or strong determinism hold, which is beyond the scope of this paper and which we leave as future work.

The functional-based view in \cite{Barbot19} is related to our Functional Proportion \prettyref{thm:FPT} on functional solutions and the preservation of functional dependencies across different domains, which means that in case $t(z)$ is a transformation, satisfying a mild injectivity condition and applicable in the source and target domains, we have $a:t(a)::c:t(c)$. The critical difference is that the authors of \cite{Barbot19} assume Lepage's central permutation axiom, which implies in their framework that the functional transformation $t$ need to be bijective.

Analogical proportions turn out to be compatible with structure-preserving mappings as shown in our First and Second Isomorphism Theorems \ref{thm:FIT} and \ref{thm:SIT}, respectively---a result which is in the vein of Gentner's \cite{Gentner83} Structure-Mapping Theory (SMT) of analogy-making (see the brief discussion in \prettyref{sec:Related_Work}).

As we construct our model from first principles using only elementary concepts of universal algebra, and since our model questions some fundamental properties of analogical proportions presupposed in the literature, to convince the reader of the plausibility of our model we need to validate it either empirically or---what we prefer here---theoretically by showing that it fits naturally into the overall mathematical landscape. For this, we will show in \prettyref{sec:Logical_} that our purely algebraic model can be naturally embedded into first-order logic via model-theoretic types \cite<cf.>[§7.1]{Hinman05}. More precisely, we show that sets of algebraic justifications of analogical proportions are in one-to-one correspondence with so-called `rewrite formulas' and `rewrite types', which is appealing as types play a fundamental role in model theory in general and showing that our model---which is primarily motivated by simple examples---has a natural logical interpretation, provides strong evidence for its applicability. In \prettyref{sec:Isomorphisms}, we then prove from this logical perspective that analogical proportions are compatible with isomorphisms, a further desired property.\\

In a broader sense, this paper is a first step towards a theory of analogical reasoning and learning systems with potential applications to fundamental AI-problems like common sense reasoning and computational learning and creativity.

\section{Preliminaries}\label{sec:Preliminaries}

Given any sequence of objects $\mathbf o=o_1\ldots o_n$, $n\geq 0$, we denote the {\em length} $n$ of $\mathbf o$ by $|\mathbf o|$. We denote the {\em power set} of a set $U$ by $\mathfrak P(U)$. The natural numbers are denoted by $\mathbb N:=\{0,1,2,\ldots\}$, and the natural numbers not containing $0$ and $1$ by $\mathbb N_2:=\{2,3,\ldots\}$; the integers are denoted by $\mathbb Z$ and the rational numbers are denoted by $\mathbb Q$. Moreover, the booleans are denoted by $\mathfrak{BOOL}:=\{0,1\}$, with conjunction $0\land 0:=1\land 0:=0\land 1:=0$ and $1\land 1:=1$, and disjunction $0\lor 0:=0$ and $1\lor 0:=0\lor 1:=1\lor 1:=1$. Given a finite alphabet $\Sigma$, we denote the set of all finite words over $\Sigma$ containing the {\em empty word} $\varepsilon$ by $\Sigma^\ast$ and we define $\Sigma^+:=\Sigma^\ast-\{\varepsilon\}$.

\subsection{Universal Algebra}

We recall some basic notions and notations of universal algebra \cite<e.g.>{Burris00}.

% todo add that we do not distinguish between $x$ and $x+0$ ... 
\subsubsection{Syntax}

A {\em language} $L$ of algebras consists of a set $Fs_L$ of {\em function symbols}, a set $Cs_L$ of {\em constant symbols}, a {\em rank function} $rk_L:Fs_L\to\mathbb N-\{0\}$, and a denumerable set $V=\{x,z_1,z_2,\ldots\}$ of {\em variables}. The sets $Fs_L,Cs_L,$ and $V$ are pairwise disjoint. Moreover, we always assume that $L$ contains the {\em equality relation symbol} $=$ interpreted as the equality relation in every algebra. A language $L$ is a {\em sublanguage} of a language $L'$---in symbols, $L\subseteq L'$---iff $Fs_L\subseteq Fs_{L'}$, $Cs_L\subseteq Cs_{L'}$, and $rk_L$ is the restriction of $rk_{L'}$ to the function symbols in $L$. In the special case that $Fs_L=Fs_{L'}$, we say that $L'$ is an {\em extension by constants}. An {\em $L$-expression} is any finite string of symbols from $L$. An {\em $L$-atomic term} is either a variable or a constant symbol. The set $T_L(V)$ of {\em $L$-terms} is the smallest set of $L$-expressions such that (i) every $L$-atomic term is an $L$-term; and (ii) for any $L$-function symbol $f$ and any $L$-terms $t_1,\ldots,t_{rk_L(f)}$, $f(t_1,\ldots,t_{rk_L(f)})$ is an $L$-term. We denote the set of variables occurring in a term $t$ by $V(t)$. We say that $t$ has {\em rank} $n$ if $V(t)\subseteq\{z_1,\ldots,z_n\}$. A term is {\em ground} if it contains no variables. 

\subsubsection{Semantics}

An {\em $L$-algebra} $\mathfrak A$ consists of (i) a non-empty set $A$, the {\em universe} of $\mathfrak A$; (ii) for each $f\in Fs_L$, a function $f^\mathfrak A:A^{rk_L(f)}\to A$, the {\em functions} of $\mathfrak A$; and (iii) for each $c\in Cs_L$, an element $c^\mathfrak A\in A$, the {\em distinguished elements} of $\mathfrak A$. %In case $\mathfrak A=(A,Fs^\mathfrak A_L,Cs^\mathfrak A_L)$ with {\em finite} sequences $Fs^\mathfrak A_L=(f^\mathfrak A)_{f\in Fs_L}=(f^\mathfrak A_1,\ldots,f^\mathfrak A_m)$ and $Cs^\mathfrak A_L=(c^\mathfrak A)_{c\in Cs_L}=(c^\mathfrak A_1,\ldots,c^\mathfrak A_n)$, we simply write $\mathfrak A=(A,f^\mathfrak A_1,\ldots,f^\mathfrak A_m,c^\mathfrak A_1,\ldots,c^\mathfrak A_n)$.

\begin{notation}\label{not:L}  With a slight abuse of notation, we will not distinguish between an $L$-algebra $\mathfrak A$ and its universe $A$ in case the operations are understood from the context. This means we will write $a\in A$ instead of $a\in A$ et cetera. Moreover, given a subset $A'$ of the universe of $\mathfrak A$, the language $L(A')$ is the language $L$ augmented by a constant symbol $a$ for each element $a\in A'$. % Wird in \prettyref{sec:Set_} verwendet, um $L(\mathfrak P(U)\cap\mathfrak P(W))$-algebras zu definieren
\end{notation}

% see \citeA[Def. 2.6.3]{Hinman05}
Given an $L$-algebra $\mathfrak A$ and an $L'$-algebra $\mathfrak A'$, for some languages $L\subseteq L'$, we say that $\mathfrak A$ is an {\em $L$-reduct} of $\mathfrak A'$ and $\mathfrak A'$ is an {\em $L'$-expansion} of $\mathfrak A$---in symbols, $\mathfrak A=\mathfrak A'\upharpoonright L$---iff $A=A'$, $f^\mathfrak A=f^{\mathfrak A'}$ for all $f\in Fs_L$, and $c^\mathfrak A=c^{\mathfrak A'}$ for all $c\in Cs_L$. % Given two $L$-algebras $\mathfrak A=(A,(f^\mathfrak A)_{f\in Fs_L},(c^\mathfrak A)_{c\in Cs_L})$ and $\mathfrak B=(B,(f^\mathfrak B)_{f\in Fs_L},(c^\mathfrak B)_{c\in Cs_L})$, we define $\mathfrak{A\cap B}:=(A\cap B,(f^\mathfrak A\cap f^\mathfrak B)_{f\in Fs_L},(c^\mathfrak A)_{c\in Cs_L})$\todo{was passiert mit constants?}.
%Given two $L$-algebras $\mathfrak A$ and $\mathfrak A'$, we say that $\mathfrak A$ is a {\em subalgebra} of $\mathfrak A'$---in symbols, $\mathfrak A\subseteq\mathfrak A'$---iff $A\subseteq A'$ and every function $f^\mathfrak A$ is the restriction of $f^{\mathfrak A'}$ to $A^{rk_L(f)}$.

For any $L$-algebra $\mathfrak A$, an {\em $\mathfrak A$-assignment} is a function $\nu:V\to\mathfrak A$. For any assignment $\nu$, let $\nu_{z\mapsto a}$ denote the assignment $\nu'$ such that $\nu'(z):=a$, and for all other variables $z'$, $\nu'(z'):=\nu(z')$. We extend the domain of the $\mathfrak A$-assignment $\nu$ from variables in $V$ to terms in $T_L(V)$ inductively as follows: (i) for every $c\in Cs_L$, $\nu(c):=c^\mathfrak A$; (ii) for every $f\in Fs_L$ and $t_1,\ldots,t_{rk_L(f)}\in T_L(V)$, $\nu(f(t_1,\ldots,t_{rk_L(f)})):=f^\mathfrak A(\nu(t_1),\ldots,\nu(t_{rk_L(f)}))$. Notice that every term $t$ induces a function on $\mathfrak A$
\begin{align*} 
    t^\mathfrak A:A^{rk_L(t)}\to A
\end{align*} given by
\begin{align*} 
    t^\mathfrak A(a_1,\ldots,a_{|V(t)|}):=\nu_{(a_1,\ldots,a_{|V(t)|})}(t),
\end{align*} where $\nu_{(a_1,\ldots,a_{|V(t)|})}(z_i):=a_i$, for all $1\leq i\leq |V(t)|$. %We denote the {\em identity function} on $\mathfrak A$ by $id_\mathfrak A$ or simply by $id$ in case $\mathfrak A$ is understood. 
%Given an $L$-algebra $\mathfrak A$, an {\em $L$-term} is an $L$-term which may contain distinguished elements of $\mathfrak A$ as constant symbols with the obvious interpretation. We denote the set of all $L$-terms with variables among $\mathbf z=z_1,\ldots,z_n$, $n\geq 0$, by $\mathfrak A[\mathbf z]$. For instance, $2z+1$ is a term in $(\mathbb N,+,1)[z]$ as $2z$ is an abbreviation for $z+z$ not containing 2, whereas $2z^2+1$ is {\em not} as $z^2=z\cdot z$ requires multiplication. 
We call a term $t$ {\em constant} in $\mathfrak A$ iff $t^\mathfrak A$ is a constant function, and we call $t$ {\em injective} in $\mathfrak A$ iff $t^\mathfrak A$ is an injective function. For instance, the term $t(z)=0z$ is constant in $(\mathbb N,\cdot,0)$ despite containing the variable $z$. Terms can be interpreted as `generalized elements' containing variables as placeholders for concrete elements, and they will play a central role in our algebraic formulation of analogical proportions given below.

A {\em homomorphism} from $\mathfrak A$ to $\mathfrak B$ is a mapping $ H:\mathfrak A\to\mathfrak B$ such that for any function symbol $f\in Fs_L$ and any elements $a_1,\ldots,a_{rk_L(f)}$,
\begin{align*} 
     H\left(f^\mathfrak A(a_1,\ldots,a_{rk_L(f)})\right)=f^\mathfrak B\left( H(a_1),\ldots, H(a_{rk_L(f)})\right).
\end{align*} An {\em isomorphism} is a bijective homomorphism, and we call two algebras $\mathfrak A$ and $\mathfrak B$ {\em isomorphic}---in symbols, $\mathfrak A\simeq\mathfrak B$---iff there exists an isomorphism from $\mathfrak A$ to $\mathfrak B$.

\subsection{First-Order Logic}

We recall the syntax and semantics of first-order logic, restricted to functional structures containing no relation symbols, by mainly following the lines of \cite{Hinman05}.

\subsubsection{Syntax}

A {\em first-order language} $L$ is a language of algebras extended by {\em relation symbols}, the {\em connectives} $\neg$ and $\lor$, and the {\em existential quantifier} $\exists$. We call $L$ a {\em functional} language if it contains no relation symbols other than equality. Notice that there is no real difference between a language of algebras and a functional first-order language; however, we make the formal distinction here to highlight the fact that full first-order logic and model theory containing relation symbols is strictly more expressive than plain universal algebra, summed up by \citeA{Chang73a} as
\begin{align*} \text{universal algebra} + \text{logic}=\text{model theory}.
\end{align*} 

\begin{notation} In this paper, we always assume that $L$ is functional.
\end{notation}

An {\em $L$-atomic formula} has the form $s=t$, where $s$ and $t$ are $L$-terms and we denote the set of all such formulas by $aFm_L$. The set of {\em $L$-formulas} is the smallest set of $L$-expressions such that (i) every $L$-atomic formula is an $L$-formula; (ii) if $\varphi$ and $\psi$ are $L$-formulas, then so are $\neg\varphi$ and $\varphi\lor\psi$; (iii) if $\varphi$ is an $L$-formula and $z\in V$ is a variable, then $\exists z\,\varphi$ is an $L$-formula. We introduce the following abbreviations: %\footnote{In what follows, we will use the notation `$:\equiv$' to denote a logical definition.}
$\varphi\land\psi:\equiv\neg(\neg\varphi\lor\neg\psi)$ and $\forall z\,\varphi:\equiv\neg\exists z\neg\varphi$. A {\em 2-formula} is a formula containing exactly two free variables.

\subsubsection{Semantics}

An {\em $L$-structure} $\mathfrak A$ is the same as an $L$-algebra in case $L$ is functional. We define the {\em logical entailment relation} inductively as follows: for any functional $L$-structure $\mathfrak A$ and any $\mathfrak A$-assignment $\nu$, (i) for any $L$-terms $s$ and $t$, $\mathfrak A\models (s=t)[\nu]$ iff $\nu(s)=\nu(t)$; (ii) for any formula $\varphi$, $\mathfrak A\models\neg\varphi[\nu]$ iff $\mathfrak A\not\models\varphi[\nu]$; (iii) for any formulas $\varphi$ and $\psi$, $\mathfrak A\models\varphi\lor\psi[\nu]$ iff $\mathfrak A\models\varphi[\nu]$ or $\mathfrak A\models\psi[\nu]$; (iv) for any formula $\varphi$ and variable $z$, $\mathfrak A\models\exists z\;\varphi[\nu]$ iff $\mathfrak A\models\varphi[\nu_{z\mapsto a}]$, for some $a\in A$. With the abbreviations introduced above, we further have, for any formulas $\varphi$ and $\psi$, (i) $\mathfrak A\models\varphi\land\psi[\nu]$ iff $\mathfrak A\models\varphi[\nu]$ and $\mathfrak A\models\psi[\nu]$; and (ii) $\mathfrak A\models\forall z\;\varphi[\nu]$ iff $\mathfrak A\models\varphi[\nu_{z\mapsto a}]$, for all $a\in A$. For any function $ H:\mathfrak A\to\mathfrak B$ and any $\mathfrak A$-assignment $\nu$, we define the {\em $ H$-induced $\mathfrak B$-assignment} by $ H(\nu)(z):= H(\nu(z))$. We say that $ H$ {\em respects} (i) a term $t$ iff for each $\mathfrak A$-assignment $\nu$, $ H(\nu(t))= H(\nu)(t)$; or, in other words, iff for each $\mathbf e\in A^{rk_L(t)}$, $ H(t^\mathfrak A(\mathbf e))=t^\mathfrak B( H(\mathbf e))$, where $ H(\mathbf e)$ means component-wise application; (ii) a formula $\varphi$ iff for each $\mathfrak A$-assignment $\nu$, $\mathfrak A\models\varphi[\nu]$ iff $\mathfrak B\models\varphi[ H(\nu)]$. The following result will be useful in \prettyref{sec:Isomorphisms} for reproving our First Isomorphism \prettyref{thm:FIT} in \prettyref{sec:Isomorphism_Theorems}; its proof can be found in \cite[Lemma 2.3.6]{Hinman05}.

\begin{lemma}\label{lem:respects} For any $\mathfrak A$ and $\mathfrak B$ and any $ H:\mathfrak A\to\mathfrak B$, $ H$ respects all $L$-terms and formulas.
\end{lemma}

\section{Analogical Proportions}\label{sec:Analogical_Proportions}

In the rest of the paper, we may assume some `known' source domain $\mathfrak A$ and some `unknown' target domain $\mathfrak B$, both $L$-algebras of same language $L$. We may think of the source domain $\mathfrak A$ as our background knowledge---a repertoire of elements we are familiar with---whereas $\mathfrak B$ stands for an unfamiliar domain which we want to explore via analogical transfer from $\mathfrak A$. For this we will first consider arrow equations of the form `$a$ transforms into $b$ as $c$ transforms into $x$'---in symbols, $a\to b\righttherefore c\to x$---where $a$ and $b$ are source elements of $\mathfrak A$, $c$ is a target element of $\mathfrak B$, and $x$ is a variable. Solutions to arrow equations will be elements of $\mathfrak B$ which are obtained from $c$ in $\mathfrak B$ as $b$ is obtained from $a$ in $\mathfrak A$ in a mathematically precise way (\prettyref{def:models}). Specifically, we want to functionally relate elements of an algebra via term rewrite rules as follows. Recall from \prettyref{exa:exa} that transforming 2 into 4 in the algebra $(\mathbb N,\cdot)$ means transforming $s(2)$ into $t(2)$,\footnote{To be more precise, we transform $s^{(\mathbb N,\cdot)}(2)$ into $t^{(\mathbb N,\cdot)}(2)$.} where $s(z):=z$ and $t(z):=z^2$ are terms. We can state this transformation more pictorially as the term rewrite rule $s\to t$ or $z\to z^2$. Now transforming the word $ab$ `in the same way' means to transform $s(ab)$ into $t(ab)$, which again is an instance of $s\to t$. Let us make this notation official.

\begin{notation}\label{not:s_to_t} We will always write $s(\mathbf z)\to t(\mathbf z)$ or $s\to t$ instead of $(s,t)$, for any pair of $L$-terms $s$ and $t$ containing variables among $\mathbf z$ such that every variable in $t$ occurs in $s$. We call such expressions {\em $L$-rewrite rules} or {\em $L$-justifications} where we often omit the reference to $L$. We denote the set of all $L$-justifications with variables among $\mathbf z$ by $J(L,\mathbf z)$. We make the convention that $\to$ binds weaker than every other algebraic operation.
\end{notation}

\begin{definition}\label{def:Jus} Define the {\em set of justifications} of two elements $a,b\in A$ in $\mathfrak A$ by
\begin{align*} 
    Jus_\mathfrak A(a,b):=\left\{s\to t\in J(L,\mathbf z) \;\middle|\; a=s^\mathfrak A(\mathbf e)\text{ and }b=t^\mathfrak A(\mathbf e),\text{ for some }\mathbf e\in A^{|\mathbf z|}\right\}.
\end{align*}
\end{definition}

For instance, in the example above, $Jus_{(\mathbb N,\cdot)}(2,4)$ and $Jus_{(\{a,b\}^\ast,\cdot)}(ab,abab)$ both contain the justification $z\to z^2$, witnessed by $\mathbf e_1=2\in\mathbb N$ and $\mathbf e_2=ab\in\{a,b\}^\ast$.

Once we have a definition of $a\to b\righttherefore c\to d$, we can define $a:b\righttherefore c:d$ and, finally, $a:b::c:d$ via appropriate symmetries as follows.

\begin{definition}\label{def:models} We define the {\em analogical proportion relation} in three steps:
\begin{enumerate}
\item An {\em arrow equation} in $(\mathfrak{A,B})$ is an expression of the form `$a$ transforms into $b$ in $\mathfrak A$ as $c$ transforms into $x$ in $\mathfrak B$'---in symbols,
\begin{align}\label{equ:a->b>c->x} 
    a\to b\righttherefore c\to x,
\end{align} where $a$ and $b$ are source elements from $\mathfrak A$, $c$ is a target element from $\mathfrak B$, and $x$ is a variable. Given a target element $d\in B$, define the {\em set of justifications} of an arrow proportion $a\to b\righttherefore c\to d$ in $(\mathfrak{A,B})$ by
\begin{align*} 
    Jus_{(\mathfrak{A,B})}(a\to b\righttherefore c\to d):=Jus_\mathfrak A(a,b)\cap Jus_\mathfrak B(c,d).
\end{align*} A justification $s\to t$ is {\em trivial} in $(\mathfrak{A,B})$ iff it justifies every arrow proportion in $(\mathfrak{A,B})$, and we say that $J$ is a {\em trivial set of justifications} in $(\mathfrak{A,B})$ iff every justification in $J$ is trivial.\footnote{See Examples \ref{exa:tau_set} and \ref{exa:tau_numeric}.} Now we call $d$ a {\em solution} to \prettyref{equ:a->b>c->x} in $(\mathfrak{A,B})$ iff either $Jus_\mathfrak A(a,b)\cup Jus_\mathfrak B(c,d)$ consists only of trivial justifications, in which case there is neither a non-trivial transformation of $a$ into $b$ in $\mathfrak A$ nor of $c$ into $d$ in $\mathfrak B$; or $Jus_{(\mathfrak{A,B})}(a\to b\righttherefore c\to d)$ is maximal with respect to subset inclusion among the sets $Jus_{(\mathfrak{A,B})}(a\to b\righttherefore c\to d')$, $d'\in B$, containing at least one non-trivial justification, that is, for any element $d'\in \mathfrak B$,
\begin{align*} 
    \emptyset\subsetneq Jus_{(\mathfrak{A,B})}(a\to b\righttherefore c\to d)&\subseteq Jus_{(\mathfrak{A,B})}(a\to b\righttherefore c\to d')
\end{align*} implies
\begin{align*} 
    \emptyset\subsetneq Jus_{(\mathfrak{A,B})}(a\to b\righttherefore c\to d')\subseteq Jus_{(\mathfrak{A,B})}(a\to b\righttherefore c\to d).
\end{align*} In this case, we say that $a,b,c,d$ are in {\em arrow proportion} in $(\mathfrak{A,B})$ written
\begin{align*} 
    (\mathfrak{A,B})\models a\to b\righttherefore c\to d.
\end{align*} We denote the set of all solutions to \prettyref{equ:a->b>c->x} in $(\mathfrak{A,B})$ by
\begin{align*} 
    Sol_{(\mathfrak{A,B})}(a\to b\righttherefore c\to x).
\end{align*} We say that $a\to b\righttherefore c\to d$ is a {\em trivial arrow proportion} in $(\mathfrak{A,B})$ iff $(\mathfrak{A,B})\models a\to b\righttherefore c\to d$ and $Jus_{(\mathfrak{A,B})}(a\to b\righttherefore c\to d)$ consists only of trivial justifications.

\item A {\em directed analogical equation} in $(\mathfrak{A,B})$ is an expression of the form
\begin{align}\label{equ:a<->b>c<->x} 
    a:b\righttherefore c:x,
\end{align} where $a$ and $b$ are again source elements from $\mathfrak A$, $c$ is a target element from $\mathfrak B$, and $x$ is a variable. We call $d$ a {\em solution} to \prettyref{equ:a<->b>c<->x} in $(\mathfrak{A,B})$ iff
\begin{align*} 
    (\mathfrak{A,B})\models a\to b\righttherefore c\to d \quad\text{and}\quad (\mathfrak{A,B})\models b\to a\righttherefore d\to c.
\end{align*} In this case, we say that $a,b,c,d$ are in {\em directed analogical proportion} in $(\mathfrak{A,B})$ written
\begin{align*} 
    (\mathfrak{A,B})\models a:b\righttherefore c:d.
\end{align*} We denote the set of all solutions to \prettyref{equ:a<->b>c<->x} in $(\mathfrak{A,B})$ by
\begin{align*} 
    Sol_{(\mathfrak{A,B})}(a:b\righttherefore c:x).
\end{align*} We say that $a:b\righttherefore c:d$ is a {\em trivial directed proportion} in $(\mathfrak{A,B})$ iff $(\mathfrak{A,B})\models a:b\righttherefore c:d$, and $Jus_{(\mathfrak{A,B})}(a\to b\righttherefore c\to d)$ and $Jus_{(\mathfrak{A,B})}(b\to a\righttherefore d\to c)$ consist only of trivial justifications.

\item An {\em analogical equation} in $(\mathfrak{A,B})$ is an expression of the form `$a$ is to $b$ in $\mathfrak A$ what $c$ is to $x$ in $\mathfrak B$'---in symbols,
\begin{align}\label{equ:a<->b<>c<->x} 
    a:b::c:x,
\end{align} where $a$ and $b$ are again source elements from $\mathfrak A$, $c$ is a target element from $\mathfrak B$, and $x$ is a variable. We call $d$ a {\em solution} to \prettyref{equ:a<->b<>c<->x} in $(\mathfrak{A,B})$ iff
\begin{align*} 
    (\mathfrak{A,B})\models a:b\righttherefore c:d \quad\text{and}\quad (\mathfrak{B,A})\models c:d\righttherefore a:b.
\end{align*} In this case, we say that $a,b,c,d$ are in {\em analogical proportion} in $(\mathfrak{A,B})$ written
\begin{align*} 
    (\mathfrak{A,B})\models a:b::c:d.
\end{align*} We denote the set of all solutions to \prettyref{equ:a<->b<>c<->x} in $(\mathfrak{A,B})$ by
\begin{align*} 
    Sol_{(\mathfrak{A,B})}(a:b::c:x).
\end{align*} We say that $a:b::c:d$ is a {\em trivial analogical proportion} in $(\mathfrak{A,B})$ iff $(\mathfrak{A,B})\models a:b::c:d$ and $Jus_{(\mathfrak{A,B})}(a\to b\righttherefore c\to d)=Jus_{(\mathfrak{B,A})}(c\to d\righttherefore a\to b)$ and $Jus_{(\mathfrak{A,B})}(b\to a\righttherefore d\to c)=Jus_{(\mathfrak{B,A})}(d\to c\righttherefore b\to a)$ consist only of trivial justifications.
\end{enumerate}
\end{definition}

\begin{notation}\label{not:AA} We will always write $\mathfrak A$ instead of $(\mathfrak{A,A})$.
\end{notation}

\begin{convention}\label{con:trivial} In what follows, we will usually omit trivial justifications from notation. So, for example, we will write $Jus_{(\mathfrak{A,B})}(a\to b\righttherefore c\to d)=\emptyset$ instead of $Jus_{(\mathfrak{A,B})}(a\to b\righttherefore c\to d)=\{\text{trivial justifications}\}$ in case $a\to b\righttherefore c\to d$ has only trivial justifications in $(\mathfrak{A,B})$, et cetera. The empty set is always a trivial set of justifications. % and so is the set $V_\neq:=\{z_i\to z_j\mid z_i,z_j\in V,i\neq j\}$. (This clearly holds because $z_i^\mathfrak A(\mathbf e)=a$ and $z_j^\mathfrak A(\mathbf e)=b$ for any tuple $\mathbf e$ whose $i$-th and $j$-th components are $a$ and $b$, respectively). 
Every justification is meant to be non-trivial unless stated otherwise. Moreover, we will always write sets of justifications modulo renaming of variables, that is, we will write $\{z\to z\}$ instead of $\{z\to z\mid z\in V\}$ et cetera.
\end{convention}

% \begin{convention} In the sequel, we will write `is maximal' instead of `is maximal with respect to subset inclusion among the sets' when we talk about sets of justifications.
% \end{convention}

Roughly, an element $d$ in the target domain is a solution to an analogical equation of the form $a:b::c:x$ iff there is no other target element $d'$ whose relation to $c$ is more similar to the relation between $a$ and $b$ in the source domain (see \prettyref{rem:in_the_same_way}), expressed in terms of maximal sets of algebraic justifications satisfying appropriate symmetries. Analogical proportions formalize the idea that analogy-making is the task of transforming different objects from the source to the target domain in `the same way';\footnote{This is why `copycat' is the name of a prominent model of analogy-making  \cite{Hofstadter95a}. See \cite{Correa12}.} or as P\'olya \cite{Polya54} puts it:
\begin{quote} Two systems are analogous if they agree in clearly definable relations of their respective parts.
\end{quote} In our formulation, the `parts' are the elements $a,b,c,d$ and the `definable relations' are represented by term rewrite rules relating $a,b$ and $c,d$ in `the same way' via maximal sets of justifications.

\begin{example} First consider the algebra $\mathfrak A_1:=(\{a,b,c,d\})$, consisting of four distinct elements with no functions and no constants (see the forthcoming \prettyref{thm:(A)}):
\begin{center}
\begin{tikzpicture} 
    \node (a)               {$a$};
    \node (b) [above=of a]  {$b$};
    \node (c) [right=of a]  {$c$};
    \node (d) [above=of c]  {$d$};
\end{tikzpicture}
\end{center} Since $Jus_{\mathfrak A_1}(a',b')\cup Jus_{\mathfrak A_1}(c',d')$ contains only trivial justifications for {\em any distinct} elements $a',b',c',d'\in A'$, we have, for example:
\begin{align*} 
    \mathfrak A_1\models a:b::c:d \quad\text{and}\quad \mathfrak A_1\models a:c::b:d.
\end{align*} On the other hand, since (cf. \prettyref{con:trivial})
\begin{align*} 
    Jus_{\mathfrak A_1}(a,a)\cup Jus_{\mathfrak A_1}(a,d)=\{z\to z\}\neq\emptyset
\end{align*} and
\begin{align*} 
    \emptyset=Jus_{\mathfrak A_1}(a\to a\righttherefore a\to d)\subsetneq Jus_{\mathfrak A_1}(a\to a\righttherefore a\to a)=\{z\to z\},
\end{align*} we have
\begin{align*} 
    \mathfrak A_1\not\models a\to a\righttherefore a\to d,
\end{align*} which implies
\begin{align*} 
    \mathfrak A_1\not\models a:a::a:d.
\end{align*} This is an instance of the forthcoming determinism axiom \prettyref{equ:determinism} proved to hold within our framework in \prettyref{thm:axioms}.

Now consider the slightly different algebra $\mathfrak A_2:=(\{a,b,c,d\},f)$, where $f$ is the unary function defined by %(we omit the loops $f(o):=o$, $o\in\{b,c,d\}$, in the figure):
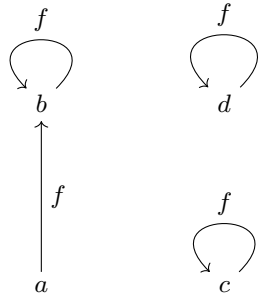
\begin{center}
\begin{tikzpicture} 
\node (a)               {$a$};
\node (b) [above=of a,yshift=1cm]  {$b$};
\node (c) [right=of a,xshift=1cm]  {$c$};
\node (d) [right=of b,xshift=1cm]  {$d$};

% siehe tikz 3.6.1a p. 165/166
\draw[->] (a) to [edge label'={$f$}] (b);
\draw[->] (b) to [edge label'={$f$}] [loop] (b);
\draw[->] (c) to [edge label'={$f$}] [loop] (c);
\draw[->] (d) to [edge label'={$f$}] [loop] (d);
\end{tikzpicture}
\end{center} We expect $a:b::c:d$ to fail in $\mathfrak A_2$ as it has no non-trivial justification. In fact, 
\begin{align*} Jus_{\mathfrak A_2}(a,b)\cup Jus_{\mathfrak A_2}(c,d)=\left\{z\to f^\ell(z) \;\middle|\; \ell\geq 1\right\}\neq\emptyset 
\end{align*} and
\begin{align*} 
    Jus_{\mathfrak A_2}(a\to b\righttherefore c\to d)=\emptyset
\end{align*} show
\begin{align*} \mathfrak A_2\not\models a:b::c:d.
\end{align*}

In the algebra $\mathfrak A_3$ given by %(we omit the loops $f(b):=g(b):=b$ and $f(c):=g(c):=c$ in the figure)
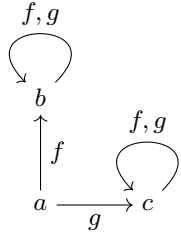
\begin{center}
\begin{tikzpicture} 
    \node (a)               {$a$};
    \node (b) [above=of a]  {$b$};
    \node (c) [right=of a]  {$c$};
    % siehe tikz 3.6.1a p. 165/166
    \draw[->] (a) to [edge label'={$f$}] (b);
    \draw[->] (a) to [edge label'={$g$}] (c);
    \draw[->] (b) to [edge label'={$f,g$}] [loop] (b);
    \draw[->] (c) to [edge label'={$f,g$}] [loop] (c);
\end{tikzpicture}
\end{center} we have
\begin{align*} 
    \mathfrak A_3\not\models a:b::a:c.
\end{align*} The intuitive reason is that $a:b::a:b$ is a more plausible proportion than $a:b::a:c$, which is reflected in the computation
\begin{align*} 
    \emptyset=Jus_{\mathfrak A_3}(a\to b\righttherefore a\to c)\subsetneq Jus_{\mathfrak A_3}(a\to b\righttherefore a\to b)=\{z\to f(z),\ldots\}.
\end{align*}
\end{example}

\begin{remark}\label{rem:in_the_same_way} It is important to emphasize that we interpret the expression `in the same way' as `maximally similar' instead of `identical' as the latter interpretation is too strict to be useful. To see why, consider, for instance, the algebra $(\mathbb N,+,\mathbb N)$ of natural numbers with addition where each number is a distinguished element, and consider the arrow equation in $(\mathbb N,+,\mathbb N)$ given by
\begin{align*} 
    2\to 4\righttherefore 3\to x.
\end{align*} We compute
\begin{align*} 
    Jus_{(\mathbb N,+,\mathbb N)}(2,4)=\{z\to z+z,z\to z+2,\ldots\}
\end{align*} and
\begin{align*} 
    Jus_{(\mathbb N,+,\mathbb N)}(3,5)=\{z\to z+2,\ldots\} \quad\text{and}\quad Jus_{(\mathbb N,+,\mathbb N)}(3,6)=\{z\to z+z,\ldots\},
\end{align*} where we have
\begin{align*} 
    z\to z+2\not\in Jus_{(\mathbb N,+,\mathbb N)}(3,6) \quad\text{and}\quad z\to z+z\not\in Jus_{(\mathbb N,+,\mathbb N)}(3,5).
\end{align*} Neither $Jus_{(\mathbb N,+,\mathbb N)}(3,5)$ nor $Jus_{(\mathbb N,+,\mathbb N)}(3,6)$ is thus {\em identical} to $Jus_{(\mathbb N,+,\mathbb N)}(2,4)$, which means that under a strict interpretation, neither $5$ nor $6$ would be a solution to the above equation. In our interpretation, on the other hand, both $5$ and $6$ are justifiable solutions according to \prettyref{def:models} as expected (cf. \prettyref{exa:2_4_3_z}).

% Now consider the following algebra $\mathfrak A$:

% Of course, the relationship between $a$ and $b$, and between $c$ and $d$ is {\em not identical} in the sense that the same kind of arrows connect $a,b$ and $c,d$; instead, it is {\em maximally similar} in the sense that there is no element $e\in A$ such that the relationship between $a$ and $b$, and between $c$ and $e$ is strictly more similar. Consider the slightly different algebra $\mathfrak B$:
% \begin{center}
% \begin{tikzpicture} 
% \node (a)               {$a$};
% \node (b) [above=of a]  {$b$};
% \node (c) [right=of a]  {$c$};
% \node (d) [above=of c]  {$d$};

% % siehe tikz 3.6.1a p. 165/166
% \draw[->] (a) to [edge label'={$f$}] (b); 
% \draw[->] (c) to [edge label'={$f$}] (b); 
% \end{tikzpicture}
% \end{center} The relation
% \begin{align*} \emptyset=Jus_\mathfrak B(a\to b\righttherefore c\to d)\subsetneq Jus_\mathfrak B(a\to b::c\to b)=\{z\to f(z)\}
% \end{align*} shows
% \begin{align*} \mathfrak B\not\models a:b::c:d. 
% \end{align*} The intuitive reason is that now the relation between $c$ and $b$ is more similar to the relation between $a$ and $b$ than the relation between $c$ and $d$; in fact, we have
% \begin{align*} \mathfrak B\models a:b::c:b. 
% \end{align*}
\end{remark}

\begin{remark} \prettyref{sec:Logical_} provides an alternative {\em logical} interpretation of analogical proportions in terms of model-theoretic types.
\end{remark}

\begin{definition} We call an $L$-term $s(\mathbf z)$ an {\em $\mathfrak A$-generalization} of an element $a$ in $\mathfrak A$ iff $a=s^\mathfrak A(\mathbf e)$, for some $\mathbf e\in A^{|\mathbf z|}$, and we denote the set of all $\mathfrak A$-generalizations of $a$ in $\mathfrak A$ by $gen_\mathfrak A(a)$. Moreover, we define for any elements $a\in A$ and $c\in B$:
\begin{align*} 
    gen_{(\mathfrak{A,B})}(a,c):=gen_\mathfrak A(a)\cap gen_\mathfrak B(c).
\end{align*}
\end{definition}

\begin{convention}\label{con:abcd} Notice that any justification $s(\mathbf z)\to t(\mathbf z)$ of $a\to b\righttherefore c\to d$ in $(\mathfrak{A,B})$ must satisfy
\begin{align}\label{equ:st} 
    a=s^\mathfrak A(\mathbf e_1) \quad\text{and}\quad b=t^\mathfrak A(\mathbf e_1) \quad\text{and}\quad c=s^\mathfrak B(\mathbf e_2) \quad\text{and}\quad d=t^\mathfrak B(\mathbf e_2),
\end{align} for some $\mathbf e_1\in A^{|\mathbf z|}$ and $\mathbf e_2\in B^{|\mathbf z|}$. In particular, this means
\begin{align*} 
    s\in gen_{(\mathfrak{A,B})}(a,c) \quad\text{and}\quad t\in gen_{(\mathfrak{A,B})}(b,d).
\end{align*} We sometimes write $s\xrightarrow{\mathbf e_1\to\mathbf e_2} t$ to make the {\em witnesses} $\mathbf e_1,\mathbf e_2$ and their transition explicit. This situation can be depicted as follows:
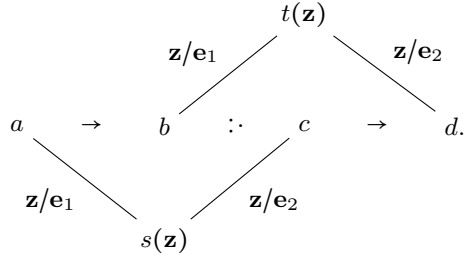
\begin{center}
\begin{tikzpicture}[node distance=1cm and 0.5cm]
% siehe tikz 3.6.1a §3.8
\node (a)               {$a$};
\node (d1) [right=of a] {$\to $};
\node (b) [right=of d1] {$b$};
\node (d2) [right=of b] {$\righttherefore$};
\node (c) [right=of d2] {$c$};
\node (d3) [right=of c] {$\to $};
\node (d) [right=of d3] {$d.$};
\node (s) [below=of b] {$s(\mathbf z)$};
\node (t) [above=of c] {$t(\mathbf z)$};

% siehe tikz 3.6.1a p. 165/166
\draw (a) to [edge label'={$\mathbf z/\mathbf e_1$}] (s); 
\draw (c) to [edge label={$\mathbf z/\mathbf e_2$}] (s);
\draw (b) to [edge label={$\mathbf z/\mathbf e_1$}] (t);
\draw (d) to [edge label'={$\mathbf z/\mathbf e_2$}] (t);
\end{tikzpicture}
\end{center} 

% Another way to think about analogical proportions is in terms of commutative diagrams:
% \begin{center}
% \begin{tikzcd}[row sep=2cm, column sep=2cm]
% a\ar[d,"s\rightarrow t"]\ar[r,"\mathbf e_1\to\mathbf e_2"] & c\ar[d,"s\rightarrow t"]\\
% b\ar[r,"\mathbf e_1\to\mathbf e_2"] & d
% \end{tikzcd}
% \end{center}
\end{convention}

The following characterization of solutions to analogical equations in terms of solutions to directed analogical equations is an immediate consequence of \prettyref{def:models}.

\begin{fact}\label{fact:Sol} For any $a,b\in A$ and $c,d\in B$, $d$ is a solution to $a:b\righttherefore c:x$ in $(\mathfrak{A,B})$ iff $d$ is a solution to $a\to b\righttherefore c\to x$ and $c$ is a solution to $b\to a\righttherefore d\to x$ in $(\mathfrak{A,B})$. This is equivalent to
\begin{align*} 
    Sol_{(\mathfrak{A,B})}(a:b\righttherefore c:x)=\left\{d\in Sol_{(\mathfrak{A,B})}(a\to b\righttherefore c\to x) \;\middle|\; c\in Sol_{(\mathfrak{A,B})}(b\to a\righttherefore d\to x)\right\}.
\end{align*} Similarly, $d$ is a solution to $a:b::c:x$ in $(\mathfrak{A,B})$ iff $d$ is a solution to $a:b\righttherefore c:x$ and $b$ is a solution to $c:d\righttherefore a:x$ in $(\mathfrak{A,B})$. This is equivalent to
\begin{align*} 
    Sol_{(\mathfrak{A,B})}(a:b::c:x)=\left\{d\in Sol_{(\mathfrak{A,B})}(a:b\righttherefore c:x) \;\middle|\; b\in Sol_{(\mathfrak{A,B})}(c:d\righttherefore a:x)\right\}.
\end{align*} We can visualize the derivation steps for proving $a:b::c:d$ as follows:
\begin{prooftree}
    \AxiomC{$a\to b\righttherefore c\to d$}
    \AxiomC{$b\to a\righttherefore d\to c$}
    \BinaryInfC{$a:b\righttherefore c:d$}
    \AxiomC{$c\to d\righttherefore a\to b$}
    \AxiomC{$d\to c\righttherefore b\to a$}
    \BinaryInfC{$c:d\righttherefore a:b$}
    \BinaryInfC{$a:b::c:d$}
\end{prooftree} This means that in order to prove $a:b::c:d$, we need to check the four relations in the first line.
\end{fact}

To guide the AI-practitioner, we shall now rewrite the above framework in a more algorithmic style.

\begin{pseudocode}\label{pseudo:Sol} First of all, one has to specify the $L$-algebras $\mathfrak A$ and $\mathfrak B$. Computing the solutions $S$ to an analogical equation $a:b::c:x$ in $(\mathfrak{A,B})$ consists of the following steps:
\begin{enumerate}
\item Compute $S_0:=Sol_{(\mathfrak{A,B})}(a\to b\righttherefore c\to x)$:
    \begin{enumerate}
    \item For each $d\in B$, if $Jus_\mathfrak A(a,b)\cup Jus_\mathfrak B(c,d)$ consists only of trivial justifications, then add $d$ to $S_0$.

    \item For each $L$-term $s(\mathbf z)\in gen_{(\mathfrak{A,B})}(a,c)$ and all witnesses $\mathbf e_1\in A^{|\mathbf z|},\mathbf e_2\in B^{|\mathbf z|}$ satisfying (cf. \prettyref{con:abcd})
    \begin{align*} 
        a=s^\mathfrak A(\mathbf e_1) \quad\text{and}\quad c=s^\mathfrak B(\mathbf e_2),
    \end{align*} and for each $L$-term $t(\mathbf z)\in gen_{\mathfrak{A}}(b)$ containing only variables occurring in $s(\mathbf z)$ and satisfying
    \begin{align*} 
        % b=t^\mathfrak A(\mathbf e_1) \quad\text{and}\quad d=t^\mathfrak B(\mathbf e_2),
        b=t^\mathfrak A(\mathbf e_1),
    \end{align*} add $s\to t$ to $Jus_{(\mathfrak{A,B})}(a\to b\righttherefore c\to t^\mathfrak B(\mathbf e_2))$.

    \item Identify those non-empty sets $Jus_{(\mathfrak{A,B})}(a\to b\righttherefore c\to d)$ which are subset maximal with respect to $d$ and add those $d$'s to $S_0$.
    \end{enumerate}
\item For each $d\in S_0$, check the following relations with the above procedure (cf. \prettyref{fact:Sol}):
    \begin{enumerate}
    \item $c\in Sol_{(\mathfrak{A,B})}(b\to a\righttherefore d\to x)$?
    \item $b\in Sol_{(\mathfrak{B,A})}(c\to d\righttherefore a\to x)$?
    \item $a\in Sol_{(\mathfrak{B,A})}(d\to c\righttherefore b\to x)$?
    \end{enumerate}
    Add those $d\in S_0$ to $S$ which pass all three tests. The set $S$ now contains all solutions to $a:b::c:x$ in $(\mathfrak{A,B})$.
\end{enumerate}
\end{pseudocode}

We now want to demonstrate analogical proportions with two illustrative examples over the natural numbers.

\begin{example}\label{exa:2_4_3_z} Consider the analogical equation
\begin{align}\label{equ:2_4_3_z} 
    2:4::3:x.
\end{align} According to \prettyref{def:models}, solving \prettyref{equ:2_4_3_z} requires three steps:
\begin{enumerate}

\item First, we need to solve the arrow equation
\begin{align*} 
    2\to 4\righttherefore 3\to x.
\end{align*} We can transform 2 into 4 in at least three different ways justified by $z\to 2+z$, $z\to 2z$, and $z\to z^2$. Here it is important to clarify the algebras involved. The first two justifications require addition, whereas the last justification requires multiplication. Moreover, the first justification additionally presupposes that 2 is a distinguished element---this is not the case for the last two justifications as $2z$ and $z^2$ are abbreviations for $z+z$ and $z\cdot z$, respectively, not involving 2. Analogously, transforming 3 `in the same way' as 2 can therefore mean at least three things: $3\to 2+3=5$, $3\to 3+3=6$, and $3\to 3^2=9$. More precisely, $z\to 2+z$ is a justification of $2\to 4\righttherefore 3\to d$ in $(\mathbb N,+,2)$ iff $d=5$ which shows that $Jus_{(\mathbb N,+,2)}(2\to 4\righttherefore 3\to 5)$ is a subset maximal set of justifications with respect to the last argument. This formally proves
\begin{align}\label{equ:N_+_2_models_2-4_3-5} 
    (\mathbb N,+,2)\models 2\to 4\righttherefore 3\to 5.
\end{align} The other two cases being analogous, we can further derive
\begin{align*} 
    (\mathbb N,+)\models 2\to 4\righttherefore 3\to 6 \quad\text{and}\quad (\mathbb N,\cdot)\models 2\to 4\righttherefore 3\to 9.
\end{align*} 

\item We now check whether these solutions are solutions to 
\begin{align*} 
    2:4\righttherefore 3:x
\end{align*} as follows. Let us start with \prettyref{equ:N_+_2_models_2-4_3-5}. We need to check that 3 is a solution to $4\to 2\righttherefore 5\to x$ in $(\mathbb N,+,2)$. The rewrite rule $2+z\to z$ is a justification of $4\to 2\righttherefore 5\to d$ in $(\mathbb N,+,2)$ iff $d=3$, which shows that 3 is indeed a solution. Hence, we have
\begin{align*} 
    (\mathbb N,+,2)\models 2:4\righttherefore 3:5.
\end{align*} The other cases being analogous, we can derive
\begin{align*} 
    (\mathbb N,+)\models 2:4\righttherefore 3:6 \quad\text{and}\quad (\mathbb N,\cdot)\models 2:4\righttherefore 3:9.
\end{align*}

\item Finally, we need to check
\begin{align*} 
    (\mathbb N,+,2)\models 3:5\righttherefore 2:4 \quad\text{and}\quad (\mathbb N,+)\models 3:6\righttherefore 2:4 \quad\text{and}\quad (\mathbb N,\cdot)\models 3:9\righttherefore 2:4.
\end{align*} This can be done by similar computations as above. Hence, we have
\begin{align*} 
    (\mathbb N,+,2)\models 2:4::3:5 \quad\text{and}\quad (\mathbb N,+)\models 2:4::3:6 \quad\text{and}\quad (\mathbb N,\cdot)\models 2:4::3:9.
\end{align*}
\end{enumerate}
\end{example}

\begin{example}\label{exa:20_4_30_x} The analogical equation $$20:4::30:x$$ has the solutions $x_1=6$ and $x_2=9$ in the multiplicative algebra $\mathfrak M:=(\mathbb N_2,\cdot,\mathbb N_2)$ as we show in \prettyref{exa:20_4_30_x_Appendix} (Appendix). The first solution, $x_1=6$, has an intuitive explanation as we obtain $4$ from $20$ by dividing by $5$---analogously, dividing $30$ by $5$ yields $6$. In the expanded algebra of rationals, this can be written as
\begin{align*} 
    20:\frac{20}5::30:\frac{30}5.
\end{align*} The second solution, $x_2=9$, is more subtle and can be roughly justified by the following reasoning (for the complete proof see \prettyref{exa:20_4_30_x_Appendix}):
\begin{align*} 
    (10\cdot 2):2^2::(10\cdot 3):3^2.
\end{align*} This last solution is less obvious than the first one and it therefore appears more interesting and more `creative'. Finally, we shall emphasize that a similar reasoning fails:
\begin{align*} 
    (10\cdot 2):2^2\not{::}(15\cdot 2):2^2.
\end{align*} The reason is that the arrow proportion $4\to 20\righttherefore 4\to 30$ has no justifications in $\mathfrak M$ (see \prettyref{ite:4-20-4-30} in \prettyref{exa:20_4_30_x_Appendix} for details), which indicates that computing (all) solutions to an analogical equation is more complicated than \prettyref{def:models} might suggest.
\end{example}

\section{Properties of Analogical Proportions}\label{sec:Properties_}

This section studies some basic mathematical properties of analogical equations and proportions.

\subsection{Characteristic Sets of Justifications}

Computing all justifications of an arrow proportion is difficult in general (see \prettyref{exa:20_4_30_x}), which fortunately can be omitted in many cases.

% \begin{definition} Given a justification $s(\mathbf z)\to t(\mathbf z)$, we define its {\em dual} by $t(\mathbf z)\to s(\mathbf z)$ if $t$ contains all variables in $\mathbf z$, extended to sets of justifications rule-wise.
% \end{definition}

\begin{definition}\label{def:J} We call a set $J$ of justifications a {\em characteristic set of justifications} of $a\to b\righttherefore c\to d$ in $(\mathfrak{A,B})$ iff $J$ is a sufficient set of justifications of $a\to b\righttherefore c\to d$ in $(\mathfrak{A,B})$, that is, iff
\begin{enumerate}
\item $J\subseteq Jus_{(\mathfrak{A,B})}(a\to b\righttherefore c\to d)$, and
\item $J\subseteq Jus_{(\mathfrak{A,B})}(a\to b\righttherefore c\to d')$ implies $d'=d$, for each $d'\in B$.
\end{enumerate} In case $J=\{s\to t\}$ is a singleton set satisfying both conditions, we call $s\to t$ a {\em characteristic justification} of $a\to b\righttherefore c\to d$ in $(\mathfrak{A,B})$. 
\end{definition}

\begin{notation}\label{not:J} In case $J$ is a characteristic set of justifications, we will occasionally write
\begin{align*} 
    (\mathfrak{A,B})\models_J a\to b\righttherefore c\to d
\end{align*} to make the set of characteristic justifications $J$ explicit.
\end{notation}

\begin{example} In \prettyref{exa:20_4_30_x_Appendix} (Appendix) we argue that $\{z_1z_2\to z_1^2,\;z_1z_2\to 2z_1\}$ is a characteristic set of justifications of $20\to 4\righttherefore 30\to 4$ in $\mathfrak M$.
\end{example}

% \begin{remark}\label{rem:trivial} Notice that the empty set is always a trivial set of justifications. In some cases, given an arrow equation $a\to b\righttherefore c\to x$, the set $Jus_{(\mathfrak{A,B})}(a\to b\righttherefore c\to d)$ of justifications of $a\to b\righttherefore c\to d$ in $(\mathfrak{A,B})$ is empty, for {\em any} $d\in B$, in which case we trivially have $(\mathfrak{A,B})\models a\to b\righttherefore c\to d$. This is, for example, the case in any structure $(A)$, consisting only of a universe $A$ without any functions on $A$---given {\em distinct} elements $a,b,c,d\in A$, we always have $Jus_{(A)}(a\to b\righttherefore c\to d)=\emptyset$ and hence $a\to b\righttherefore c\to d$ is a trivial directed proportion in $(A)$.
% \end{remark}

The following lemma provides a sufficient condition of characteristic justifications in terms of mild injectivity.

\begin{lemma}[Uniqueness Lemma]\label{lem:UL} Let $s(\mathbf z)\to t(\mathbf z)$ be a non-trivial justification of $a\to b\righttherefore c\to d$ in $(\mathfrak{A,B})$.% such that $s$ and $t$ contain the same variables.
\begin{enumerate}
\item If there is a unique $\mathbf e\in B^{|\mathbf z|}$ such that $c=s^\mathfrak B(\mathbf e)$, then $s\to t$ is a characteristic justification of $a\to b\righttherefore c\to d$ in $(\mathfrak{A,B})$.

\item Consequently, if there are unique $\mathbf e_1,\mathbf e_2\in B^{|\mathbf z|}$ satisfying $c=s^\mathfrak B(\mathbf e_1)$ and $d=t^\mathfrak B(\mathbf e_2)$, and every variable in $s$ occurs in $t$, then $(\mathfrak{A,B})\models a:b\righttherefore c:d$.

\item Moreover, if there are unique $\mathbf e_1,\mathbf e_2\in A^{|\mathbf z|}$ and unique $\mathbf e_3,\mathbf e_4\in B^{|\mathbf z|}$ such that
\begin{align*} 
    a=s^\mathfrak A(\mathbf e_1) \quad\text{and}\quad b=t^\mathfrak A(\mathbf e_2) \quad\text{and}\quad c=s^\mathfrak B(\mathbf e_3) \quad\text{and}\quad d=t^\mathfrak B(\mathbf e_4),
\end{align*} and every variable in $s$ occurs in $t$, then $(\mathfrak{A,B})\models a:b::c:d$.

\item Hence, in case $s$ and $t$ are injective in $\mathfrak A$ and $\mathfrak B$ and contain the same variables, then $(\mathfrak{A,B})\models a:b::c:d$.
\end{enumerate}
\end{lemma}
\begin{proof} We prove each item separately:
\begin{enumerate}
\item\label{ite:inj_1} Since $s(\mathbf z)\to t(\mathbf z)$ is a justification of $a\to b\righttherefore c\to d$ in $(\mathfrak{A,B})$ by assumption, there are sequences of elements $\mathbf e_1\in A^{|\mathbf z|}$ and $\mathbf e_2\in B^{|\mathbf z|}$ satisfying \prettyref{equ:st}, and $\mathbf e_2$ is uniquely determined by assumption. Consequently, given any element $d'\in B$, $s\to t$ is a justification of $a\to b\righttherefore c\to d'$ in $(\mathfrak{A,B})$ iff $d'=t^\mathfrak B(\mathbf e_2)=d$, which shows that $s\to t$ is indeed a characteristic justification.

\item\label{ite:inj_2} Since $s\to t$ is a justification of $a\to b\righttherefore c\to d$ in $(\mathfrak{A,B})$ and there are unique $\mathbf e_1,\mathbf e_2\in B^{|\mathbf z|}$ such that $c=s^\mathfrak B(\mathbf e_1)$ and $d=t^\mathfrak B(\mathbf e_2)$ by assumption, $s\to t$ and $t\to s$ (recall that $s$ and $t$ contain the same variables by assumption) are characteristic justifications  of $a\to b\righttherefore c\to d$ and $b\to a\righttherefore d\to c$ in $(\mathfrak{A,B})$, respectively, by the argument in \prettyref{ite:inj_1}.

\item\label{ite:inj_3} Analogous to \prettyref{ite:inj_2}.
\item Direct consequence of \prettyref{ite:inj_3}.
\end{enumerate}
\end{proof}

% \subsection{Trivial Proportion Theorem}

% beachte, dass unter neuer def von jus es sein kann, dass $t\to s$ keine jus ist, da $s$ var enthaelt die nicht in $t$ vorkommen
% Define the {\em dual} of the rewrite rule $s\to t$ by
% \begin{align*} (s\to t)^d:=t\to s.
% \end{align*} By \prettyref{def:Jus}, we have
% \begin{align}\label{equ:Jus^d} Jus_\mathfrak A(a,b)=Jus_\mathfrak A(b,a)^d\quad\text{for any $a,b\in A$}.
% \end{align} This yields the following implication.

% \begin{theorem}[Trivial Proportion Theorem]\label{thm:Trivial_Proportion_Theorem} For any $a,b\in A$ and $c,d\in B$,\footnote{Recall from \prettyref{con:trivial} that $Jus_\mathfrak A(a,b)=\emptyset$ actually means that $Jus_\mathfrak A(a,b)$ consists only of trivial justifications.}
% \begin{align*} Jus_\mathfrak A(a,b)\cup Jus_\mathfrak B(c,d)=\emptyset \quad\Rightarrow\quad (\mathfrak{A,B})\models a:b::c:d.
% \end{align*}
% \end{theorem}
% \begin{proof} From
% \begin{align*} 
%     Jus_\mathfrak A(a,b)\cup Jus_\mathfrak B(c,d)=\emptyset \quad\Leftrightarrow\quad ...
% \end{align*} we deduce
% \begin{align*} Jus_\mathfrak A(a,b)\cup Jus_\mathfrak B(c,d)=\emptyset \quad\Leftrightarrow\quad Jus_\mathfrak A(b,a)\cup Jus_\mathfrak B(d,c)=\emptyset,
% \end{align*} which immediately implies $\mathfrak A\models a:b::c:d$ via \prettyref{fact:Sol}.
% \end{proof}

\subsection{Functional Proportion Theorem}\label{sec:Functional_Proportion_Theorem}

The following reasoning pattern---which roughly says that {\em functional dependencies} are preserved across (different) domains---will often be used in the rest of the paper.

\begin{theorem}[Functional Proportion Theorem]\label{thm:FPT} Let $t(z)$ be an $L$-term.
\begin{enumerate}
\item Given $a\in A$ and $c\in B$, the rewrite rule $z\to t(z)$ characteristically justifies
\begin{align}\label{equ:AB_models_a->t(a)>c->t(c)} 
    (\mathfrak{A,B})\models a\to t^\mathfrak A(a)\righttherefore c\to t^\mathfrak B(c).
\end{align}
\item Given $a\in A$ and $c\in B$, if $e=c\in B$ is the unique element satisfying $t^\mathfrak B(e)=t^\mathfrak B(c)$, then $(\mathfrak{A,B})\models t^\mathfrak A(a)\to a\righttherefore t^\mathfrak B(c)\to c$, which together with \prettyref{equ:AB_models_a->t(a)>c->t(c)} implies
\begin{align*} 
    (\mathfrak{A,B})\models a:t^\mathfrak A(a)\righttherefore c:t^\mathfrak B(c).
\end{align*}
\item Consequently, if $e_1=a\in A$ and $e_2=c\in B$ are unique elements satisfying $t^\mathfrak A(e_1)=t^\mathfrak A(a)$ and $t^\mathfrak B(e_2)=t^\mathfrak B(c)$, then 
\begin{align*} 
    (\mathfrak{A,B})\models a:t^\mathfrak A(a)::c:t^\mathfrak B(c).
\end{align*} In this case, we call $t^\mathfrak B(c)$ a {\em functional solution} of $a:b::c:x$ in $(\mathfrak{A,B})$ {\em characteristically justified by $z\to t(z)$}.
\end{enumerate} Hence, if $t$ is injective in $\mathfrak A$ and $\mathfrak B$, then $(\mathfrak{A,B})\models a:t^\mathfrak A(a)::c:t^\mathfrak B(c)$.
\end{theorem}
\begin{proof} The rewrite rule $z\to t(z)$ is a characteristic justification of $a\to t^\mathfrak A(a)\righttherefore c\to t^\mathfrak B(c)$ in $(\mathfrak{A,B})$ by the Uniqueness \prettyref{lem:UL} as $z$ is injective in $\mathfrak B$. The other cases follow by analogous arguments from the Uniqueness \prettyref{lem:UL}.
\end{proof}

% \prettyref{thm:FPT} motivates the following definition.

% \begin{definition} We call $d\in B$ a {\em functional solution} to $a\to b\righttherefore c\to x$ in $(\mathfrak{A,B})$ iff there is an $L$-term $t(z)$ such that $z\to t(z)$ is a characteristic justification of $a\to b\righttherefore c\to d$ in $(\mathfrak{A,B})$. Moreover, we call $d$ a {\em functional solution} to $a:b::c:x$ in $(\mathfrak{A,B})$ iff $d$ is a solution to $a:b::c:x$ and a functional solution to $a\to b\righttherefore c\to x$ in $(\mathfrak{A,B})$.
% \end{definition}

Functional solutions are plausible since transforming $a$ into $t(a)$ and $c$ into $t(c)$ is a direct implementation of `transforming $a$ and $c$ in the same way', and it is therefore surprising that functional solutions can be nonetheless unexpected and therefore `creative' as will be demonstrated, for instance, in \prettyref{exa:0_0_2_40}.

\begin{remark}\label{rem:t} An interesting consequence of \prettyref{thm:FPT} is that in case $t$ is a ground term, we have the directed analogical proportion
\begin{align}\label{equ:abcb} 
    (\mathfrak{A,B})\models a\to t^\mathfrak A\righttherefore c\to t^\mathfrak B,\quad\text{for {\em all} $a\in A$ and $c\in B$},
\end{align} characteristically justified by \prettyref{thm:FPT} via $z\to t$. This can be intuitively interpreted as follows: every ground term is a `name' in our language for a specific element of the algebra, which means that it is in a sense a `known' element. As the framework is designed to compute `novel' or `unknown' elements in the target domain via analogy-making, \prettyref{equ:abcb} means that `known' target elements can always be computed.
\end{remark}

The following example shows how we can solve analogical equations across different domains within our framework.

\begin{example}\label{exa:2_4_ab_z} We want to formally solve the analogical equation of \prettyref{exa:exa} given by
\begin{align}\label{equ:2_4_ab_z} 
    2:4::ab:x.
\end{align} For this, we first need to specify the algebras involved. Let $L$ be the language consisting of a single binary function symbol $\cdot$, and let $(\mathbb N,\cdot)$ and $(\Sigma^+,\cdot)$, where $\Sigma:=\{a,b\}$, be $L$-algebras. This means we interpret $\cdot$ as multiplication of numbers in $\mathbb N$ and as concatenation of words in $\Sigma^+$. As a direct consequence of \prettyref{thm:FPT} with $t(z):=z\cdot z$, injective in $(\mathbb N,\cdot)$ and $(\Sigma^+,\cdot)$, we can formally derive the solution $abab$ to \prettyref{equ:2_4_ab_z}:
\begin{align*} 
    ((\mathbb N,\cdot),(\Sigma^\ast,\cdot))\models 2:4::ab:abab.
\end{align*}
\end{example}

% \begin{example} Let $L$ be the language consisting of a binary function symbol $\circ$ and a constant symbol $c$. Consider the $L$-algebras
% \begin{align*} \mathfrak A:=(\mathfrak P(\{a,b\}),\cup,\{b\}) \quad\text{and}\quad \mathfrak B:=(\mathbb N,+,1).
% \end{align*} That is, we interpret $\circ$ as union in $\mathfrak A$ and as addition in $\mathfrak B$, and we interpret the constant symbol $c$ as the set $\{b\}$ in $\mathfrak A$ and as the number 1 in $\mathfrak B$. Now consider the analogical equation
% \begin{align*} \{a\}:\{a,b\}::1:x.
% \end{align*} Define the $L$-term
% \begin{align*} t(z):=z\circ c.
% \end{align*} The interpretations of $t$ in $\mathfrak A$ and $\mathfrak B$ are then given by
% \begin{align*} t^\mathfrak A(z):=z\cup\{b\} \quad\text{and}\quad t^\mathfrak B(z):=z+1.
% \end{align*} As a consequence of \prettyref{thm:FPT}, we have
% \begin{align*} (\mathfrak{A,B})\models\{a\}\to\{a,b\}\righttherefore 1\to 2.
% \end{align*}
% \end{example}

% todo fuege referenzen ein zu paper, wo axiome bereits vorkommen
\subsection{Axioms}\label{sec:Axioms}

\citeA{Lepage03} \cite<cf.>[pp. 796-797]{Miclet08} introduces the following axioms in the linguistic context as a guideline for formal models of analogical proportions (over a single universe), adapted here to our framework formulated above:\footnote{\citeA{Lepage03} formulates his axioms to hold in a single domain without any reference to an underlying structure $\mathfrak A$.}
\begin{align}
    (\mathfrak{A,B})&\models a:b::c:d\quad\Leftrightarrow\quad (\mathfrak{B,A})\models c:d::a:b\quad\text{(symmetry)},\\
    \label{equ:central_permutation} \mathfrak A&\models a:b::c:d\quad\Leftrightarrow\quad\mathfrak A\models a:c::b:d\quad\text{(central permutation)},\\
    \label{equ:strong_inner_reflexivity} \mathfrak A&\models a:a::c:d\quad\Rightarrow\quad d=c\quad\text{(strong inner reflexivity)},\\
    \label{equ:strong_reflexivity} \mathfrak A&\models a:b::a:d\quad\Rightarrow\quad d=b\quad\text{(strong reflexivity)}.
\end{align} %Of course, in our framework, \prettyref{equ:central_permutation} makes sense only in case $b,c\in A\cap B$, and \prettyref{equ:strong_inner_reflexivity} only in case $a,b\in A\cap B$.

% Notice that within a single domain, we can interpret $::$ as a binary relation $::\subseteq\mathfrak A^2\times\mathfrak A^2$ between pairs of elements of $\mathfrak A$.

Although Lepage's axioms appear reasonable in the word domain (but see \prettyref{problem:Lepage}), they cannot be straightforwardly applied to the general case. %for example, consider the domain $\mathbb Z$ of integers ordered by $<$ without any other algebraic structure on $\mathbb Z$---here $0<1$ and $-1<2$ reasonably imply $0:1::-1:2$; however, one can argue that the proportion $0:-1::1:2$, obtained by central permutation, fails since $0>-1$ whereas $1<2$.\footnote{We shall emphasize here that this example cannot be formalized within our framework as relation symbols are not part of functional structures.}% (but see the last paragraph of the discussion on future work in \prettyref{sec:Future_Work}).} % Symmetry, on the other hand, fails in `asymmetric' structures as in \prettyref{fig:asymmetric} below (and see \prettyref{rem:asymmetric}). 
Strong inner reflexivity fails, for instance, if the relation of $a$ to itself is similar to the relation between $c$ and $d$. Strong reflexivity fails, for example, in algebras where the relation of $a$ to $b$ and $d$ is identical, for distinct elements $b$ and $d$. In our framework, by making the underlying structures of an analogical proportion explicit, it turns out that except for symmetry none of Lepage's axioms holds in the general case, justified by counter-examples (\prettyref{thm:axioms}). This has critical consequences, as his axioms are assumed by many authors \cite<e.g.>{Barbot19,Miclet08} to hold beyond the word domain.

We replace Lepage's above list by the following list of axioms:
\begin{align}
    \label{equ:symmetry} (\mathfrak{A,B})&\models a:b::c:d\quad\Leftrightarrow\quad (\mathfrak{B,A})\models c:d::a:b\quad\text{(symmetry)},\\
    \label{equ:inner_symmetry} (\mathfrak{A,B})&\models a:b::c:d \quad\Leftrightarrow\quad (\mathfrak{A,B})\models b:a::d:c\quad\text{(inner symmetry)},\\
    \label{equ:inner_reflexivity} (\mathfrak{A,B})&\models a:a::c:c\quad\text{(inner reflexivity)},\\
    % someday/maybe kann man reflexivity von $\mathfrak A$ auf $(\mathfrak{A,B})$ verallgemeinern?
    \label{equ:reflexivity} \mathfrak A&\models a:b::a:b\quad\text{(reflexivity)},\\
    \label{equ:determinism} \mathfrak A&\models a:a::a:d \quad\Leftrightarrow\quad d=a\quad\text{(determinism)}.
\end{align}

Moreover, we consider the following property, for $a,b$ contained in $\mathfrak A$ and in $\mathfrak B$:
\begin{align}\label{equ:commutativity} 
    (\mathfrak{A,B})\models a:b::b:a\quad\text{(commutativity).}
\end{align}

Furthermore, we consider the following properties, for $L$-algebras $\mathfrak{A,B,C}$ and elements $a,b\in A$, $c,d\in B$, $e,f\in C$:
\begin{prooftree}\label{equ:transitivity}
    \AxiomC{$(\mathfrak{A,B})\models a:b::c:d$}
        \AxiomC{$(\mathfrak{B,C})\models c:d::e:f$}
        \RightLabel{(transitivity),}
    \BinaryInfC{$(\mathfrak{A,C})\models a:b::e:f$}
\end{prooftree} and, for elements $a,b,e\in A$ and $c,d,f\in B$, the property
\begin{prooftree}\label{equ:inner_transitivity}
    \AxiomC{$(\mathfrak{A,B})\models a:b::c:d$}
        \AxiomC{$(\mathfrak{A,B})\models b:e::d:f$}
        \RightLabel{(inner transitivity),}
    \BinaryInfC{$(\mathfrak{A,B})\models a:e::c:f$}
\end{prooftree} and, for elements $a\in A$, $b$ in $\mathfrak A$ and $\mathfrak B$, $c$ in $\mathfrak B$ and $\mathfrak C$, and $d\in C$, the property
\begin{prooftree}\label{equ:central_transitivity}
    \AxiomC{$(\mathfrak{A,B})\models a:b::b:c$}
        \AxiomC{$(\mathfrak{B,C})\models b:c::c:d$}
        \RightLabel{(central transitivity).}
    \BinaryInfC{$(\mathfrak{A,C})\models a:b::c:d$}
\end{prooftree} Notice that central transitivity follows from transitivity.

Finally, we consider the following schema, where $\mathfrak A'$ and $\mathfrak B'$ are $L'$-algebras, for some language $L\subseteq L'$:
\begin{prooftree}\label{equ:monotonicity}
    \AxiomC{$(\mathfrak{A,B})\models a:b::c:d$}
        \AxiomC{$\mathfrak A=\mathfrak A'\upharpoonright L$}
            \AxiomC{$\mathfrak B=\mathfrak B'\upharpoonright L$}
            \RightLabel{(monotonicity).}
    \TrinaryInfC{$(\mathfrak{A',B'})\models a:b::c:d$}
\end{prooftree}
% \begin{align}\label{equ:monotonicity} \mathfrak A\models a:b::c:d \quad\text{and}\quad \mathfrak A=\mathfrak A'\upharpoonright L \quad\Rightarrow\quad \mathfrak A'\models a:b::c:d\quad\text{(monotonicity)}.
% \end{align}

\begin{remark}\label{rem:Lepage03} Notice that strong inner reflexivity is a conditional statement, whereas inner reflexivity is an assertion not implied by strong inner reflexivity. The same applies to strong reflexivity and reflexivity. Importantly, central permutation and strong reflexivity imply strong inner reflexivity, and central permutation together with strong inner reflexivity imply strong reflexivity.
\end{remark}

We have the following analysis of the above axioms within our framework.

\begin{theorem}\label{thm:axioms} The analogical proportion relation, as defined in \prettyref{def:models}, satisfies
\begin{itemize}
    \item symmetry \prettyref{equ:symmetry},
    \item inner symmetry \prettyref{equ:inner_symmetry},
    \item inner reflexivity \prettyref{equ:inner_reflexivity},
    \item reflexivity \prettyref{equ:reflexivity},
    \item determinism \prettyref{equ:determinism},
\end{itemize}  and, in general, it does not satisfy
\begin{itemize}
    \item central permutation \prettyref{equ:central_permutation},
    \item strong inner reflexivity \prettyref{equ:strong_inner_reflexivity},
    \item strong reflexivity \prettyref{equ:strong_reflexivity},
    \item commutativity \prettyref{equ:commutativity},
    \item transitivity,
    \item inner transitivity,
    \item central transitivity,
    \item monotonicity.
\end{itemize}
\end{theorem}
\begin{proof} We have the following proofs:
\begin{itemize}
    \item Symmetry \prettyref{equ:symmetry} and inner symmetry \prettyref{equ:inner_symmetry} are immediate consequences of \prettyref{def:models} as the framework is designed to satisfy these axioms.

    \item Inner reflexivity \prettyref{equ:inner_reflexivity} is an immediate consequence of \prettyref{thm:FPT} with $t(z):=z$, injective in $\mathfrak A$ and $\mathfrak B$.

    % todo gilt das auch fuer $(\mathb{A,B})$ statt $\mathfrak A$?
    \item Next, we prove reflexivity \prettyref{equ:reflexivity}. If $Jus_\mathfrak A(a,b)\cup Jus_\mathfrak A(a,b)=Jus_\mathfrak A(a,b)$ consists only of trivial justifications, we are done. Otherwise, there is at least one non-trivial justification in $Jus_\mathfrak A(a,b)=Jus_\mathfrak A(a\to b\righttherefore a\to b)$. We proceed by showing that $Jus_\mathfrak A(a\to b\righttherefore a\to b)$ is subset maximal with respect to the last $b$. For any $d\in A$, we have
    \begin{align*} 
        Jus_\mathfrak A(a\to b\righttherefore a\to d)&=Jus_\mathfrak A(a,b)\cap Jus_\mathfrak A(a,d)\\
        &\subseteq Jus_\mathfrak A(a,b)=Jus_\mathfrak A(a\to b\righttherefore a\to b),
    \end{align*} which shows that $Jus_\mathfrak A(a\to b\righttherefore a\to b)$ is indeed maximal. Hence, $\mathfrak A\models a\to b\righttherefore a\to b$. The same line of reasoning proves $\mathfrak A\models b\to a\righttherefore b\to a$. Hence, $\mathfrak A\models a:b\righttherefore a:b$, which finally proves $\mathfrak A\models a:b::a:b$.

    % todo gilt das auch fuer $(\mathb{A,B})$ statt $\mathfrak A$?
    \item Next, we prove determinism \prettyref{equ:determinism}. ($\Leftarrow$) Inner reflexivity \prettyref{equ:inner_reflexivity} implies
    \begin{align*} 
        \mathfrak A\models a:a::a:a.
    \end{align*} $(\Rightarrow)$ We assume $\mathfrak A\models a:a::a:d$. Since $z\to z\in Jus_\mathfrak A(a,a)$, the set $Jus_\mathfrak A(a,a)\cup Jus_\mathfrak A(a,d)$ cannot consist only of trivial justifications. Every justification $s\xrightarrow{\mathbf e_1\to\mathbf e_2}t$ of $a\to a\righttherefore a\to d$ can be transformed into a justification $s\xrightarrow{\mathbf e_1\to\mathbf e_1}t$ (replace $\mathbf e_2$ by $\mathbf e_1$) of $a\to a\righttherefore a\to a$. On the other hand, we have
    \begin{align*} 
        z\to z\in Jus_\mathfrak A(a\to a\righttherefore a\to a)
    \end{align*} whereas
    \begin{align*} 
        z\to z\not\in Jus_\mathfrak A(a\to a\righttherefore a\to d),\quad\text{for all $d\neq a$}.
    \end{align*} This shows
    \begin{align*} 
        Jus_\mathfrak A(a\to a\righttherefore a\to d)\subsetneq Jus_\mathfrak A(a\to a\righttherefore a\to a),
    \end{align*} which implies
    \begin{align*} 
        \mathfrak A\not\models a:a::a:d,\quad\text{for all $d\neq a$.}
    \end{align*}

    \item Next, we disprove central permutation \prettyref{equ:central_permutation}. For this consider the algebra $\mathfrak A:=(\{a,b,c,d\},f)$, given by (we omit the loops $f(o):=o$, for $o\in\{b,c,d\}$, in the figure)
    \begin{center}
        \begin{tikzpicture} 
        \node (a)               {$a$};
        \node (b) [above=of a]  {$b$};
        \node (c) [right=of a]  {$c$};
        \node (d) [right=of b]  {$d$};
        \draw[->] (a) to [edge label'={$f$}] (c);
    \end{tikzpicture}
    \end{center} We have $\mathfrak A\models a:b::c:d$, whereas $\mathfrak A\not\models a:c::b:d$ as there is an arrow from $a$ to $c$ but not from $b$ to $d$ in $\mathfrak A$. 

    A second proof follows from the forthcoming \prettyref{thm:(A)} (depending only on inner reflexivity already shown above), which yields
    \begin{align*} 
        (\{a,b,c\})\models a:b::a:c \quad\text{whereas}\quad (\{a,b,c\})\not\models a:a::b:c.
    \end{align*}

    \item Next, we disprove strong inner reflexivity \prettyref{equ:strong_inner_reflexivity}. For this consider the algebra $\mathfrak A:=(\{a,c,d\},f)$ given by
    
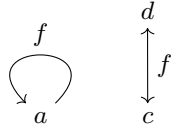
\begin{center}
    \begin{tikzpicture} 
        \node (a)               {$a$};
        \node (c) [right=of a]  {$c$};
        \node (d) [above=of c]  {$d$};
        \draw[->] (a) to [edge label'={$f$}] [loop] (a);
        \draw[<->] (c) to [edge label'={$f$}] (d);
    \end{tikzpicture}
    \end{center} As $f$ is injective in $\mathfrak A$, applying the Functional Proportion \prettyref{thm:FPT} yields $\mathfrak A\models a:f(a)::c:f(c)$ which is equivalent to $\mathfrak A\models a:a::c:d$.

    \item Next, we disprove strong reflexivity \prettyref{equ:strong_reflexivity}. By the forthcoming \prettyref{thm:(A)} (which depends only on inner reflexivity already proved above), we have
    \begin{align*} 
        (\{a,b,d\})\models a:b::a:d.
    \end{align*}

    \item Commutativity fails in the algebra $\mathfrak A:=(\{a,b\},f)$ given by
    
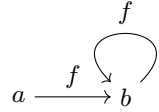
\begin{center}
    \begin{tikzpicture} 
        \node (a)               {$a$};
        \node (b) [right=of a]  {$b$};
        \draw[->] (a) to [edge label={$f$}] (b);
        \draw[->] (b) to [edge label'={$f$}] [loop] (b);
    \end{tikzpicture}
    \end{center} This follows from
    \begin{align*}
        Jus_\mathfrak A(a,b)\cup Jus_\mathfrak A(b,a)=\{z\to f(z),\ldots\}\neq\emptyset
    \end{align*} whereas
    \begin{align*} 
        Jus_\mathfrak A(a\to b\righttherefore b\to a)=\emptyset.
    \end{align*}

    \item Transitivity fails in the algebra $\mathfrak A:=(\{a,b,c,d,e,f\},g,h)$ given by (we omit the loops $g(o):=o$ for $o\in\{b,d,e,f\}$, and $h(o):=o$ for $o\in\{a,b,d,f\}$, in the figure)
    \begin{center}
    \begin{tikzpicture} 
        \node (a)               {$a$};
        \node (b) [right=of a]  {$b$};
        \node (c) [right=of b, xshift=0.8cm]  {$c$};
        \node (d) [right=of c]  {$d$};
        \node (e) [right=of d, xshift=0.8cm]  {$e$};
        \node (f) [right=of e]  {$f$};
        \draw[->] (a) to [edge label={$g$}] (b);
        % \draw[->] (b) to [edge label'={$g$}] [loop] (b);
        \draw[->] (c) to [edge label={$g,h$}] (d);
        % \draw[->] (d) to [edge label'={$g,h$}] [loop] (d);
        \draw[->] (e) to [edge label={$h$}] (f);
        % \draw[->] (e) to [edge label'={$g$}] [loop] (e);
        % \draw[->] (f) to [edge label'={$g,h$}] [loop] (f);
    \end{tikzpicture}
    \end{center} The arrow proportions $a\to g(a)\righttherefore c\to g(c)=a\to b\righttherefore c\to d$ and $c\to d\righttherefore a\to b$ are immediate consequences of the Functional Proportion \prettyref{thm:FPT}, and the arrow proportions $b\to a\righttherefore d\to c$ and $d\to c\righttherefore b\to a$ follow from the fact that
    \begin{align*} 
        Jus_\mathfrak A(b\to a\righttherefore d\to c)=\left\{g^\ell(z)\to z \;\middle|\; \ell\geq 0\right\}=Jus_\mathfrak A(d\to c\righttherefore b\to a)
    \end{align*} are non-empty and maximal. This shows $\mathfrak A\models a:b::c:d$. An analogous argument shows $\mathfrak A\models c:d::e:f$. On the other hand,
    \begin{align*} 
        Jus_\mathfrak A(a,b)\cup Jus_\mathfrak A(e,f)\neq\emptyset \quad\text{whereas}\quad Jus_\mathfrak A(a\to b\righttherefore e\to f)=\emptyset
    \end{align*} shows $\mathfrak A\not\models a:b::e:f$.

    \item Inner transitivity fails in the algebra $\mathfrak A:=(\{a,b,c,d,e,f\},g)$ given by (we omit the loops $g(o):=o$, for $o\in\{b,e,c,d,f\}$, in the figure)
    
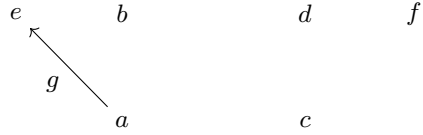
\begin{center}
    \begin{tikzpicture} 
        \node (a) {$a$};
        \node (b) [above=of a] {$b$};
        \node (e) [left=of b] {$e$};
        \node (c) [right=of a,xshift=1cm] {$c$};
        \node (d) [above=of c] {$d$};
        \node (f) [right=of d] {$f$};
        \draw[->] (a) to [edge label={$g$}] (e);
    \end{tikzpicture}
    \end{center} We have $\mathfrak A\models a:b::c:d$ and $\mathfrak A\models b:e::d:f$ by an argument analogous to the proof of \prettyref{thm:(A)}. On the other hand, 
    \begin{align*} 
        Jus_\mathfrak A(a,e)\cup Jus_\mathfrak A(c,f)\neq\emptyset \quad\text{whereas}\quad Jus_\mathfrak A(a\to e\righttherefore c\to f)=\emptyset
    \end{align*} shows $\mathfrak A\not\models a:e::c:f$.

    \item Central transitivity fails in the algebra $\mathfrak A:=(\{a,b,c,d\},g,h)$ given by (we omit the loops $g(o):=o$ for $o\in\{c,d\}$, and $h(o):=o$ for $o\in\{a,d\}$, in the figure)
    \begin{center}
    \begin{tikzpicture} 
        \node (a)               {$a$};
        \node (b) [right=of a]  {$b$};
        \node (c) [right=of b]  {$c$};
        \node (d) [right=of c]  {$d$};
        \draw[->] (a) to [edge label={$g$}] (b);
        \draw[->] (b) to [edge label={$g,h$}] (c);
        \draw[->] (c) to [edge label={$h$}] (d);
    \end{tikzpicture}
    \end{center} The proof is analogous to the above disproof of transitivity.

    % todo expansion by constants only! redefine `expansion' in preliminaries!
    \item Finally, we disprove monotonicity. For this, consider the algebra $\mathfrak A:=(\{a,b,c,d\})$, consisting only of its universe. We have (see the forthcoming \prettyref{thm:(A)} which depends only on inner reflexivity already proved above)
    \begin{align*} 
        \mathfrak A\models a:b::c:d.
    \end{align*} Now consider the expansion $\mathfrak A':=(\{a,b,c,d\},f)$ of $\mathfrak A$ given by (we omit the loops $f(o):=o$, for $o\in\{b,c,d\}$, in the figure)
    
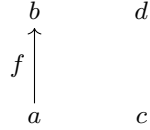
\begin{center}
    \begin{tikzpicture} 
        \node (a)               {$a$};
        \node (b) [above=of a]  {$b$};
        \node (c) [right=of a]  {$c$};
        \node (d) [right=of b]  {$d$};
        \draw[->] (a) to [edge label={$f$}] (b);
    \end{tikzpicture}
    \end{center} We have
    \begin{align*} 
        \mathfrak A'\not\models a:b::c:d.
    \end{align*}
\end{itemize}
\end{proof}

% \begin{remark}\label{rem:asymmetric} The outer asymmetry of our framework proved above is a direct consequence of the asymmetric formulation of the problem of solving an analogical equation. More precisely, given an analogical equation of the form $a:b::c:x$, the elements $a,b,c$ are {\em fixed} and one asks for an element $d$ which is to $c$ what the fixed $b$ is to $a$. One can imagine a different interpretation of an analogical proportion in which the condition for the set of justifications of $a:b::c:d$ to be subset maximal with respect to both $d$ {\em and} $b$. Under that interpretation the symmetry axiom is satisfied. However, this has the undesired consequence that reflexivity \prettyref{equ:reflexivity} fails in general since, for example, in the structure $(A)$, consisting only of (at least) two distinct elements $a,b$, we have
% \begin{align*} \emptyset=Jus_{(A)}(a\to b\righttherefore a\to b)\subsetneq Jus_{(A)}(a\to a\righttherefore a\to a)=\{z\to z\},
% \end{align*} which means that $a\to b\righttherefore a\to b$ is not a directed analogical proportion in $(A)$. In this paper, we choose to interpret analogical proportions always as solutions to analogical equations, which are asymmetric in nature as proved above.
% \end{remark}

\begin{remark} Notice that Lepage's axioms (except for symmetry) fail in a single algebra $\mathfrak A$ by \prettyref{thm:axioms}, which according to \prettyref{not:AA} is the special case $(\mathfrak{A,A})$ of $(\mathfrak{A,B})$---this means that, in general, Lepage's axioms fail in $(\mathfrak{A,B})$ as well.
\end{remark}

\citeA{Lepage03} proposed his axioms in the linguistic setting of words. This raises the following important question which is beyond the scope of the paper and which therefore remains an open problem.

\begin{problem}\label{problem:Lepage} Do Lepage's axioms hold within our framework in the word domain?
\end{problem}

\begin{remark}\label{rem:local} We shall emphasize that the negative results in \prettyref{thm:axioms} regarding transitivity show that the property of being in analogical proportion is a {\em local} property. By this we mean that, for example, the relation between $a$ and $b$, and between $b$ and $e$ does not fully determine the relation between $a$ and $e$ as inner transitivity fails in general. The same applies to (central) transitivity.
\end{remark}

\begin{problem}\label{problem:transitivity} In which algebras is the analogical proportion relation transitive and therefore an equivalence relation? The same question can be asked for inner and central transitivity.
\end{problem}

Finally, we have the following characterization of analogical proportions in algebras with no functions and no constants.

\begin{theorem}\label{thm:(A)} Let $\mathfrak A:=(A)$ be an algebra consisting only of its universe. For any $a,b,c,d\in A$, we have
\begin{align*}
    (A)\models a:b::c:d \quad\Leftrightarrow\quad (a=b\text{ and }c=d) \quad\text{or}\quad (a\neq b\text{ and }c\neq d).
\end{align*}
\end{theorem}
\begin{proof} $(\Leftarrow)$ (i) If $a=b$ and $c=d$, then $(A)\models a:b::c:d$ holds by inner reflexivity \prettyref{equ:inner_reflexivity}. (ii) If $a\neq b$ and $c\neq d$, then (cf. \prettyref{con:trivial})
\begin{align*} 
    Jus_{(A)}(a,b)\cup Jus_{(A)}(c,d)=Jus_{(A)}(b,a)\cup Jus_{(A)}(d,c)=\emptyset,
\end{align*} which entails $(A)\models a:b::c:d$.% follows from the Trivial Proportion \prettyref{thm:Trivial_Proportion_Theorem}.

$(\Rightarrow)$ By assumption, we have $(A)\models a\to b\righttherefore c\to d$. We distinguish two cases: (i) if $Jus_{(A)}(a,b)\cup Jus_{(A)}(c,d)$ consists only of trivial justifications, then we must have $a\neq b$ and $c\neq d$ (since otherwise the the non-trivial justification $z\to z$ would be included); (ii) otherwise, $Jus_{(A)}(a\to b\righttherefore c\to d)$ contains the only available non-trivial justification $z\to z$, which implies $a=b$ and $c=d$.
\end{proof}

\begin{corollary} In addition to the positive axioms of \prettyref{thm:axioms}, every algebra $\mathfrak A:=(A)$, consisting only of its universe, satisfies commutativity, inner transitivity, transitivity, central transitivity, and strong inner reflexivity.
\end{corollary}
\begin{proof} A direct consequence of \prettyref{thm:(A)}.
\end{proof}

\subsection{Isomorphism Theorems}\label{sec:Isomorphism_Theorems}

It is reasonable to expect isomorphisms---which are structure-preserving bijective mappings between algebras---to be compatible with analogical proportions. Consider the following simple example.

\begin{example}\label{exa:eta} Let $\Sigma:=\{a\}$ be the alphabet consisting of the single letter $a$, and let $\Sigma^\ast$ denote the set of all words over $\Sigma$ including the empty word $\varepsilon$. We can identify every sequence $a^n=a\ldots a$ ($n$ consecutive $a$'s) with the non-negative integer $n$, for every $n\geq 0$. Therefore, define the isomorphism $ H:(\mathbb N,+)\to (\Sigma^\ast,\cdot)$ via
\begin{align*} 
     H(0):=\varepsilon \quad\text{and}\quad  H(n):=a^n,\quad n\geq 1.
\end{align*} We expect the following analogical proportion to hold:
\begin{align*} 
    ((\mathbb N,+),(\Sigma^\ast,\cdot))\models m:n::a^m:a^n,\quad\text{for all }m,n\geq 0.
\end{align*} That this is indeed the case is the content of the First Isomorphism \prettyref{thm:FIT} below.
\end{example}

\begin{lemma}[Isomorphism Lemma]\label{lem:IL} For any isomorphism $ H:\mathfrak A\to \mathfrak B$ and any elements $a,b\in A$, we have
\begin{align}\label{equ:Jus_eta} 
    Jus_\mathfrak A(a,b)=Jus_\mathfrak B( H(a), H(b)).
\end{align}
\end{lemma}
\begin{proof} ($\subseteq$) Let $s(\mathbf z)\to t(\mathbf z)\in Jus_\mathfrak A(a,b)$. By \prettyref{def:Jus}, this means
\begin{align*} 
    a=s^\mathfrak A(\mathbf e) \quad\text{and}\quad b=t^\mathfrak A(\mathbf e),\quad\text{for some $\mathbf e\in A^{|\mathbf z|}$}.
\end{align*} Since $ H$ is an isomorphism, we have
\begin{align*} 
     H(a)= H(s^\mathfrak A(\mathbf e))=s^\mathfrak B( H(\mathbf e)) \quad\text{and}\quad  H(b)= H(t^\mathfrak A(\mathbf e))=t^\mathfrak B( H(\mathbf e)),
\end{align*} which shows $s\to t\in Jus_\mathfrak B( H(a), H(b))$ and consequently
\begin{align*} 
    Jus_\mathfrak A(a,b)\subseteq Jus_\mathfrak B( H(a), H(b)).
\end{align*}

($\supseteq$) Since $ H$ is an isomorphism, its inverse $ H^{-1}$ is an isomorphism as well, and we can apply the already shown case above to prove
\begin{align*} 
    Jus_\mathfrak B( H(a), H(b))\subseteq Jus_\mathfrak A( H^{-1}( H(a)), H^{-1}( H(b)))=Jus_\mathfrak A(a,b).
\end{align*}
\end{proof}

\begin{theorem}[First Isomorphism Theorem]\label{thm:FIT} For any isomorphism $ H:\mathfrak A\to\mathfrak B$ and any elements $a,b\in A$, we have
\begin{align}\label{equ:a_b_eta(a)_eta(b)} 
    (\mathfrak{A,B})\models a:b:: H(a): H(b).
\end{align}
\end{theorem}
\begin{proof} If $Jus_\mathfrak A(a,b)\cup Jus_\mathfrak B( H(a), H(b))$ consists only of trivial justifications, we are done. Otherwise, there is at least one non-trivial justification $s\to t$ in $Jus_\mathfrak A(a,b)$ or in $Jus_\mathfrak B( H(a), H(b))$, in which case the Isomorphism \prettyref{lem:IL} implies that $s\to t$ is in both $Jus_\mathfrak A(a,b)$ {\em and} $Jus_\mathfrak B( H(a), H(b))$, which means that $Jus_{(\mathfrak{A,B})}(a\to b\righttherefore H(a)\to H(b))$ contains at least one non-trivial justification as well. We proceed by showing that $Jus_{(\mathfrak{A,B})}(a\to b\righttherefore H(a)\to H(b))$ is subset maximal with respect to $ H(b)$:
\begin{align*} 
    Jus_{(\mathfrak{A,B})}(a\to b\righttherefore H(a)\to H(b))
    &=Jus_\mathfrak A(a,b)\cap Jus_\mathfrak B( H(a), H(b))\\
    &=Jus_\mathfrak A(a,b)\\
    &\supseteq Jus_\mathfrak A(a,b)\cap Jus_\mathfrak B( H(a),d)\\
    &=Jus_{(\mathfrak{A,B})}(a\to b\righttherefore  H(a),d),\quad\text{for every $d\in B$,}
\end{align*} where the second identity follows from \prettyref{lem:IL}. This shows that $Jus_{(\mathfrak{A,B})}(a\to b\righttherefore H(a)\to H(b))$ is indeed a subset maximal set of justifications of $a\to b\righttherefore  H(a)\to H(b)$ in $(\mathfrak{A,B})$ with respect to $ H(b)$. An analogous argument shows the remaining directed proportions (see \prettyref{fact:Sol}).
\end{proof}

\begin{remark} In \prettyref{sec:Isomorphisms} we will see a different proof of the First Isomorphism \prettyref{thm:FIT} from the logical perspective of model-theoretic types.
\end{remark}

The following counter-example shows that the First Isomorphism \prettyref{thm:FIT} cannot be generalized to homomorphisms. 
% However, we can prove in \prettyref{thm:FIT_hom} below a directed variant.

\begin{example}\label{exa:eta_hom} Let $\mathfrak A:=(\{a,b,a',b'\},f^\mathfrak A)$ and $\mathfrak B:=(\{c,d\},f^\mathfrak B)$, where $f^\mathfrak A$ and $f^\mathfrak B$ are unary functions defined by
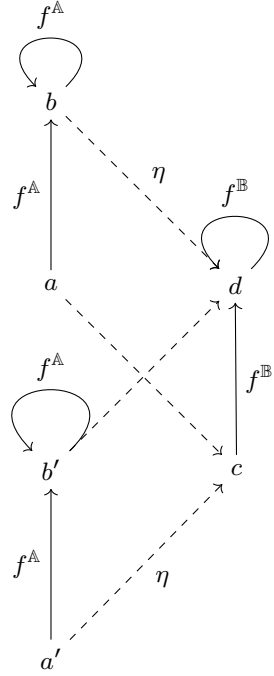
\begin{center}
\begin{tikzpicture} 
    \node (a') {$a'$};
    \node (b') [above=of a',yshift=1cm]  {$b'$};
    \node (a) [above=of b',yshift=1cm]  {$a$};
    \node (b) [above=of a,yshift=1cm]  {$b$};
    \node (d) [right=of a,xshift=1cm]  {$d$};
    \node (c) [right=of b',xshift=1cm]  {$c$};
    \draw[->] (a) to [edge label={$f^\mathfrak A$}] (b);
    \draw[->] (a') to [edge label={$f^\mathfrak A$}] (b');
    \draw[->] (c) to [edge label'={$f^\mathfrak B$}] (d);
    \draw[->, dashed] (b) to [edge label={$ H$}] (d);
    \draw[->, dashed] (b') to (d);
    \draw[->, dashed] (a) to (c);
    \draw[->, dashed] (a') to [edge label'={$ H$}] (c);
    \draw[->] (b) to [edge label'={$f^\mathfrak A$}] [loop] (b);
    \draw[->] (b') to [edge label'={$f^\mathfrak A$}] [loop] (b');
    \draw[->] (d) to [edge label'={$f^\mathfrak B$}] [loop] (d);
    % \draw (d) to [edge label'={$id$}] (d);
\end{tikzpicture}
\end{center} One can verify that $ H$ as defined by the dashed arrows in the figure above is a homomorphism from $\mathfrak A$ to $\mathfrak B$.

We claim
\begin{align*} 
    (\mathfrak{A,B})\not\models a':b:: H(a'): H(b).
\end{align*} We have
\begin{align*} 
    Jus_\mathfrak B( H(a'), H(b))=\left\{z\to f^\ell(z) \;\middle|\; \ell\geq 1\right\}
\end{align*} and therefore
\begin{align*} 
    Jus_\mathfrak A(a',b)\cup Jus_\mathfrak B( H(a'), H(b))\neq\emptyset,
\end{align*} but
\begin{align*} 
    Jus_{(\mathfrak{A,B})}(a'\to b\righttherefore H(a')\to H(b))=\emptyset,
\end{align*} which directly yields
\begin{align*} 
    (\mathfrak{A,B})\not\models a'\to b\righttherefore H(a')\to H(b).
\end{align*} %The intuitive reason is that there is an edge from $ H(a')$ to $ H(b)$ in $\matbb B$, but no edge from $a'$ to $b$ in $\mathfrak B$.
\end{example}

The following theorem shows that the analogical proportion relation is invariant under isomorphic transformations.

\begin{theorem}[Second Isomorphism Theorem]\label{thm:SIT} For any elements $a,b\in A$ and $c,d\in B$, and any isomorphisms $ H:\mathfrak A\to\mathfrak C$ and $ G:\mathfrak B\to\mathfrak D$, we have
\begin{align*} 
    (\mathfrak{A,B})\models a:b::c:d \quad\Leftrightarrow\quad (\mathfrak{C,D})\models  H(a):H(b)::G(c):G(d).
\end{align*}
\end{theorem}
\begin{proof}  An immediate consequence of the Isomorphism \prettyref{lem:IL} which yields
\begin{align*} 
    Jus_\mathfrak A(a,b)=Jus_\mathfrak C(H(a),H(b)) \quad\text{and}\quad Jus_\mathfrak B(c,d)=Jus_\mathfrak D(G(c),G(d)).
\end{align*}
\end{proof}

% todo counter-example for second iso thm and homomorphisms

\section{Sets and Numbers}\label{sec:Sets_}

In this section, we demonstrate our abstract algebraic framework of analogical proportions introduced above by investigating some elementary properties of proportions between sets and numbers, and by comparing our model with the models in \cite{Miclet08,Stroppa06}.

\subsection{Set Proportions}\label{sec:Set_}

This section studies analogical proportions between sets called set proportions.

\begin{notation}\label{not:set} In the rest of this section, let $L:=\{\cap,.^c\}$ be the language of sets, interpreted in the usual way, let $U$ and $W$ be universes, and let
\begin{align*} 
    \mathfrak A(U,W)&:=(\mathfrak P(U),\cap,.^c,\mathfrak P(U)\cap\mathfrak P(W)),\\ 
    \mathfrak B(W,U)&:=(\mathfrak P(W),\cap,.^c,\mathfrak P(U)\cap\mathfrak P(W)),
\end{align*} be $L(\mathfrak P(U)\cap\mathfrak P(W))$-algebras containing the distinguished sets in $\mathfrak P(U)\cap\mathfrak P(W)$ as constants (cf. \prettyref{not:L}). We will write $\mathfrak A(U)$ instead of $\mathfrak A(U,U)$. We introduce the following abbreviations:
\begin{align*} 
    A\cup B:=(A^c\cap B^c)^c \quad\text{and}\quad A-B:=A\cap B^c.
\end{align*} Notice that in case $\mathfrak A=\mathfrak B$, {\em every} set in $\mathfrak A$ is a distinguished set; the empty set is always a distinguished set.
\end{notation}

The following proposition summarizes some elementary properties of set proportions.

\begin{proposition}\label{prop:A_A^c_C_C^c} The following proportions hold in $(\mathfrak{A,B})$, for all $A\in A$, $C\in B$, and $B,E\in{A\cap B}$:
\begin{align}\
    \label{equ:A_A^c_C_C^c} A:A^c&::C:C^c,\\
    \label{equ:A_AuE_C_CuE} A:A\cup E&::C:C\cup E,\\
    \label{equ:A_AnE_C_CnE} A:A\cap E&::C:C\cap E,\\
    \label{equ:A-AuC_C-AuC} A\to A\cup C&\righttherefore C\to A\cup C \quad\text{if }A,C\in{A\cap B},\\
    \label{equ:A-AnC_C-AnC} A\to A\cap C&\righttherefore C\to A\cap C \quad\text{if }A,C\in{A\cap B},\\
    \label{equ:A-U_C-W} A\to U&\righttherefore C\to W,\\
    \label{equ:A-0_C-0} A\to \emptyset&\righttherefore C\to \emptyset.
\end{align} Moreover, in case $B\subseteq A$ and $B\subseteq C$, we further have the arrow proportion
\begin{align}\label{equ:A_B_C_B} A\to B\righttherefore C\to B.
\end{align}
\end{proposition}
\begin{proof} All proportions are immediate consequences of \prettyref{thm:FPT} with $t(Z)$ defined as follows:
\begin{itemize}
    \item The proportions in \prettyref{equ:A_A^c_C_C^c} follow from the fact that $t(Z):=Z^c$ is injective in $\mathfrak A$ and $\mathfrak B$.
    \item The directed proportions $A\to A\cup E\righttherefore C\to C\cup E$ and $C\to C\cup E\righttherefore A\to A\cup E$ in \prettyref{equ:A_AuE_C_CuE} follow with $t(Z):=Z\cup E$, and $A\cup E\to A\righttherefore C\cup E\to C$ and $C\cup E\to C\righttherefore A\cup E\to A$ with $t(Z):=Z-(E-Z)$.
    \item The directed proportions $A\to A\cap E\righttherefore C\to C\cap E$ and $C\to C\cap E\righttherefore A\to A\cap E$ in \prettyref{equ:A_AnE_C_CnE} follow with $t(Z):=Z\cap E$, and $A\cap E\to A\righttherefore C\cap E\to C$ and $C\cap E\to C\righttherefore A\cap E\to A$ with $t(Z):=Z\cup (Z-E)$.
    \item The directed proportion in \prettyref{equ:A-AuC_C-AuC} follows with both $t(Z):=Z\cup A\cup C$ and, since $A\cup C$ is a distinguished set by assumption, $t(Z):=A\cup C$.
    \item The directed proportion in \prettyref{equ:A-AnC_C-AnC} follows with both $t(Z):=Z\cup (A\cap C)$ and, since $A\cap C$ is a distinguished set by assumption, $t(Z):=A\cap C$.
    \item The directed proportion in \prettyref{equ:A-U_C-W} follows with $t(Z):=Z\cup Z^c$.
    \item The directed proportion in \prettyref{equ:A-0_C-0} follows with $t(Z):=Z\cap Z^c$ or $t(Z):=\emptyset$.
    \item From $B\subseteq A$ and $B\subseteq C$ we deduce that $B$ is a distinguished set, which means that $t(Z):=B$ is a valid definition implying \prettyref{equ:A_B_C_B}.
\end{itemize}
\end{proof}

\begin{example}\label{exa:tau_set} The rewrite rule
\begin{align*} 
    (X\cap Y)\cup (X-Y)\to (X\cap Y)\cup (Y-X)
\end{align*} justifies every arrow set proportion $A\to B\righttherefore C\to D$ in $(\mathfrak{A,B})$, which shows that it is a trivial justification in $(\mathfrak{A,B})$. This example shows that trivial justifications may contain useful information about the underlying structures---in this case, it encodes the trivial observation that any two sets $A$ and $B$ are symmetrically related via $A=(A\cap B)\cup (A-B)$ and $B=(A\cap B)\cup (B-A)$.
\end{example}

% In what follows, we want to compare our notion of set proportion with two models due to \cite{Stroppa06} and \cite{Miclet08}.

\subsubsection{Stroppa and Yvon}

The following definition is due to \citeA[Proposition 4]{Stroppa06}.\footnote{We adapt the definition in \cite{Stroppa06} to our schema by making the underlying structure $\mathfrak A$ explicit---recall from \prettyref{not:AA} that $\mathfrak A$ is an abbreviation for $(\mathfrak{A,A})$, which according to \prettyref{not:set} means that {\em every} set in $\mathfrak A$ is a distinguished set---this means, we can use every set in $\mathfrak A$ to form terms.}

\begin{definition}\label{def:models_SY_sets} For any sets $A,B,C,D\in A$, define
\begin{align*} 
    \mathfrak A\models_{SY} A:B::C:D \quad:\Leftrightarrow\quad &A=A_1\cup A_2,\quad B=A_1\cup D_2,\\
    &C=D_1\cup A_2,\quad D=D_1\cup D_2,\\
    &\text{for some $A_1,A_2,D_1,D_2\in A$.}
\end{align*}
\end{definition}

For example, with $A_1:=\{a_1\}$, $A_2:=\{a_2\}$, $D_1:=\{d_1\}$, and $D_2:=\{d_2\}$, we obtain the set proportion
\begin{align}\label{equ:a12} 
    \{a_1,a_2\}:\{a_1,d_2\}::\{d_1,a_2\}:\{d_1,d_2\}.
\end{align} So, roughly, we obtain the set $\{a_1,d_2\}$ from $\{a_1,a_2\}$ by replacing $a_2$ by $d_2$, which coincides with the transformation from $\{d_1,a_2\}$ into $\{d_1,d_2\}$.

Although \prettyref{def:models_SY_sets} works in some cases, in general we disagree with the notion of set proportions in \cite{Stroppa06} justified by the following counter-example.

\begin{example}\label{exa:aaa0} Let $A_1:=A_2:=\{a\}$ and $D_1:=D_2:=\emptyset$. \prettyref{def:models_SY_sets} yields
\begin{align*} 
    \mathfrak A(\{a\})\models_{SY}\{a\}:\{a\}::\{a\}:\emptyset.
\end{align*} This is implausible as it has no non-trivial justification. In fact, determinism \prettyref{equ:determinism} implies that $\{a\}$ is the only solution to $\{a\}:\{a\}::\{a\}:X$ in $\mathfrak A(\{a\})$ according to our \prettyref{def:models} (cf. \prettyref{thm:axioms}).
\end{example}

\subsubsection{Miclet, Bayoudh, and Delhay}

There is at least one more definition of set proportions in the literature due to \citeA[Definition 2.3]{Miclet08}.\footnote{To be more precise, the definition in \cite{Miclet08} is stated informally as
\begin{quote} Four sets $A,B,C$ and $D$ are in analogical proportion $A:B::C:D$ iff $A$ can be transformed into $B$, and $C$ into $D$, by adding and subtracting the same elements to $A$ and $C$.
\end{quote} This definition is ambiguous. One interpretation is the one we choose in \prettyref{def:models_MBD_sets}---another interpretation would be equivalent to \prettyref{def:models_SY_sets}. Moreover, as in the case of \prettyref{def:models_SY_sets}, we adapt the definition of \citeA{Miclet08} to our schema by making the underlying structure $\mathfrak A$ explicit.}

\begin{definition}\label{def:models_MBD_sets} Given a {\em finite} universe $U$ and sets $A,B,C,D\in A$, 
\begin{align*} 
    \mathfrak A\models_{MBD} A:B::C:D \quad:\Leftrightarrow\quad B=(A-E)\cup F \text{ and } D=(C-E)\cup F,\\\text{for some finite sets $E$ and $F$.}
\end{align*}
\end{definition}

% \begin{remark} Notice that both \cite{Stroppa06} and \cite{Miclet08} define set proportions only for sets over the same universe which is a serious restriction to its practical applicability. Even more problematic, \cite{Miclet08} define set proportions only for {\em finite} sets.
% \end{remark}

\begin{remark} Notice that \citeA{Miclet08} define set proportions only for {\em finite} sets which is a serious restriction to its practical applicability.
\end{remark}

For example, 
\begin{align*} 
    \{a_1,d_2\}=(\{a_1,a_2\}-\{a_2\})\cup\{d_2\}\text{ and }\{d_1,d_2\}=(\{d_1,a_2\}-\{a_2\})\cup\{d_2\}
\end{align*} shows that \prettyref{equ:a12} holds with respect to \prettyref{def:models_MBD_sets} as well.

We have the following implication.

\begin{theorem}\label{thm:models_MBD_sets} For any finite sets $A,B,C,D\in A$, we have
\begin{align*} 
    \mathfrak A\models_{MBD} A:B::C:D \quad\Rightarrow\quad \mathfrak A\models A:B::C:D.
\end{align*}
\end{theorem}
\begin{proof} Let $B$ and $D$ be written as in \prettyref{def:models_MBD_sets}. The arrow proportions $A\to B\righttherefore C\to D$ and $C\to D\righttherefore A\to B$ are immediate consequences of \prettyref{thm:FPT} with $t(Z):=(Z-E)\cup F$, and $B\to A\righttherefore D\to C$ and $D\to C\righttherefore B\to A$ with $t(Z):=(Z-(F-Z))\cup (Z\cap E)$.\footnote{Here it is important to emphasize that we assume every set in $\mathfrak A$ to be a distinguished set by \prettyref{not:set} (recall from \prettyref{not:AA} that $\mathfrak A$ is an abbreviation for $(\mathfrak{A,A})$).}
\end{proof}

The following example shows that the converse of \prettyref{thm:models_MBD_sets} fails in general.

\begin{example}\label{exa:ab0Z} Consider the analogical equation in $\mathfrak A(\{a,b\})$ given by
\begin{align}\label{equ:ab0Z} 
    \{a\}:\{b\}::\emptyset:X.
\end{align} As $\{b\}$ is the complement of $\{a\}$ in $\mathfrak A(\{a,b\})$, $\{a,b\}=\emptyset^c$ is a solution of \prettyref{equ:ab0Z} in our framework as a consequence of \prettyref{equ:A_A^c_C_C^c}. In contrast, since there are no finite sets $E$ and $F$ in $\mathfrak A(\{a,b\})$ satisfying $\{b\}=(\{a\}\cup E)-F$ and $\{a,b\}=(\emptyset\cup E)-F$, $\{a,b\}$ is not a solution of \prettyref{equ:ab0Z} according to \prettyref{def:models_MBD_sets}.
\end{example}

% stimmt so noch nicht, da ich nur pfeile ableiten kann
% \prettyref{thm:models_MBD_sets} together with \prettyref{exa:ab0Z} shows that our notion of set proportion yields strictly more justifiable solutions than the notion of \cite{Miclet08}.

\subsection{Arithmetical Proportions}

This section studies analogical proportions between numbers called arithmetical proportions. Let us first summarize some elementary properties.

\begin{proposition} For any integers $a,c\in\mathbb Z$, we have
\begin{align*} 
    (\mathbb Z,+,-)\models a:-a::c:-c,
\end{align*} and
\begin{align*} 
    (\mathbb Q,\cdot,.^{-1})\models a:\frac 1 a::c:\frac 1 c,\quad\text{for $a\neq 0$, $c\neq 0$},
\end{align*} and, given some distinguished integers $k,\ell\in\mathbb Z$, $k,\ell\neq 0$,
\begin{align*} 
    (\mathbb Z,+,\mathbb Z)\models a:ka+\ell::c:kc+\ell \quad\text{and}\quad (\mathbb Z,\cdot,\mathbb Z)\models a:a^k\cdot\ell::c:c^k\cdot\ell.
\end{align*} Moreover, for any integers $a,c,k,\ell,m,n\in\mathbb Z$, $k,m\neq 0$, we have
\begin{align*} 
    (\mathbb Z,+,\mathbb Z)\models ka+\ell:ma+n::kc+\ell:mc+n.
\end{align*}
\end{proposition}
\begin{proof} The first three lines are immediate consequences of \prettyref{thm:FPT} with $t(z)$ defined as follows: the first line is justified via $t(z):=-z$, the second via $t(z):=\frac 1 z$, and the third line is justified via $t(z):=kz+\ell$ and $t(z):=z^k\cdot\ell$, as in each case $t$ is injective in the respective algebra since $k,\ell\neq 0$ holds by assumption. Since the terms $s(z):=kz+\ell$ and $t(z):=mz+n$, $k,m\neq 0$, are injective in $(\mathbb Z,+,\mathbb Z)$ containing the same variable $z$, the Uniqueness \prettyref{lem:UL} implies via the justification $s(z)\to t(z)$ the arithmetical proportion $ka+\ell:ma+n::kc+\ell:mc+n$ in $(\mathbb Z,+,\mathbb Z)$.
\end{proof}

The following result formally proves two well-known arithmetical proportions known as `difference' and `geometric' proportion within our framework.

\begin{theorem}\label{thm:difference_proportion} 
% For any natural numbers $a,b,c,d\in\mathbb N$, % someday/maybe $succ$ definiert?
% \begin{align*} (\mathbb N,succ,0)\models a:b::c:d \quad\Leftrightarrow\quad a-b=c-d\quad\text{{\em (difference proportion)}}.
% \end{align*} 
For any integers $a,b,c,d\in\mathbb Z$,
\begin{align*} 
    &a-b=c-d \quad\Rightarrow\quad (\mathbb Z,+,-,\mathbb Z)\models a:b::c:d\quad\text{{\em (difference proportion)}},\\
    &\frac b a=\frac d c \quad\Rightarrow\quad (\mathbb Q,\cdot,.^{-1},\mathbb Q)\models a:b::c:d,\quad a\neq 0,c\neq 0,\quad\text{{\em (geometric proportion)}}.
\end{align*}
\end{theorem}
\begin{proof} The two proportions are direct consequences of \prettyref{thm:FPT} with $t(z):=z+b-a$ and $t(z):=z\frac b a$, respectively, since both terms are injective in the respective algebras.\footnote{Notice that in the algebra $(\mathbb Z,+,-,\mathbb Z)$, {\em every} integer is a distinguished element, which shows that the constants $a$ and $b$ in $z+b-a$ are syntactically correct.}
\end{proof}

The following counter-examples show that the converse of \prettyref{thm:difference_proportion} fails in general.

\begin{example} \prettyref{thm:FPT} implies%\footnote{Notice that $2m$ is an abbreviation for $m+m$, which does not contain multiplication.}
\begin{align*} 
    (\mathbb Z,+,-,\mathbb Z)\models a:2a::c:2c,\quad\text{for {\em all} integers $a$ and $c$.}
\end{align*} On the other hand, we have $2a-a=2c-c$ iff $a=c$. Similarly, \prettyref{thm:FPT} implies
\begin{align*} 
    (\mathbb Q,\cdot,.^{-1},\mathbb Q)\models a:a^2::c:c^2,\quad\text{for {\em all} integers $a$ and $c$, $a\neq 0$, $c\neq 0$.}
\end{align*} On the other hand, we have $\frac {a^2} a=\frac {c^2} c$ iff $a=c$.
\end{example}

The following result summarizes some unexpected and therefore `creative' arithmetical proportions containing the numbers 0 and 1.

\begin{proposition}\label{prop:a_0_c_0} For any integers $a,c,k,\ell\in\mathbb Z$, $k,\ell\neq 0$, we have
\begin{align} 
    \label{equ:a-0_c-0} (\mathbb Z,+,-,\mathbb Z)&\models a\to 0\righttherefore c\to 0, \\
    (\mathbb Q,\cdot,.^{-1},\mathbb Q)&\models a\to 1\righttherefore c\to 1,\\
    \label{equ:0_0_c_kc^ell} (\mathbb Z,+,\cdot,\mathbb Z)&\models 0:0::c:kc^\ell,\\
    (\mathbb Z,\cdot)&\models 1:1::c:c^\ell.
\end{align}
\end{proposition}
\begin{proof} An immediate consequence of \prettyref{thm:FPT} with $t(z)$ defined as follows: the first two arrow proportions are justified via $t(z):=z-z$ and $t(z):=\frac z z$; and the next two proportions are justified via $t(z):=kz^\ell$ and $t(z):=z^\ell$, both injective for $k,\ell\neq 0$ in the respective algebras.
\end{proof}

The following counter-example shows that the first two arrow proportions in \prettyref{prop:a_0_c_0} cannot be inverted.

\begin{example} To refute the inner converse $0\to a\righttherefore 0\to c$ of \prettyref{equ:a-0_c-0} in $(\mathbb Z,+,-,\mathbb Z)$, for $a\neq c$, notice that we can transform each justification $s\xrightarrow{\mathbf e_1\to\mathbf e_2}t$ of $0\to a\righttherefore 0\to c$ into a justification $s\xrightarrow{\mathbf e_1\to\mathbf e_1}t$ of $0\to a\righttherefore 0\to a$. On the other hand, $0\to a$ is a justification of $0\to a\righttherefore 0\to a$ which does not justify $0\to a\righttherefore 0\to c$ unless $c=a$. This shows
\begin{align*} 
    (\mathbb Z,+,-,\mathbb Z)\not\models 0\to a\righttherefore 0\to c,\quad\text{for $c\neq a$}.
\end{align*}

A similar argument shows
\begin{align*} 
    (\mathbb Q,\cdot,.^{-1},\mathbb Q)\not\models 1\to a\righttherefore 1\to c,\quad\text{for $c\neq a$}.
\end{align*}
\end{example}

\begin{example}\label{exa:0_0_2_40} \prettyref{prop:a_0_c_0} yields `creative' arithmetical proportions which are unexpected if their justifications are not made explicit. For example, given $c:=2$, $k:=10$, and $\ell:=2$, \prettyref{equ:0_0_c_kc^ell} yields the arithmetical proportion:\footnote{In this example it is essential that every coefficient occurring in a justification is a distinguished element.}
\begin{align*} 
    (\mathbb Z,+,\cdot,\mathbb Z)\models 0:0::2:40.
\end{align*} In such cases, where the characteristic justifications of a proportion are not obvious from the context, it is preferable to write (cf. \prettyref{not:J}):
\begin{align*} 
    (\mathbb Z,+,\cdot,\mathbb Z)\models_{\{z\to 10z^2\}} 0:0::2:40.
\end{align*} Here is another arithmetical proportion of this kind formalizing \prettyref{exa:2_0_3_z} (notice that $z+5$ and $1000z+3000$ are injective in $(\mathbb Z,+,\mathbb Z)$ and see \prettyref{lem:UL}):
\begin{align*} 
    (\mathbb Z,+,\mathbb Z)\models_{\{z+5\xrightarrow{-3\to-2}1000z+3000\}}2:0::3:1000.
\end{align*} We believe that analogical proportions of this form, which are unexpected but still reasonably justifiable, are crucial for formalizing creativity.
\end{example}

\begin{example}\label{exa:tau_numeric} The rewrite rule
\begin{align*} 
    x+y-y\to x+y-x
\end{align*} justifies any arrow proportion $a\to b\righttherefore c\to d$ in $(\mathbb Z,+,-)$, which shows that it is a trivial justification encoding the trivial observation that any two integers $a$ and $b$ are symmetrically related via $b=a+b-a$ and $a=b+a-b$.
\end{example}

\subsubsection*{Stroppa and Yvon}

The following notion of arithmetical proportion is an instance of the more general definition due to \citeA[Proposition 2]{Stroppa06} given for abelian semigroups.

\begin{definition}\label{def:models_SY_numbers} For any integers $a,b,c,d\in\mathbb Z$, define
\begin{align*} 
    (\mathbb Z,+,\mathbb Z)\models_{SY}a:b::c:d \quad:\Leftrightarrow\quad &a=a_1+a_2,\; b=a_1+d_2,\\ 
        &c=d_1+a_2,\; d=d_1+d_2,\\
        &\text{for some $a_1,a_2,d_1,d_2\in\mathbb Z$.}
\end{align*}
\end{definition}

For example, with $a:=1+1$, $b:=1+2$, $c:=2+1$, and $d:=2+2$, we obtain the arithmetical proportion
\begin{align*} 
    2:3::3:4.
\end{align*}

We have the following implication.

\begin{theorem}\label{thm:models_SY_numbers} For any integers $a,b,c,d\in\mathbb Z$, we have
\begin{align*} 
    (\mathbb Z,+,\mathbb Z)\models_{SY}a:b::c:d \quad\Rightarrow\quad (\mathbb Z,+,\mathbb Z)\models a:b::c:d.
\end{align*}
\end{theorem}
\begin{proof} An immediate consequence of \prettyref{thm:FPT} with $t(z):=z-a_2+d_2$, injective in $(\mathbb Z,+,\mathbb Z)$.
\end{proof}
% \begin{proof} Let $a,b,c,d$ be decomposed as in \prettyref{def:models_SY_numbers}. Define the terms
% \begin{align*} s(z):=z+a_2 \quad\text{and}\quad t(z):=z+d_2.
% \end{align*} Then $s\xrightarrow{a_1\to d_1} t$ is a justification of $a:b::c:d$ in $(\mathbb Z,+,\mathbb Z)$---since $s$ is injective in $(\mathbb Z,+,\mathbb Z)$, $s\to t$ is characteristic by \prettyref{lem:UL}.
% \end{proof}

The following counter-example shows that the converse of \prettyref{thm:models_SY_numbers} fails in general.

\begin{example}\label{exa:0_0_1_z} Consider the analogical equation in $(\mathbb Z,+,\mathbb Z)$ given by
\begin{align*} 
    0:0::1:x.
\end{align*} The justification $z\to z+z$, injective in $(\mathbb Z,+,\mathbb Z)$, implies the solution $x=2$ as a consequence of \prettyref{thm:FPT}. This solution cannot be obtained from \prettyref{def:models_SY_numbers} by the following argument. Suppose, towards a contradiction, that $0,1,2$ can be decomposed according to \prettyref{def:models_SY_numbers} into
\begin{align*} 
    0=a_1+a_2,\quad 0=a_1+d_2,\quad 1=d_1+a_2, \quad\text{and}\quad 2=d_1+d_2.
\end{align*} From the first two identities we deduce $a_2=d_2$, which further implies $1=d_1+a_2=d_1+d_2=2$---a contradiction.
\end{example}

\prettyref{thm:models_SY_numbers} together with \prettyref{exa:0_0_1_z} shows that our notion of arithmetical proportion yields strictly more justifiable solutions than the notion of Stroppa and Yvon \cite{Stroppa06}.

\section{Logical Interpretation}\label{sec:Logical_}

As we have constructed our model from first principles using only elementary concepts of universal algebra, and since our model questions some basic properties of analogical proportions presupposed in the literature (\prettyref{sec:Axioms}), to convince the reader of the `soundness' of our model, we show in this section that our purely algebraic model from above can be naturally embedded into the logical setting of first-order logic via model-theoretic types, which play a key role in classical model theory \cite<cf.>[§7.1]{Hinman05}. We then reprove our First Isomorphism \prettyref{thm:FIT}, which says that analogical proportions are compatible with isomorphisms, from this logical perspective. This hopefully convinces the reader that the introduced notions for formalizing analogical proportions from above---which are motivated by simple examples---are theoretically well-founded.

\subsection{Rewrite Types}\label{sec:_Types}

We reformulate our model of analogical proportions from above into first-order logic via a restricted form of model-theoretic types.

\begin{notation} In this section, $\mathfrak A$ and $\mathfrak B$ denote functional $L$-structures, which are $L$-structures containing no relation symbols.
\end{notation}

Recall from \prettyref{sec:Analogical_Proportions} that justifications of analogical proportions are rewrite rules of the form $s\to t$ expressing a functional relationship between elements. We now want to translate such justifications into logical formulas, which motivates the following definition.

\begin{definition}\label{def:varphi} We associate with each rewrite rule $s(\mathbf z)\to t(\mathbf z)$, where $t$ contains only variables occurring in $s$, the {\em $L$-rewrite formula}
\begin{align}\label{equ:varphi_s->t} 
    \varphi_{s(\mathbf z)\to t(\mathbf z)}(x,y):\equiv\exists\mathbf z\;[x=s(\mathbf z)\land y=t(\mathbf z)].
\end{align} We denote the set of all $L$-rewrite formulas by $rwFm_L$.
\end{definition}

Now that we have defined rewrite formulas, we continue by translating sets of justifications into sets of rewrite formulas via a restricted notion of the well-known model-theoretic types \cite<e.g.>[§7.1]{Hinman05}.

\begin{definition} Define the {\em $L$-rewrite type} of two elements $a,b\in A$ by
\begin{align*} 
    rwType_\mathfrak A(a,b):=\left\{\varphi(x,y)\in rwFm_L \;\middle|\; \mathfrak A\models\varphi(a,b)\right\},
\end{align*} extended to arrow proportions $a\to b\righttherefore c\to d$ in $(\mathfrak{A,B})$ by
\begin{align*} 
    rwType_{(\mathfrak{A,B})}(a\to b\righttherefore c\to d):=rwType_\mathfrak A(a,b)\cap rwType_\mathfrak B(c,d).
\end{align*}
\end{definition}

We have the following correspondence between sets of justifications and rewrite formulas and types (cf. \prettyref{con:abcd}):
\begin{align}\label{equ:Jus_rwType} 
    s\to t\in Jus_{(\mathfrak{A,B})}(a\to b\righttherefore c\to d) \quad\Leftrightarrow\quad \varphi_{s\to t}\in rwType_{(\mathfrak{A,B})}(a\to b\righttherefore c\to d).
\end{align} This is interesting as it shows that functional justifications of the form $s\to t$, which were motivated by simple examples, have an intuitive logical meaning. The expression `$a$ transforms into $b$ in $\mathfrak A$ as $c$ transforms into $d$ in $\mathfrak B$' can now be reinterpreted from a logical point of view as follows: the functional relationships between $a$ and $b$ in $\mathfrak A$ and between $c$ and $d$ in $\mathfrak B$ are formally captured by all rewrite formulas $\varphi(x,y)$ such that $\varphi(a,b)$ and $\varphi(b,a)$ holds in $\mathfrak A$ and $\varphi(c,d)$ and $\varphi(d,c)$ holds in $\mathfrak B$, respectively, and the elements $a,b,c,d$ are in directed analogical proportion in $(\mathfrak{A,B})$ iff the set of shared functional properties is maximal with respect to $d$. This leads us to the following logical variant of \prettyref{def:models}.

\begin{definition}\label{def:models_rw} We call an element $d\in B$ an {\em rw-solution} of an analogical equation $a:b::c:x$ in $(\mathfrak{A,B})$---in symbols, $(\mathfrak{A,B})\models_{rw} a:b::c:d$---iff it satisfies the conditions in \prettyref{def:models} with $Jus$ replaced by $rwType$. The notion of triviality is adapted in the obvious way.
\end{definition}

We have the following logical characterization of analogical proportions in terms of model-theoretic rewrite types.

\begin{fact}\label{fact:rw} For any functional $L$-structures $\mathfrak A$ and $\mathfrak B$, and any elements $a,b\in A$ and $c,d\in B$, we have
\begin{align*} 
    (\mathfrak{A,B})\models a:b::c:d \quad\Leftrightarrow\quad (\mathfrak{A,B})\models_{rw} a:b::c:d.
\end{align*}
\end{fact}
\begin{proof} An immediate consequence of \prettyref{equ:Jus_rwType}.
\end{proof}

\subsection{Isomorphisms}\label{sec:Isomorphisms}

We reprove the First Isomorphism \prettyref{thm:FIT} from the logical perspective of types.

\begin{proof}[Proof of \prettyref{thm:FIT}] If $rwType_\mathfrak A(a,b)\cup rwType_\mathfrak B( H(a), H(b))$ consists only of trivial rewrite formulas, we are done. Otherwise, we need to show that $rwType_{(\mathfrak{A,B})}(a\to b\righttherefore H(a)\to H(b))$ contains at least one non-trivial rewrite formula---this can be shown by an analogous argument as in the algebraic proof of \prettyref{thm:FIT} in \prettyref{sec:Isomorphism_Theorems}.

We proceed by showing that $rwType_{(\mathfrak{A,B})}(a\to b\righttherefore H(a)\to H(b))$ is subset maximal with respect to $ H(b)$. From $\mathfrak A\to\mathfrak B$ we deduce with \prettyref{lem:respects} that for any rewrite 2-$L$-formula $\varphi$,
\begin{align*} 
    \mathfrak A\models_{rw}\varphi(a,b) \quad\Leftrightarrow\quad \mathfrak B\models_{rw}\varphi( H(a), H(b)).
\end{align*} This further implies 
\begin{align*} 
    rwType_\mathfrak A(a,b)=rwType_\mathfrak B( H(a), H(b)),
\end{align*} and, consequently,
\begin{align*} 
    rwType_{(\mathfrak{A,B})}(a\to b\righttherefore  H(a)&\to  H(b))\\
    &=rwType_\mathfrak A(a,b)\cap rwType_\mathfrak B( H(a), H(b))\\
    &=rwType_\mathfrak A(a,b)\\
    &\supseteq rwType_\mathfrak A(a,b)\cap rwType_{(\mathfrak{A,B})}( H(a),d)\\
    &=rwType_{(\mathfrak{A,B})}(a\to b\righttherefore  H(a)\to d),\quad\text{for every $d\in B$.}
\end{align*} This shows that $rwType_{(\mathfrak{A,B})}(a\to b\righttherefore  H(a)\to  H(b))$ is indeed a subset maximal set of justifications of $a\to b\righttherefore  H(a)\to  H(b)$ in $(\mathfrak{A,B})$ with respect to $ H(b)$. An analogous argument shows the remaining directed proportions (see \prettyref{fact:Sol}). Now apply \prettyref{fact:rw}.
\end{proof}

\section{Related Work}\label{sec:Related_Work}

Arguably, the most prominent (symbolic) model of analogical reasoning to date is Gentner's \cite{Gentner83} {\em Structure-Mapping Theory} (or {\em SMT}), first implemented by Falkenhainer, Forbus, and Gentner \cite{Falklenhainer89}. Our approach shares with Gentner's SMT its symbolic nature. However, while in SMT mappings are constructed with respect to meta-logical considerations---for instance, Gentner's {\em systematicity principle} prefers connected knowledge over independent facts---in our framework `mappings' are realized via analogical proportions satisfying mathematically well-defined properties. In Theorems \ref{thm:FIT} and \ref{thm:SIT} we have shown that analogical proportions are compatible with structure-preserving mappings---a result which is in the vein of SMT. We leave a more detailed comparison between our algebraic approach and Gentner's SMT as future work.

Formal models of analogy have been studied by artificial intelligence researchers for decades \cite<cf.>{Hall89,Prade14a}. In this paper, we compared our algebraic framework of analogical proportions with two recently introduced models of analogical proportions from the literature \cite{Miclet08,Stroppa06}---introduced for applications to artificial intelligence and machine learning, specifically for natural language processing and handwritten character recognition---in the concrete domains of sets and numbers, and we showed that in each case we either disagree with the notion from the literature justified by some counter-example or we can show that our model yields strictly more justifiable solutions. This provides some evidence for its applicability. We expect similar results in other domains where the models of \cite{Stroppa06} and \cite{Miclet08} are applicable.% (it is not clear how to generalize the frameworks of \cite{Stroppa06} and \cite{Miclet08} to more general algebras). 

The functional-based view in \cite{Barbot19} is related to our \prettyref{thm:FPT} on the preservation of functional dependencies across different domains (\prettyref{sec:Functional_Proportion_Theorem}). The critical difference is that Barbot, Miclet, and Prade \cite{Barbot19} assume Lepage's central permutation axiom \prettyref{equ:central_permutation}, which implies in their framework that functional transformations need to be bijective, a serious restriction. In our framework, on the other hand, we proved in \prettyref{thm:axioms} that central permutation does not hold in general---justified by some reasonable counter-example---and so functional transformations can be functions induced by terms satisfying a mild injectivity condition (\prettyref{thm:FPT}).

To summarize, our framework differs substantially from the aforementioned models of analogical proportions:
\begin{enumerate}
\item Our model is abstract and algebraic in nature, formulated in the general language of universal algebra. Specific structures, like sets and numbers, are instances of the generic framework. To the best of our knowledge, this is not the case for the aforementioned frameworks, which are formulated only for concrete structures.

\item In our model, we make the underlying mathematical structures $\mathfrak A$ (source) and $\mathfrak B$ (target) explicit. This allows us to distinguish, for example, between the similar structures $(\mathbb Z,+)$ and $(\mathbb Z,\cdot)$, which yield different arithmetical proportions. More importantly, it allows us to derive analogical proportions between two {\em different} domains, e.g. numbers and words (\prettyref{exa:2_4_ab_z}). This distinction is not made in any of the aforementioned articles.

\item As a consequence, we could prove in \prettyref{thm:axioms} that all of Lepage's axioms (except for symmetry)---which are taken for granted in the aforementioned articles---are not axioms, but properties which may or may not hold in a specific structure. This has critical consequences, as structures satisfying some or all of Lepage's axioms behave differently than structures in which the properties fail. % For instance, central permutation and strong reflexivity---a property that is satisfied in every structure---imply strong reflexivity, which again should not be treated as an axiom (\prettyref{thm:axioms}).
\end{enumerate}

For further references on analogical reasoning we refer the interested reader to \cite{Hall89} and \cite{Prade14a}.

\subsection{Boolean Proportions}

Miclet and Prade \cite{Miclet09} \cite<cf.>{Prade10,Prade17} and Klein \cite{Klein82} study analogical proportions between boolean elements which correspond in our general framework to boolean proportions within the specific 2-element boolean algebra $\mathfrak{BOOL}$ (defined in \prettyref{sec:Preliminaries}), and which are related to set proportions as studied in \prettyref{sec:Set_}. In a recent paper \cite{Antic21-3}, we compare our model to the frameworks of \cite{Miclet09} and \cite{Klein82} in detail and derive the following conclusions:
\begin{enumerate}
\item In \cite{Antic21-3}, we derive, for every set of boolean constants $B\subseteq\mathfrak{BOOL}$:
\begin{align*} (\mathfrak{BOOL},\lor,\neg,B)\models a:b::c:d \quad\Leftrightarrow\quad (a=b\text{ and }c=d) \quad\text{or}\quad (a\neq b\text{ and }c\neq d).
\end{align*} Surprisingly, this turns out to be equivalent to Klein's \cite{Klein82} characterization of boolean proportions.

\item Miclet and Prade's \cite{Miclet09} definition, on the other hand, does not consider the proportions $0:1::1:0$ and $1:0::0:1$ to be in boolean proportion, `justified' on page 642 as follows:
\begin{quote} 
    The two other cases, namely $0:1::1:0$ and $1:0::0:1$, do not fit the idea that $a$ is to $b$ as $c$ is to $d$, since the changes from $a$ to $b$ and from $c$ to $d$ are not in the same sense. They in fact correspond to cases of maximal analogical dissimilarity, where `$d$ is not at all to $c$ what $b$ is to $a$', but rather `$c$ is to $d$ what $b$ is to $a$'.
\end{quote} Arguably, this is counter-intuitive as in case negation is available (which it implicitly is in \cite{Miclet09}), 
common sense tells us that `$a$ is to its negation $\neg a$ what $c$ is to its negation $\neg c$' makes sense, and $z\to\neg z$ (or, equivalently, $\neg z\to z$) is therefore a plausible characteristic justification of $0:1::1:0$ and $1:0::0:1$ in our framework.
\end{enumerate} 

These correspondences are promising as our model was not explicitly geared towards the boolean setting and provide further evidence for the applicability of our framework.

\subsection{Word Proportions}

A conceptually related approach to solving analogical word equations is given by \citeA{Dastani03}. At this point, it is not entirely clear how our framework relates to the rather complicated model of \cite{Dastani03} built on top of concepts such as `gestalts' of sequential patterns, structural information theory (SIT), algebraic coding systems for SIT, information load, representation systems, local homomorphism, constraints, et cetera. We challenge the reader to find instances where the model of \cite{Dastani03} is more expressive---in the word domain---than our model, which would (partially) justify their heavy machinery. To give a glimpse of what we mean, consider the following simple example \cite<cf.>[p.4]{Navarrete17}.

\begin{example}\label{exa:Dastani03} Let $\Sigma:=\{a,b,c,d\}$ be an alphabet. Define $\mathfrak A:=(\Sigma^+,\cdot^\mathfrak A,succ^\mathfrak A)$ and $\mathfrak B:=(\Sigma^+,\cdot^\mathfrak B,succ^\mathfrak B)$ via
\begin{align}\label{equ:lex} 
    succ^\mathfrak{A}(a):=b,\quad succ^\mathfrak{A}(b):=c,\quad succ^\mathfrak{A}(c):=d,\quad succ^\mathfrak{A}(d):=da
\end{align} and
\begin{align*} 
    succ^\mathfrak{B}(d):=c,\quad succ^\mathfrak{B}(c):=b,\quad succ^\mathfrak{B}(b):=a,\quad succ^\mathfrak{B}(a):=ad,
\end{align*} extended to words lexicographically.\footnote{See \url{https://en.wikipedia.org/wiki/Lexicographic_order}.} Consider the analogical equation in $(\mathfrak{A,B})$ given by
\begin{align*} 
    abc:abcd::dcb:x.
\end{align*} This equation is asking for a word which is to $dcb$ in $\mathfrak B$ what $abcd$ is to $abc$ in $\mathfrak A$. Observe that we obtain $abcd$ from $abc$ by concatenating the `successor' of $c$ at the end of $abc$. Since there is a unique $\mathbf e_2:=(d,c,b)$ satisfying $(z_1z_2z_3)^\mathfrak B(d,c,b)=dcb$ and $(z_1z_2z_3\cdot succ(z_3))^\mathfrak B(d,c,b)=dcba$ (notice that the empty word does not occur in $\mathfrak B$), the solution $z=dcba$ is characteristically justified by the Uniqueness \prettyref{lem:UL} via the justification
\begin{center}
\begin{tikzpicture}[node distance=1.5cm and 0cm]
% siehe tikz 3.6.1a §3.8
\node (a)               {$abc$};
\node (d1) [right=of a] {$\to$};
\node (b) [right=of d1] {$abcd$};
\node (d2) [right=of b] {$\righttherefore$};
\node (c) [right=of d2] {$dcb$};
\node (d3) [right=of c] {$\to$};
\node (d) [right=of d3] {$dcba$.};
\node (s) [below=of b] {$z_1z_2z_3$};
\node (t) [above=of c] {$z_1z_2z_3\cdot succ(z_3)$};

% siehe tikz 3.6.1a p. 165/166
\draw (a) to [edge label'={$\text{(unique)}\quad z_1z_2z_3/abc$}] (s); 
\draw (c) to [edge label={$z_1z_2z_3/dcb\quad\text{(unique)}$}] (s);
\draw (b) to [edge label={$\text{(unique)}\quad z_1z_2z_3/abc$}] (t);
\draw (d) to [edge label'={$z_1z_2z_3/dcb\quad\text{(unique)}$}] (t);
\end{tikzpicture}
\end{center} \citeA{Dastani03} obtain the same solution in a different way by using the algebras generated by the letters in $\Sigma$ and operators (named in their terminology `gestalts') such as `iteration', `successor', `symmetry', `alternation', `representation systems', et cetera, and, finally, by computing the solution $dcba$ via `local homomorphisms'.
\end{example}

\subsection{Category Theory}\label{sec:Category_Theory}

In mathematics, category theory \cite<cf.>{Awodey10} is a branch of algebra devoted to formalizing structural analogies. Roughly, the key axiom for a `collection of arrows' to form a category requires that to any pair of arrows $f:a\to b\in Arr(a,b)$ and $g:b\to c\in Arr(b,c)$ there is an arrow $g\circ f:a\to c\in Arr(a,c)$, called the {\em composite} of $f$ and $g$. There is a striking similarity between an arrow $f:a\to b\in Arr(a,b)$ in a category and a justification $s\to t\in Jus(a,b)$, which leads to the question whether the collection of all justifications (`arrows') of all pairs of elements forms a category; particularly, whether to each pair of justifications $s\to t\in Jus(a,b)$ and $r\to u\in Jus(b,c)$ there is a `composite' justification $s\to t\circ r\to u\in Jus(a,c)$. The following counter-example shows that this is, in general, {\em not} the case: consider the algebra $\mathfrak A:=(\{a,b,c\},f)$, where $f$ is given by (we omit the loop $f(b):=b$ in the figure)
\begin{center}
\begin{tikzpicture} 
\node (a) {$a$};
\node (b) [above=of a] {$b$};
\node (c) [above=of b] {$c$};

\draw[->] (a) to [edge label'={$f$}] (b);
\draw[->] (c) to [edge label={$f$}] (b);
\end{tikzpicture}
\end{center} We have
\begin{align*} 
    z\to f(z)\in Jus(a,b) \quad\text{and}\quad f(z)\to z\in Jus(b,c),
\end{align*} but there is no non-trivial justification in $Jus(a,c)$. It is interesting, however, to study analogical proportions in algebras where the justifications do form a category, which is related to \prettyref{problem:transitivity}. This is left as future work.

\section{Future Work}\label{sec:Future_Work}

This theoretical paper introduces and studies central properties of analogical proportions within the general setting of universal algebra and (the functional fragment of) first-order logic, and within the specific domains of sets and numbers.

\subsection{Multiple Variables}

One way to generalize our notion of an analogical equation is to allow multiple variables and higher-order terms to occur in different parts of an analogical equation as, for example, in 
\begin{align}\label{equ:term_proportion} 
    2+x_1:3+x_1::5+x_1+x_2:x_3.
\end{align} Here we are asking for a triple $\mathbf d=(d_1,d_2,d_3)$---the solution to \prettyref{equ:term_proportion}---satisfying the arithmetical proportion
\begin{align*} 
    2+d_1:3+d_1::5+d_1+d_2:d_3.
\end{align*} For instance, in $(\mathbb Z,+,-,\mathbb Z)$ a solution is given by $(2,0,8)$ yielding the difference proportion (cf. \prettyref{thm:difference_proportion})
\begin{align*} 
    (\mathbb Z,+,-,\mathbb Z)\models 4:5::7:8.
\end{align*} % Notice that \prettyref{equ:term_proportion} is related to {\em term proportions} with ground instances and solutions.

% \subsection{Sequences of Elements}

% Another way to generalize our notion of an analogical equation is to allow sequences of elements to occur in an analogical proportion as for example in
% \begin{align*} (a_1,\ldots,a_k):(b_1,\ldots,b_m)::(c_1,\ldots,c_n):x,\quad k,m,n\geq 1.
% \end{align*}

\subsection{Systems of Analogical Equations and Proportions}\label{sec:Systems_}

In this paper, we have considered only single analogical equations and proportions. We have argued in \prettyref{rem:local} that the property of being in analogical proportion is a {\em local} property. An interesting way to express and analyze {\em global} properties of an algebra is to study {\em systems of analogical equations and proportions}, where we can express multiple relationships between elements, of the form
\begin{align*} 
    a_1:b_1&::c_1:x\\
    &\vdots\\
    a_n:b_n&::c_n:x.
\end{align*} Here solutions are constrained by multiple equations and the identity
\begin{align*} 
    Sol_{(\mathfrak{A,B})}\left(
    \begin{array}{c}
        a_1:b_1::c_1:x\\
        \vdots\\
        a_n:b_n::c_n:x   
    \end{array}
    \right)=Sol_{(\mathfrak{A,B})}(a_1:b_1::c_1:x)\cap\ldots\cap Sol_{(\mathfrak{A,B})}(a_n:b_n::c_n:x)
\end{align*} shows that we can reduce solving systems of equations to solving single equations.

We can further generalize the framework by considering {\em infinite} systems of analogical equations and proportions, for example expressed via universal quantification. For instance, consider the infinite list of proportions, where $e\in A$ is a fixed element and $\mathfrak A=(A,\cdot)$ is an infinite algebra with $\cdot$ being a binary operation:
\begin{align}\label{equ:forall} \forall a\in A:\quad a:a::a:a\cdot e.
\end{align} By determinism \prettyref{equ:determinism}---shown in \prettyref{thm:axioms} to hold in any algebra $\mathfrak A$---we know that \prettyref{equ:forall} implies $a\cdot e=a$, which means that we can interpret \prettyref{equ:forall} as a {\em definition} of $e$ being a neutral element in $\mathfrak A$ in terms of analogical proportions.

\subsection{Algorithms}

From a practical point of view, the main task for future research is to develop algorithms for the computation of some or all solutions to analogical equations as defined in this paper. This problem is highly non-trivial in the general case and fairly complicated even in concrete cases (see \prettyref{exa:20_4_30_x_Appendix}). A reasonable starting point is therefore to first study small concrete mathematical domains such as, for example, the booleans or small number fields (e.g. the integers modulo m, $\mathbb Z_m$, for some small $m$) from the computational perspective. After that a reasonable next step is to study analogical proportions in finitely representable infinite structures \cite<cf.>{Ebbinghaus99,Libkin12}, which are more relevant to computer science and artificial intelligence research than the infinite models studied in classical universal algebra. Here interesting connections between, e.g., word proportions and logics on words studied in algebraic formal language and automata theory will hopefully become available, which may then lead to concrete algorithms for solving analogical equations over words, trees, and related data structures.
%This is at the core of the main task for future research which is to develop algorithms for the computation of some or all solutions to analogical equations as defined in this paper, which is highly non-trivial in general. %At its core, this requires algebraic methods for constructing and solving algebraic equations of the form \yref{equ:fgfg}-\yref{equ:ffgg}. For instance, in the arithmetic setting of \yref{sec:Numbers}, this task amounts to constructing some or all polynomials (i.e., justifications) $f$ and $g$ with integer coefficients, given some positive integers $m,n,$ and $k$, such that some line of Diophantine equations in \yref{equ:fgfg}-\yref{equ:ffgg} has integer solutions, which is non-trivial.

\subsection{Axioms}

In \prettyref{sec:Axioms}, we argued that \citeS{Lepage03} central permutation axiom---a critical axiom assumed by many authors \cite<e.g.>{Barbot19}---strong inner reflexivity, and strong reflexivity axioms are in general not satisfied as there are algebras in which these axioms fail, justified by counter-examples (\prettyref{thm:axioms}). It is interesting to investigate in which structures some of his axioms hold and to provide general characterizations (see \prettyref{problem:Lepage}).

\subsection{Applications to AI}

Another key line of research is to apply our model to various AI-related problems such as, e.g., common sense reasoning, formalizing metaphors, learning by analogy, and computational creativity. In \prettyref{thm:axioms} we have seen that analogical proportions are non-monotonic in nature, and it is important to investigate potential connections to non-monotonic reasoning in detail---itself a central research area of logic-based artificial intelligence prominently formalized by \cite{Gelfond91} in answer set programming \cite<cf.>{Brewka11}. Here interesting and surprising phenomena may occur. Moreover, it will be useful (and challenging) to fully apply our model to logic programming \cite<cf.>{Apt90} by first introducing appropriate algebraic operations on the space of all programs \cite{Antic21-2,Antic21-1}, and then by considering analogical proportions between logic programs of the form $P:Q::R:S$ \cite{Antic23-23}---in combination with unexpected or `creative' proportions, this line of research may lead to interesting formalizations of computational creativity \cite{Boden98}. We wish to expand this study to other domains relevant for computer science and artificial intelligence as, for instance, trees, graphs, automata, neural networks, et cetera. More broadly speaking, from the point of view of computational creativity, it is interesting to analyze unexpected (directed) analogical proportions as in Examples \ref{exa:2_0_3_z}, \ref{exa:0_0_2_40}, and \ref{exa:20_4_30_x}, and to try to find a qualitative notion of the degree of creativity of a solution to an analogical proportion in terms of its set of justifications. For instance, why does the solution $x=9$ of the arrow equation $20\to 4\righttherefore 30\to x$ in \prettyref{exa:20_4_30_x} appear to be more `creative' than $x=6$? Is there a relationship between the degree of unexpectedness or creativity of an analogical proportion and the algebraic structure of its set of justification? Related to this question, it will often be useful to have an evaluation criterion for the plausibility of an analogical proportion. In the setting of logic programming described before, for instance, the user may provide sets of positive and negative examples to impose additional constraints on solutions to analogical equations. In general, algorithmically ranking analogical proportions and formulating general principles for evaluating the plausibility of solutions to analogical equations is crucial and non-trivial.

\subsection{Universal Algebra}

From a mathematical point of view, relating analogical proportions to other concepts of universal algebra and related subjects is an interesting line of research. Specifically, studying analogical proportions in abstract mathematical structures like, for example, various kinds of lattices, semigroups and groups, rings, et cetera, is particularly interesting in the case of proportions between objects from different domains. Here it will be essential to study the relationship between properties of elements like being `neutral' or `absorbing' and their proportional properties (e.g. \prettyref{prop:a_0_c_0} and \prettyref{exa:0_0_2_40}; and see the discussion in the last paragraph of \prettyref{sec:Systems_}). At this point---due to the author's lack of expertise---it is not clear how exactly analogical proportions fit into the overall landscape of universal algebra and relating analogical proportions to other concepts of algebra and logic is therefore an important line of future research.

In this paper, we have studied analogical proportions between elements of possibly different algebras having the {\em same} underlying language $L$. It is challenging to generalize the framework to algebras over different languages, which requires an alignment of operations, possibly of different arity. For this, it might be useful to study meta-proportions between algebras of the form $\mathfrak{A:B::C:D}$.

\subsection{Logical Extensions}

In this paper, we have studied the {\em functional} case expressed via algebras and functional structures containing no relations and via rewrite formulas of the special logical form \prettyref{equ:varphi_s->t}, not containing relation symbols other than equality. An important next step is to enlarge the notion of a rewrite formula to include relation symbols and other logical constructs. This is challenging as allowing arbitrary formulas yields an overfitting where, in some domains, all elements are in analogical proportion. For example, consider the structure $(\mathbb N,succ,0)$, where $succ:\mathbb N\to\mathbb N$ is the successor function. In this structure, we can identify every natural number $a$ with the numeral $\underline a:=succ^a(0)$. Hence, given some natural numbers $a,b,c,d\in\mathbb N$, the formula
\begin{align*} 
    \varphi(x,y):\equiv(x=\underline a\land y=\underline b)\lor(x=\underline c\land y=\underline d)
\end{align*} characteristically justifies $a:b::c:d$ since $(\mathbb N,succ,0)\models\varphi(a',b')$ and $(\mathbb N,succ,0)\models\varphi(c',d')$ holds iff $a=a'$, $b=b'$, $c=c'$, and $d=d'$. The challenge is therefore to find generalizations of rewrite formulas which do not lead to overfitting. A reasonable starting point is to consider variants of rewrite formulas as, for example, formulas of the form
\begin{align*} \exists\mathbf z[xRs(\mathbf z)\land yRt(\mathbf z)],
\end{align*} where $R$ is an arbitrary binary relation symbol (other than equality).

\subsection{Term Rewriting}

Lastly, it is interesting to examine the role of {\em term rewriting} \cite<cf.>{Baader98} in analogical reasoning. More precisely, in this paper justifications of analogical proportions have the form $s\to t$, for some terms $s$ and $t$, which are rewriting rules as studied in term rewriting. This raises the question whether methods of term rewriting can be applied to reasoning about analogical proportions.

\section{Conclusion}\label{sec:Conclusion}

This paper introduced from first principles an abstract algebraic framework of analogical proportions in the general setting of universal algebra. This enabled us to compare mathematical objects possibly across different domains in a uniform way which is crucial for AI-systems. It turned out that our notion of analogical proportions has appealing mathematical properties. We showed that analogical proportions are compatible with injective functional transformations (\prettyref{thm:FPT}) and structure-preserving mappings (\prettyref{thm:FIT}) as desired. We further discussed \citeS{Lepage03} axioms and argued why in general we disagree with all of his axioms except for symmetry, while we agree with four axioms added in this paper, namely inner symmetry, inner reflexivity, reflexivity, and determinism (\prettyref{thm:axioms}). Moreover, it turned out that analogical proportions are non-monotonic in nature, which may have interesting connections to non-monotonic reasoning. We then compared our framework with two recently introduced frameworks of analogical proportions from the literature \cite{Miclet08,Stroppa06} within the concrete domains of sets and numbers, and in each case we either disagreed with the notion from the literature justified by some counter-example or we showed that our model yields strictly more justifiable solutions. Finally, we showed that our model has a natural logical interpretation in terms of model-theoretic types. This provides evidence for its applicability. In a broader sense, this paper is a first step towards a theory of analogical reasoning and learning systems with potential applications to fundamental AI-problems like common sense reasoning and computational learning and creativity.

\section*{Conflict of interest}

The authors declare that they have no conflict of interest.

\bibliographystyle{theapa}
\bibliography{/Users/christianantic/Bibdesk/Bibliography,/Users/christianantic/Bibdesk/Publications,/Users/christianantic/Bibdesk/Preprints,/Users/christianantic/Bibdesk/Preprints2}

\begin{thebibliography}{}

\bibitem[\protect\BCAY{Anti\'c}{Anti\'c}{2023a}]{Antic21-3}
Anti\'c, C. \BBOP2023a\BBCP.
\newblock \BBOQ Boolean proportions\BBCQ\
\newblock \url{https://arxiv.org/pdf/2109.00388.pdf}.

\bibitem[\protect\BCAY{Anti\'c}{Anti\'c}{2023b}]{Antic23-23}
Anti\'c, C. \BBOP2023b\BBCP.
\newblock \BBOQ Logic program proportions\BBCQ\
\newblock {\Bem Annals of Mathematics and Artificial Intelligence}.
\newblock \url{https://doi.org/10.1007/s10472-023-09904-8}.

\bibitem[\protect\BCAY{Anti\'c}{Anti\'c}{2023c}]{Antic21-2}
Anti\'c, C. \BBOP2023c\BBCP.
\newblock \BBOQ Sequential composition of answer set programs\BBCQ\
\newblock \url{https://arxiv.org/pdf/2104.12156.pdf}.

\bibitem[\protect\BCAY{Anti\'c}{Anti\'c}{2023d}]{Antic21-1}
Anti\'c, C. \BBOP2023d\BBCP.
\newblock \BBOQ Sequential composition of propositional logic programs\BBCQ\
\newblock \url{https://arxiv.org/pdf/2009.05774.pdf}.

\bibitem[\protect\BCAY{Apt}{Apt}{1990}]{Apt90}
Apt, K.~R. \BBOP1990\BBCP.
\newblock \BBOQ Logic programming\BBCQ\
\newblock In van Leeuwen, J.\BED, {\Bem Handbook of Theoretical Computer
  Science}, \lowercase{\BVOL}~B, \BPGS\ 493--574. Elsevier, Amsterdam.

\bibitem[\protect\BCAY{Awodey}{Awodey}{2010}]{Awodey10}
Awodey, S. \BBOP2010\BBCP.
\newblock {\Bem {Category Theory}\/} (2 \BEd)., \lowercase{\BVOL}~52 of {\Bem
  Oxford Logic Guides}.
\newblock Oxford University Press, New York.

\bibitem[\protect\BCAY{Baader\ \BBA\ Nipkow}{Baader\ \BBA\
  Nipkow}{1998}]{Baader98}
Baader, F.\BBACOMMA\  \BBA\ Nipkow, T. \BBOP1998\BBCP.
\newblock {\Bem Term Rewriting and All That}.
\newblock Cambridge University Press, Cambridge UK.

\bibitem[\protect\BCAY{Barbot, Miclet,\ \BBA\ Prade}{Barbot
  et~al.}{2019}]{Barbot19}
Barbot, N., Miclet, L., \BBA\ Prade, H. \BBOP2019\BBCP.
\newblock \BBOQ Analogy between concepts\BBCQ\
\newblock {\Bem Artificial Intelligence}, {\Bem 275}, 487--539.

\bibitem[\protect\BCAY{Boden}{Boden}{1998}]{Boden98}
Boden, M.~A. \BBOP1998\BBCP.
\newblock \BBOQ Creativity and artificial intelligence\BBCQ\
\newblock {\Bem Artificial Intelligence}, {\Bem 103\/}(1-2), 347--356.

\bibitem[\protect\BCAY{Brewka, Eiter,\ \BBA\ Truszczynski}{Brewka
  et~al.}{2011}]{Brewka11}
Brewka, G., Eiter, T., \BBA\ Truszczynski, M. \BBOP2011\BBCP.
\newblock \BBOQ {Answer set programming at a glance}\BBCQ\
\newblock {\Bem Communications of the ACM}, {\Bem 54\/}(12), 92--103.

\bibitem[\protect\BCAY{Burris\ \BBA\ Sankappanavar}{Burris\ \BBA\
  Sankappanavar}{2000}]{Burris00}
Burris, S.\BBACOMMA\  \BBA\ Sankappanavar, H. \BBOP2000\BBCP.
\newblock {\Bem {A Course in Universal Algebra}}.
\newblock \url{http://www.math.hawaii.edu/~ralph/Classes/619/univ-algebra.pdf}.

\bibitem[\protect\BCAY{Chang\ \BBA\ Keisler}{Chang\ \BBA\
  Keisler}{1973}]{Chang73a}
Chang, C.~C.\BBACOMMA\  \BBA\ Keisler, H.~J. \BBOP1973\BBCP.
\newblock {\Bem Model Theory}.
\newblock North-Holland, Amsterdam.

\bibitem[\protect\BCAY{Correa, Prade,\ \BBA\ Richard}{Correa
  et~al.}{2012}]{Correa12}
Correa, W., Prade, H., \BBA\ Richard, G. \BBOP2012\BBCP.
\newblock \BBOQ When intelligence is just a matter of copying\BBCQ\
\newblock In Raedt, L.~D., Bessiere, C., Dubois, D., Doherty, P., Frasconi, P.,
  Heintz, F., \BBA\ Lucas, P.\BEDS, {\Bem ECAI 2012}, \lowercase{\BVOL}\ 242 of
  {\Bem Frontiers in Artificial Intelligence and Applications}, \BPGS\
  276--281.

\bibitem[\protect\BCAY{Dastani, Indurkhya,\ \BBA\ Scha}{Dastani
  et~al.}{2003}]{Dastani03}
Dastani, M., Indurkhya, B., \BBA\ Scha, R. \BBOP2003\BBCP.
\newblock \BBOQ Analogical projection in pattern perception\BBCQ\
\newblock {\Bem Journal of Experimental \& Theoretical Artificial
  Intelligence}, {\Bem 15\/}(4), 489--511.

\bibitem[\protect\BCAY{Ebbinghaus\ \BBA\ Flum}{Ebbinghaus\ \BBA\
  Flum}{1999}]{Ebbinghaus99}
Ebbinghaus, H.-D.\BBACOMMA\  \BBA\ Flum, J. \BBOP1999\BBCP.
\newblock {\Bem Finite Model Theory\/} (2 \BEd).
\newblock Springer Monographs in Mathematics. Springer-Verlag,
  Berlin/Heidelberg.

\bibitem[\protect\BCAY{Eiter, Ianni,\ \BBA\ Krennwallner}{Eiter
  et~al.}{2009}]{Eiter09}
Eiter, T., Ianni, G., \BBA\ Krennwallner, T. \BBOP2009\BBCP.
\newblock \BBOQ {Answer set programming: a primer}\BBCQ\
\newblock In {\Bem Reasoning Web. Semantic Technologies for Information
  Systems, {\em volume 5689 of} Lecture Notes in Computer Science}, \BPGS\
  40--110. Springer, Heidelberg.

\bibitem[\protect\BCAY{Falkenhainer, Forbus,\ \BBA\ Gentner}{Falkenhainer
  et~al.}{1989}]{Falklenhainer89}
Falkenhainer, B., Forbus, K.~D., \BBA\ Gentner, D. \BBOP1989\BBCP.
\newblock \BBOQ The structure-mapping engine: algorithm and examples\BBCQ\
\newblock {\Bem Artificial Intelligence}, {\Bem 41\/}(1), 1--63.

\bibitem[\protect\BCAY{Gelfond\ \BBA\ Lifschitz}{Gelfond\ \BBA\
  Lifschitz}{1991}]{Gelfond91}
Gelfond, M.\BBACOMMA\  \BBA\ Lifschitz, V. \BBOP1991\BBCP.
\newblock \BBOQ Classical negation in logic programs and disjunctive
  databases\BBCQ\
\newblock {\Bem New Generation Computing}, {\Bem 9\/}(3-4), 365--385.

\bibitem[\protect\BCAY{Gentner}{Gentner}{1983}]{Gentner83}
Gentner, D. \BBOP1983\BBCP.
\newblock \BBOQ Structure-mapping: a theoretical framework for analogy\BBCQ\
\newblock {\Bem Cognitive Science}, {\Bem 7\/}(2), 155--170.

\bibitem[\protect\BCAY{Gust, Krumnack, K{\"u}hnberger,\ \BBA\ Schwering}{Gust
  et~al.}{2008}]{Gust08}
Gust, H., Krumnack, U., K{\"u}hnberger, K.-U., \BBA\ Schwering, A.
  \BBOP2008\BBCP.
\newblock \BBOQ Analogical reasoning: a core of cognition\BBCQ\
\newblock {\Bem K{\"u}nstliche Intelligenz}, {\Bem 22\/}(1), 8--12.

\bibitem[\protect\BCAY{Hall}{Hall}{1989}]{Hall89}
Hall, R.~P. \BBOP1989\BBCP.
\newblock \BBOQ Computational approaches to analogical reasoning: a comparative
  analysis\BBCQ\
\newblock {\Bem Artificial Intelligence}, {\Bem 39\/}(1), 39--120.

\bibitem[\protect\BCAY{Hinman}{Hinman}{2005}]{Hinman05}
Hinman, P.~G. \BBOP2005\BBCP.
\newblock {\Bem {Fundamentals of Mathematical Logic}}.
\newblock A K Peters, Wellesley, MA.

\bibitem[\protect\BCAY{Hofstadter}{Hofstadter}{2001}]{Hofstadter01}
Hofstadter, D. \BBOP2001\BBCP.
\newblock \BBOQ Analogy as the core of cognition\BBCQ\
\newblock In Gentner, D., Holyoak, K.~J., \BBA\ Kokinov, B.~K.\BEDS, {\Bem {The
  Analogical Mind: Perspectives from Cognitive Science}}, \BPGS\ 499--538. MIT
  Press/Bradford Book, Cambridge MA.

\bibitem[\protect\BCAY{Hofstadter\ \BBA\ Mitchell}{Hofstadter\ \BBA\
  Mitchell}{1995}]{Hofstadter95a}
Hofstadter, D.\BBACOMMA\  \BBA\ Mitchell, M. \BBOP1995\BBCP.
\newblock \BBOQ The copycat project: a model of mental fluidity and
  analogy-making\BBCQ\
\newblock In {\Bem Fluid Concepts and Creative Analogies. Computer Models of
  the Fundamental Mechanisms of Thought}, \BCH~5, \BPGS\ 205--267. Basic Books,
  New York.

\bibitem[\protect\BCAY{Hofstadter\ \BBA\ Sander}{Hofstadter\ \BBA\
  Sander}{2013}]{Hofstadter13}
Hofstadter, D.\BBACOMMA\  \BBA\ Sander, E. \BBOP2013\BBCP.
\newblock {\Bem Surfaces and Essences. Analogy as the Fuel and Fire of
  Thinking}.
\newblock Basic Books, New York.

\bibitem[\protect\BCAY{Klein}{Klein}{1982}]{Klein82}
Klein, S. \BBOP1982\BBCP.
\newblock \BBOQ Culture, mysticism and social structure and the calculation of
  behavior\BBCQ\
\newblock In {\Bem ECAI 1982}, \BPGS\ 141--146.

\bibitem[\protect\BCAY{Krieger}{Krieger}{2003}]{Krieger03}
Krieger, M.~H. \BBOP2003\BBCP.
\newblock {\Bem {Doing Mathematics: Convention, Subject, Calculation,
  Analogy}}.
\newblock World Scientific, New Jersey.

\bibitem[\protect\BCAY{Lepage}{Lepage}{2003}]{Lepage03}
Lepage, Y. \BBOP2003\BBCP.
\newblock {\Bem {De L'Analogie. Rendant Compte de la Commutation en
  Linguistique}}.
\newblock Ha{\-}bi{\-}li{\-}ta{\-}tion \`a diriger les recherches, Universit\'e
  Joseph Fourier, Grenoble.

\bibitem[\protect\BCAY{Libkin}{Libkin}{2012}]{Libkin12}
Libkin, L. \BBOP2012\BBCP.
\newblock {\Bem Elements of Finite Model Theory}.
\newblock Springer-Verlag, Berlin/Heidelberg.

\bibitem[\protect\BCAY{Miclet, Bayoudh,\ \BBA\ Delhay}{Miclet
  et~al.}{2008}]{Miclet08}
Miclet, L., Bayoudh, S., \BBA\ Delhay, A. \BBOP2008\BBCP.
\newblock \BBOQ Analogical dissimilarity: definition, algorithms and two
  experiments in machine learning\BBCQ\
\newblock {\Bem Journal of Artificial Intelligence Research}, {\Bem 32},
  793--824.

\bibitem[\protect\BCAY{Miclet\ \BBA\ Prade}{Miclet\ \BBA\
  Prade}{2009}]{Miclet09}
Miclet, L.\BBACOMMA\  \BBA\ Prade, H. \BBOP2009\BBCP.
\newblock \BBOQ Handling analogical proportions in classical logic and fuzzy
  logics settings\BBCQ\
\newblock In Sossai, C.\BBACOMMA\  \BBA\ Chemello, G.\BEDS, {\Bem ECSQARU 2009,
  LNAI 5590}, \BPGS\ 638--650. Springer-Verlag, Berlin/Heidelberg.

\bibitem[\protect\BCAY{Navarrete\ \BBA\ Dartnell}{Navarrete\ \BBA\
  Dartnell}{2017}]{Navarrete17}
Navarrete, J.~A.\BBACOMMA\  \BBA\ Dartnell, P. \BBOP2017\BBCP.
\newblock \BBOQ Towards a category theory approach to analogy: analyzing
  re-representation and acquisition of numerical knowledge\BBCQ\
\newblock {\Bem Computatinal Biology}, {\Bem 13\/}(8), 1--38.

\bibitem[\protect\BCAY{P{\'o}lya}{P{\'o}lya}{1954}]{Polya54}
P{\'o}lya, G. \BBOP1954\BBCP.
\newblock {\Bem {Induction and Analogy in Mathematics}}, \lowercase{\BVOL}~1 of
  {\Bem Mathematics and Plausible Reasoning}.
\newblock Princeton University Press, Princeton, New Jersey.

\bibitem[\protect\BCAY{Prade\ \BBA\ Richard}{Prade\ \BBA\
  Richard}{2010}]{Prade10}
Prade, H.\BBACOMMA\  \BBA\ Richard, G. \BBOP2010\BBCP.
\newblock \BBOQ Reasoning with logical proportions\BBCQ\
\newblock In {\Bem KR 2010}, \BPGS\ 545--555. AAAI Press.

\bibitem[\protect\BCAY{Prade\ \BBA\ Richard}{Prade\ \BBA\
  Richard}{2014}]{Prade14a}
Prade, H.\BBACOMMA\  \BBA\ Richard, G. \BBOP2014\BBCP.
\newblock \BBOQ A short introduction to computational trends in analogical
  reasoning\BBCQ\
\newblock In Prade, H.\BBACOMMA\  \BBA\ Richard, G.\BEDS, {\Bem Approaches to
  Analogical Reasoning: Current Trends}, Studies in Computational Intelligence
  548, \BPGS\ 1--22. Springer-Verlag, Berlin/Heidelberg.

\bibitem[\protect\BCAY{Prade\ \BBA\ Richard}{Prade\ \BBA\
  Richard}{2017}]{Prade17}
Prade, H.\BBACOMMA\  \BBA\ Richard, G. \BBOP2017\BBCP.
\newblock \BBOQ Analogical proportions and analogical reasoning --- an
  introduction\BBCQ\
\newblock In Aha, D.~W.\BBACOMMA\  \BBA\ Lieber, J.\BEDS, {\Bem ICCBR 2017,
  LNAI 10339}, \BPGS\ 16--32. Springer, Berlin, Heidelberg.

\bibitem[\protect\BCAY{Sowa\ \BBA\ Majumdar}{Sowa\ \BBA\
  Majumdar}{2003}]{Sowa03}
Sowa, J.~F.\BBACOMMA\  \BBA\ Majumdar, A.~K. \BBOP2003\BBCP.
\newblock \BBOQ Analogical reasoning\BBCQ\
\newblock In Ganter, B., Moor, A., \BBA\ Lex, W.\BEDS, {\Bem ICCS 2003, LNAI
  2746}, \BPGS\ 16--36. Springer-Verlag, Berlin/Heidelberg.

\bibitem[\protect\BCAY{Stroppa\ \BBA\ Yvon}{Stroppa\ \BBA\
  Yvon}{2006}]{Stroppa06}
Stroppa, N.\BBACOMMA\  \BBA\ Yvon, F. \BBOP2006\BBCP.
\newblock \BBOQ Formal models of analogical proportions\BBCQ\
\newblock {Technical Report D008}, Telecom ParisTech - \'Ecole Nationale
  Sup\'erieure de T\'el\'ecommunications, T\'el\'ecom Paris.

\bibitem[\protect\BCAY{Winston}{Winston}{1980}]{Winston80}
Winston, P.~H. \BBOP1980\BBCP.
\newblock \BBOQ Learning and reasoning by analogy\BBCQ\
\newblock {\Bem Communications of the ACM}, {\Bem 23\/}(12), 689--703.

\bibitem[\protect\BCAY{Wos}{Wos}{1993}]{Wos93}
Wos, L. \BBOP1993\BBCP.
\newblock \BBOQ The problem of reasoning by analogy\BBCQ\
\newblock {\Bem Journal of Automated Reasoning}, {\Bem 10\/}(3), 421--422.

\end{thebibliography}

\appendix
\newpage

\section*{Appendix}

\begin{example}\label{exa:20_4_30_x_Appendix} Explicating \prettyref{exa:20_4_30_x}, we wish to compute all solutions to the analogical equation in the multiplicative algebra
\begin{align*} 
    \mathfrak M:=(\mathbb N_2,\cdot,\mathbb N_2)
\end{align*} given by
\begin{align}\label{equ:20_4_30_x} 
    20:4::30:x.
\end{align} Recall that $\mathbb N_2=\{2,3,\ldots\}$ consists of the natural numbers not containing 0 and 1---this will allow us to bound the number of generalizations of a given number by the number of its prime factors. Following \prettyref{fact:Sol}, we first compute  all solutions to the arrow equation in $\mathfrak M$ given by
\begin{align}\label{equ:20-4_30-z} 
    20\to 4\righttherefore 30\to x.
\end{align} The justification $z\to 4$ is unique to $x=4$, which immediately entails that it is a solution to \prettyref{equ:20-4_30-z}. Notice that in the algebra $(\mathbb Q,\cdot)$, the rewrite rule $z\to\frac z 5$ is a justification of $20\to 4\righttherefore 30\to d$ iff $d=6$, which yields the solution $x=\frac{30} 5=6$. Unfortunately, this justification is not available in $\mathfrak M$ and to prove that $6$ is a solution to $20\to 4\righttherefore 30\to x$, we need to show that either $Jus_\mathfrak M(20,4)\cup Jus_\mathfrak M(30,6)$ contains only trivial justifications, or that $Jus_{\mathfrak M}(20\to 4\righttherefore 30\to 6)$ is a non-empty maximal set of justifications with respect to 6. Here a natural question arises: Are $x=4,6$ the only solutions to the arrow equation \prettyref{equ:20-4_30-z} in $\mathfrak M$? The answer is `no' as we can show that, unexpectedly, $x=9$ is another justifiable solution!

We begin by computing all $\mathfrak M$-generalizations of 20, 4, and 30. For this, it will be convenient to first compute their unique prime decompositions:
\begin{align}\label{equ:20_4_30} 20=2\cdot 2\cdot 5 \quad\text{and}\quad 4=2\cdot 2 \quad\text{and}\quad 30=2\cdot 3\cdot 5.
\end{align} These numbers are similar in the following sense:
\begin{enumerate}
\item The only difference between 20 and 30 is the second prime factor; both numbers are divisible by 2 and 5.
\item The only difference between 20 and 4 is that 4 is not divisible by 5; both numbers are divisible by 2 and 4.
\item The numbers 4 and 30 have the first prime factor 2 in common.
\end{enumerate} These similarities are reflected in the computation of their $\mathfrak M$-generalizations:
\begin{align*} gen_{\mathfrak M}(20)&= \left\{
    \begin{array}{lll}
        2\cdot 2\cdot 5 & z_1\cdot 2\cdot 5 & z_1\cdot z_1\cdot 5\\
        2\cdot z_2\cdot 5 & 2\cdot 2\cdot z_3\qquad & z_1\cdot z_2\cdot 5\\
        2\cdot z_2\cdot z_3\qquad & z_1\cdot 2\cdot z_3 & z_1\cdot z_2\cdot z_3\\
        z_1\cdot 2 & z_1\cdot 5 & z_1\cdot z_2\\
        z_1\cdot z_1\cdot z_2 & z_1
    \end{array}
    \right\}\\
    &= \left\{
    \begin{array}{lll}
        20 & 10z & 5z^2\\
         & 4z\qquad & 5z_1z_2\\
        2z_1z_2\qquad & & z_1z_2z_3\\
        2z & 5z & z_1z_2\\
        z_1^2z_2 & z
    \end{array}
    \right\}
\end{align*}
\begin{align*} 
    gen_{\mathfrak M}(30)&= \left\{
    \begin{array}{lll}
        2\cdot 3\cdot 5 & z_1\cdot 3\cdot 5 & \\
        2\cdot z_2\cdot 5 & 2\cdot 3\cdot z_3\qquad & z_1\cdot z_2\cdot 5\\
        2\cdot z_2\cdot z_3\qquad & z_1\cdot 3\cdot z_3 & z_1\cdot z_2\cdot z_3\\
        z_1\cdot 2 & z_1\cdot 5 & z_1\cdot z_2\\
        & z_1\cdot 3 & z_1
    \end{array}
    \right\}\\
    &= \left\{
    \begin{array}{lll}
        30 & 15z & \\
        10z & 6z & 5z_1z_2\\
        2z_1z_2\qquad & 3z_1z_3\qquad & z_1z_2z_3\\
        2z & 5z & z_1z_2\\
        & 3z & z
    \end{array}
    \right\}
\end{align*}
\begin{align*} 
    gen_{\mathfrak M}(4)
        &=\{2\cdot 2,2\cdot z_2,z_1\cdot 2,z_1\cdot z_2,z_1\cdot z_1,z_1\}\\
        &=\{4,2z,z_1z_2,z^2,z\}.
\end{align*} This yields
\begin{align*} 
    gen_\mathfrak M(20,30)&=gen_\mathfrak M(20)\cap gen_\mathfrak M(30)\\
    % &=\{2\cdot z_2\cdot 5,\;z_1\cdot z_2\cdot 5,\;2\cdot z_2\cdot z_3,\;z_1\cdot z_2\cdot z_3,\;z_1\cdot 5,\;z_1\cdot 2,\;z_1\cdot z_2,\;z_1\}\\
    &=\{10z,5z_1z_2,2z_1z_2,z_1z_2z_3,5z,2z,z_1z_2,z\}.
\end{align*} We now compute all non-empty sets of justifications $Jus_{\mathfrak M}(20\to 4\righttherefore 30\to d)$, for all $d\in\mathbb N_2$, with the following procedure (cf. \prettyref{pseudo:Sol}): for each $\mathfrak M$-generalization $s(\mathbf z)\in gen_{\mathfrak M}(20,30)$ and each witness $(\mathbf e_1,\mathbf e_2)$ satisfying\footnote{See \prettyref{con:abcd}.}\footnote{We omit here the superscript from notation, that is, we write $s(\mathbf e_1)$ instead of $s^{\mathfrak M}(\mathbf e_1)$ et cetera.}
\begin{align}\label{equ:20_30} 20=s(\mathbf e_1) \quad\text{and}\quad 30=s(\mathbf e_2),
\end{align} and for each $t(\mathbf z)\in gen_{\mathfrak M}(4)$ satisfying
\begin{align*} 4=t(\mathbf e_1),
\end{align*} we add $s\to t$ to $Jus_{\mathfrak M}(20\to 4\righttherefore 30\to t(\mathbf e_2))$. 

First of all, notice that $z\to 4$ is (only) in $Jus_{\mathfrak M}(20\to 4\righttherefore 30\to 4)$, which immediately yields the solution $x=4$ to \prettyref{equ:20-4_30-z} (cf. \prettyref{rem:t}).

\begin{enumerate}
\item $s(z)=10z$: The only witnesses $e_1,e_2\in M$ satisfying \prettyref{equ:20_30} are $e_1=2$ and $e_2=3$, and the only $t_1,t_2\in gen_{\mathfrak M}(4)$ satisfying $4=t(2)$ are $t_1(z)=2z$ and $t_2(z)=z^2$. We therefore have $10z\to 2z\in Jus_{\mathfrak M}(20\to 4\righttherefore 30\to t_1(3))=Jus_{\mathfrak M}(20\to 4\righttherefore 30\to 6)$ and $10z\to z^2\in Jus_{\mathfrak M}(20\to 4\righttherefore 30\to t_2(3))=Jus_{\mathfrak M}(20\to 4\righttherefore 30\to 9)$. This can be depicted as follows:
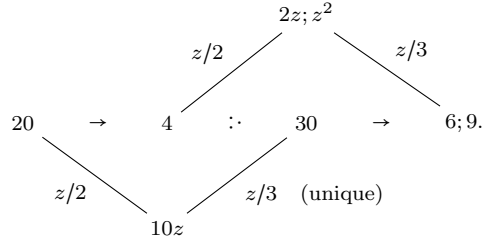
\begin{center}
\begin{tikzpicture}[node distance=1cm and 0.5cm]
% siehe tikz 3.6.1a §3.8
\node (a)               {$20$};
\node (d1) [right=of a] {$\to $};
\node (b) [right=of d1] {$4$};
\node (d2) [right=of b] {$\righttherefore $};
\node (c) [right=of d2] {$30$};
\node (d3) [right=of c] {$\to $};
\node (d) [right=of d3] {$6;9$.};
\node (s) [below=of b] {$10z$};
\node (t) [above=of c] {$2z; z^2$};

% siehe tikz 3.6.1a p. 165/166
\draw (a) to [edge label'={$z/2$}] (s); 
\draw (c) to [edge label={$z/3\quad\text{(unique)}$}] (s);
\draw (b) to [edge label={$z/2$}] (t);
\draw (d) to [edge label'={$z/3$}] (t);
\end{tikzpicture}
\end{center} Since $3$ is a {\em unique} witness satisfying $s(3)=30$, in case $10z\to 2z$ is a justification of $20\to 4\righttherefore 30\to d$, we must have $d=2\cdot 3=6$, which shows that $Jus_{\mathfrak M}(20\to 4\righttherefore 30\to 6)$ is non-empty and subset maximal with respect to $6$ (see the Uniqueness \prettyref{lem:UL}). We have thus derived
\begin{align*} \mathfrak M\models 20\to 4\righttherefore 30\to 6.
\end{align*} Notice that the justification $10z\to 2z$ emulates the justification $z\to\frac z 5$ from above. The same line of reasoning with $10z\to z^2$ instead of $10z\to 2z$ shows
\begin{align*} \mathfrak M\models 20\to 4\righttherefore 30\to 9.
\end{align*}

\item $s(z_1,z_2)=5z_1z_2$: The only $\mathbf e_1\in M^2$ satisfying $20=s(\mathbf e_1)$ is given by $(2,2)$. There are two $\mathbf e_2^{(1)},\mathbf e_2^{(2)}\in M^2$ satisfying $30=s(\mathbf e_2^{(1)})=s(\mathbf e_2^{(2)})$ given by $\mathbf e_2^{(1)}=(2,3)$ and $\mathbf e_2^{(2)}=(3,2)$. Moreover, we have five $t_1(z_1,z_2),t_2(z_1,z_2),t_3(z_1,z_2),t_4(z_1,z_2),t_5(z_1,z_2)\in gen_{\mathfrak M}(4)$ satisfying $4=t_i(2,2)$, $1\leq i\leq 5$, given by
\begin{align*} t_1(z_1,z_2)=2z_1,\quad t_2(z_1,z_2)=2z_2,\quad t_3(z_1,z_2)=z_1^2,\quad t_4(z_1,z_2)=z_2^2,\quad t_5(z_1,z_2)=z_1z_2.
\end{align*} We therefore have the following justifications:\footnote{To be more succinct, we have summarized here two diagrams into one separated by semicolons.}
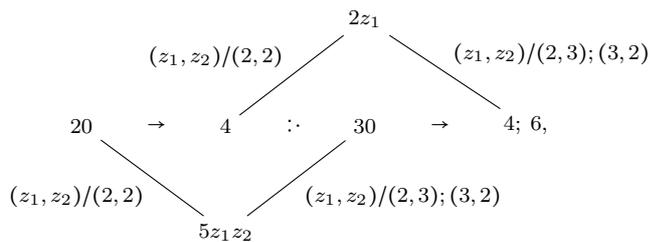
\begin{center}
\begin{tikzpicture}[node distance=1cm and 0.5cm]
% siehe tikz 3.6.1a §3.8
\node (a)               {$20$};
\node (d1) [right=of a] {$\to $};
\node (b) [right=of d1] {$4$};
\node (d2) [right=of b] {$\righttherefore $};
\node (c) [right=of d2] {$30$};
\node (d3) [right=of c] {$\to $};
\node (d) [right=of d3] {$4;\;6$,};
\node (s) [below=of b] {$5z_1z_2$};
\node (t) [above=of c] {$2z_1$};

% siehe tikz 3.6.1a p. 165/166
\draw (a) to [edge label'={$(z_1,z_2)/(2,2)$}] (s); 
\draw (c) to [edge label={$(z_1,z_2)/(2,3);(3,2)$}] (s);
\draw (b) to [edge label={$(z_1,z_2)/(2,2)$}] (t);
\draw (d) to [edge label'={$(z_1,z_2)/(2,3);(3,2)$}] (t);
\end{tikzpicture}
\end{center}
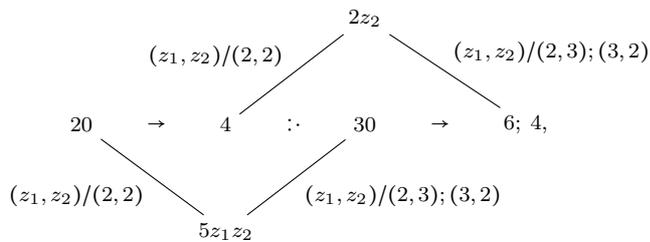
\begin{center}
\begin{tikzpicture}[node distance=1cm and 0.5cm]
% siehe tikz 3.6.1a §3.8
\node (a)               {$20$};
\node (d1) [right=of a] {$\to $};
\node (b) [right=of d1] {$4$};
\node (d2) [right=of b] {$\righttherefore $};
\node (c) [right=of d2] {$30$};
\node (d3) [right=of c] {$\to $};
\node (d) [right=of d3] {$6;\;4,$};
\node (s) [below=of b] {$5z_1z_2$};
\node (t) [above=of c] {$2z_2$};

% siehe tikz 3.6.1a p. 165/166
\draw (a) to [edge label'={$(z_1,z_2)/(2,2)$}] (s); 
\draw (c) to [edge label={$(z_1,z_2)/(2,3);(3,2)$}] (s);
\draw (b) to [edge label={$(z_1,z_2)/(2,2)$}] (t);
\draw (d) to [edge label'={$(z_1,z_2)/(2,3);(3,2)$}] (t);
\end{tikzpicture}
\end{center}
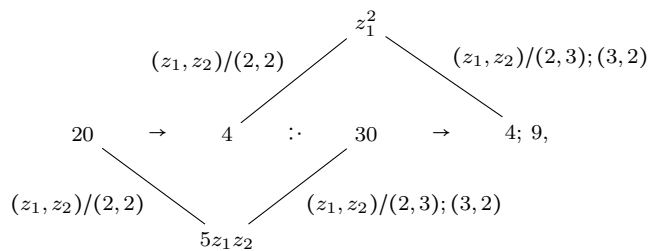
\begin{center}
\begin{tikzpicture}[node distance=1cm and 0.5cm]
% siehe tikz 3.6.1a §3.8
\node (a)               {$20$};
\node (d1) [right=of a] {$\to $};
\node (b) [right=of d1] {$4$};
\node (d2) [right=of b] {$\righttherefore $};
\node (c) [right=of d2] {$30$};
\node (d3) [right=of c] {$\to $};
\node (d) [right=of d3] {$4;\;9,$};
\node (s) [below=of b] {$5z_1z_2$};
\node (t) [above=of c] {$z_1^2$};

% siehe tikz 3.6.1a p. 165/166
\draw (a) to [edge label'={$(z_1,z_2)/(2,2)$}] (s); 
\draw (c) to [edge label={$(z_1,z_2)/(2,3);(3,2)$}] (s);
\draw (b) to [edge label={$(z_1,z_2)/(2,2)$}] (t);
\draw (d) to [edge label'={$(z_1,z_2)/(2,3);(3,2)$}] (t);
\end{tikzpicture}
\end{center}
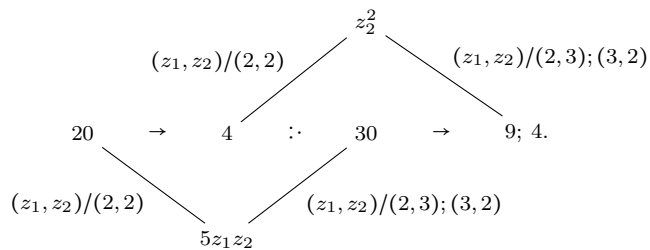
\begin{center}
\begin{tikzpicture}[node distance=1cm and 0.5cm]
% siehe tikz 3.6.1a §3.8
\node (a)               {$20$};
\node (d1) [right=of a] {$\to $};
\node (b) [right=of d1] {$4$};
\node (d2) [right=of b] {$\righttherefore $};
\node (c) [right=of d2] {$30$};
\node (d3) [right=of c] {$\to $};
\node (d) [right=of d3] {$9;\;4.$};
\node (s) [below=of b] {$5z_1z_2$};
\node (t) [above=of c] {$z_2^2$};

% siehe tikz 3.6.1a p. 165/166
\draw (a) to [edge label'={$(z_1,z_2)/(2,2)$}] (s); 
\draw (c) to [edge label={$(z_1,z_2)/(2,3);(3,2)$}] (s);
\draw (b) to [edge label={$(z_1,z_2)/(2,2)$}] (t);
\draw (d) to [edge label'={$(z_1,z_2)/(2,3);(3,2)$}] (t);
\end{tikzpicture}
\end{center}
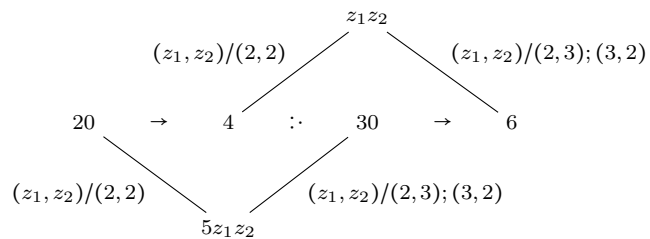
\begin{center}
\begin{tikzpicture}[node distance=1cm and 0.5cm]
% siehe tikz 3.6.1a §3.8
\node (a)               {$20$};
\node (d1) [right=of a] {$\to $};
\node (b) [right=of d1] {$4$};
\node (d2) [right=of b] {$\righttherefore $};
\node (c) [right=of d2] {$30$};
\node (d3) [right=of c] {$\to $};
\node (d) [right=of d3] {$6$};
\node (s) [below=of b] {$5z_1z_2$};
\node (t) [above=of c] {$z_1z_2$};

% siehe tikz 3.6.1a p. 165/166
\draw (a) to [edge label'={$(z_1,z_2)/(2,2)$}] (s); 
\draw (c) to [edge label={$(z_1,z_2)/(2,3);(3,2)$}] (s);
\draw (b) to [edge label={$(z_1,z_2)/(2,2)$}] (t);
\draw (d) to [edge label'={$(z_1,z_2)/(2,3);(3,2)$}] (t);
\end{tikzpicture}
\end{center} Notice the symmetries between the first and second, and between the third and fourth diagrams above; in what follows, we will only mention one of multiple symmetric cases. %We add $5z_1z_2\to 2z_2$, $5z_1z_2\to 2z_1$, $5z_1z_2\to z_1^2$ and $5z_1z_2\to z_2^2$ to $Jus_{\mathfrak M}(20\to 4\righttherefore 30\to 4)$; $5z_1z_2\to 2z_1$ and $5z_1z_2\to 2z_2$ to $Jus_{\mathfrak M}(20\to 4\righttherefore 30\to 6)$; and $5z_1z_2\to z_1^2$ and $5z_1z_2\to z_2^2$ to $Jus_{\mathfrak M}(20\to 4\righttherefore 30\to 9)$.\\ %It is important to emphasize that at this point, we cannot derive the solution $x=4$ as we do not know whether $Jus_{\mathfrak M}(20\to 4\righttherefore 30\to 4)$ is a proper subset of $Jus_{\mathfrak M}(20\to 4\righttherefore 30\to 6)$ or not (but see the case $s(z_1,z_2)=z_1z_2$ below).\\

\item $s(z_1,z_2)=2z_1z_2$: In the remaining cases, we present only the self-explaining diagrams enumerating all valid arrow proportions:
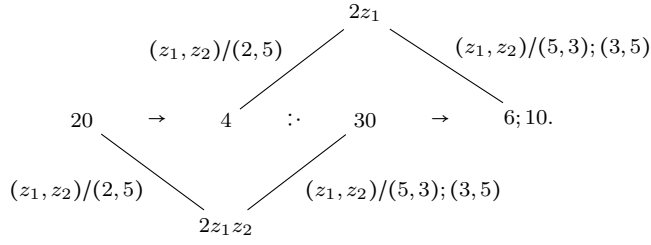
\begin{center}
\begin{tikzpicture}[node distance=1cm and 0.5cm]
% siehe tikz 3.6.1a §3.8
\node (a)               {$20$};
\node (d1) [right=of a] {$\to$};
\node (b) [right=of d1] {$4$};
\node (d2) [right=of b] {$\righttherefore $};
\node (c) [right=of d2] {$30$};
\node (d3) [right=of c] {$\to$};
\node (d) [right=of d3] {$6;10$.};
\node (s) [below=of b] {$2z_1z_2$};
\node (t) [above=of c] {$2z_1$};

% siehe tikz 3.6.1a p. 165/166
\draw (a) to [edge label'={$(z_1,z_2)/(2,5)$}] (s); 
\draw (c) to [edge label={$(z_1,z_2)/(5,3);(3,5)$}] (s);
\draw (b) to [edge label={$(z_1,z_2)/(2,5)$}] (t);
\draw (d) to [edge label'={$(z_1,z_2)/(5,3);(3,5)$}] (t);
\end{tikzpicture}
\end{center} We see here that $20\to 4\righttherefore 30\to 10$ has the justification $2z_1z_2\to 2z_1$ in $\mathfrak M$, which is also a justification of $20\to 4\righttherefore 30\to 6$; on the other hand, we have seen above that $5z_1z_2\to 2z_1$ and $5z_1z_2\to 2z_2$ are justifications of $20\to 4\righttherefore 30\to 6$, but not of $20\to 4\righttherefore 30\to 10$, which shows that, up to this point, $Jus_{\mathfrak M}(20\to 4\righttherefore 30\to 10)$ is strictly contained in $Jus_{\mathfrak M}(20\to 4\righttherefore 30\to 6)$ and therefore not subset maximal with respect to 10. The case $t(z_1,z_2)=2z_2$ is analogous. We further have the arrow proportion
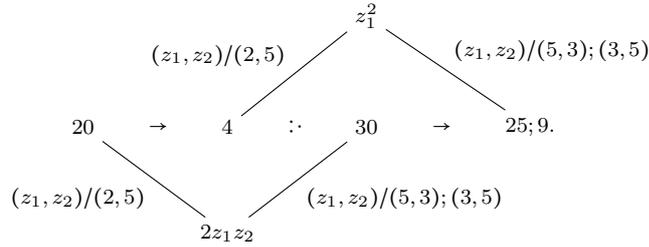
\begin{center}
\begin{tikzpicture}[node distance=1cm and 0.5cm]
% siehe tikz 3.6.1a §3.8
\node (a)               {$20$};
\node (d1) [right=of a] {$\to$};
\node (b) [right=of d1] {$4$};
\node (d2) [right=of b] {$\righttherefore $};
\node (c) [right=of d2] {$30$};
\node (d3) [right=of c] {$\to$};
\node (d) [right=of d3] {$25;9$.};
\node (s) [below=of b] {$2z_1z_2$};
\node (t) [above=of c] {$z_1^2$};

% siehe tikz 3.6.1a p. 165/166
\draw (a) to [edge label'={$(z_1,z_2)/(2,5)$}] (s); 
\draw (c) to [edge label={$(z_1,z_2)/(5,3);(3,5)$}] (s);
\draw (b) to [edge label={$(z_1,z_2)/(2,5)$}] (t);
\draw (d) to [edge label'={$(z_1,z_2)/(5,3);(3,5)$}] (t);
\end{tikzpicture}
\end{center} The only remaining case $t(z_1,z_2)=z_2^2$ is analogous.

\item $s(z_1,z_2,z_3)=z_1z_2z_3$:
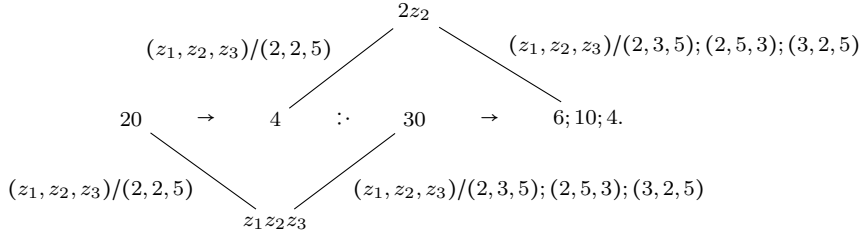
\begin{center}
\begin{tikzpicture}[node distance=1cm and 0.5cm]
% siehe tikz 3.6.1a §3.8
\node (a)               {$20$};
\node (d1) [right=of a] {$\to$};
\node (b) [right=of d1] {$4$};
\node (d2) [right=of b] {$\righttherefore $};
\node (c) [right=of d2] {$30$};
\node (d3) [right=of c] {$\to$};
\node (d) [right=of d3] {$6;10;4$.};
\node (s) [below=of b] {$z_1z_2z_3$};
\node (t) [above=of c] {$2z_2$};

% siehe tikz 3.6.1a p. 165/166
\draw (a) to [edge label'={$(z_1,z_2,z_3)/(2,2,5)$}] (s); 
\draw (c) to [edge label={$(z_1,z_2,z_3)/(2,3,5);(2,5,3);(3,2,5)$}] (s);
\draw (b) to [edge label={$(z_1,z_2,z_3)/(2,2,5)$}] (t);
\draw (d) to [edge label'={$(z_1,z_2,z_3)/(2,3,5);(2,5,3);(3,2,5)$}] (t);
\end{tikzpicture}
\end{center} Here it is sufficient to analyze the three cases $(2,3,5);(2,5,3);(3,2,5)$, where the second argument varies, instead of all six permutations of $(2,3,5)$ since for $2z_2$ only the second argument is relevant. The cases $t(z_1,z_2,z_3)=2z_1$ and $t(z_1,z_2,z_3)=2z_3$ are analogous. We further have the arrow proportions
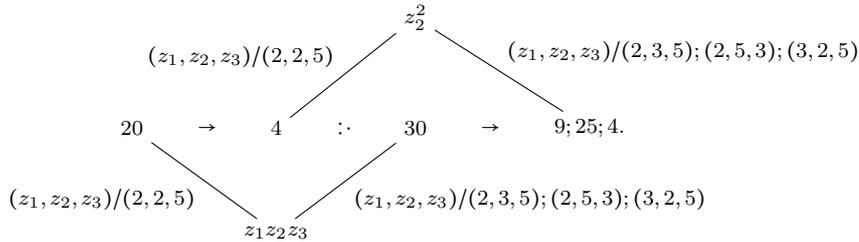
\begin{center}
\begin{tikzpicture}[node distance=1cm and 0.5cm]
% siehe tikz 3.6.1a §3.8
\node (a)               {$20$};
\node (d1) [right=of a] {$\to$};
\node (b) [right=of d1] {$4$};
\node (d2) [right=of b] {$\righttherefore $};
\node (c) [right=of d2] {$30$};
\node (d3) [right=of c] {$\to$};
\node (d) [right=of d3] {$9;25;4$.};
\node (s) [below=of b] {$z_1z_2z_3$};
\node (t) [above=of c] {$z_2^2$};

% siehe tikz 3.6.1a p. 165/166
\draw (a) to [edge label'={$(z_1,z_2,z_3)/(2,2,5)$}] (s); 
\draw (c) to [edge label={$(z_1,z_2,z_3)/(2,3,5);(2,5,3);(3,2,5)$}] (s);
\draw (b) to [edge label={$(z_1,z_2,z_3)/(2,2,5)$}] (t);
\draw (d) to [edge label'={$(z_1,z_2,z_3)/(2,3,5);(2,5,3);(3,2,5)$}] (t);
\end{tikzpicture}
\end{center} The cases $t(z_1,z_2,z_3)=z_1^2$ and $t(z_1,z_2,z_3)=z_3^2$ are analogous.

\item $s(z_1,z_2)=z_1z_2$:
\begin{center}
\begin{tikzpicture}[node distance=1cm and 0.5cm]
% siehe tikz 3.6.1a §3.8
\node (a)               {$20$};
\node (d1) [right=of a] {$\to $};
\node (b) [right=of d1] {$4$};
\node (d2) [right=of b] {$\righttherefore $};
\node (c) [right=of d2] {$30$};
\node (d3) [right=of c] {$\to $};
\node (d) [right=of d3] {$2;15;6;5;3;10$.};
\node (s) [below=of b] {$z_1z_2$};
\node (t) [above=of c] {$z_1$};

% siehe tikz 3.6.1a p. 165/166
\draw (a) to [edge label'={$(z_1,z_2)/(4,5)$}] (s); 
\draw (c) to [edge label={$(z_1,z_2)/\begin{array}{l}
    (2,15);(15,2);(6,5);\\
    (5,6);(3,10);(10,3)
\end{array}$}] (s);
\draw (b) to [edge label={$(z_1,z_2)/(4,5)$}] (t);
\draw (d) to [edge label'={$(z_1,z_2)/\begin{array}{l}
    (2,15);(15,2);(6,5);\\
    (5,6);(3,10);(10,3)
\end{array}$}] (t);
\end{tikzpicture}
\end{center} The case $t(z_1,z_2)=z_2$ is analogous. We further have the arrow proportion
\begin{center}
\begin{tikzpicture}[node distance=1cm and 0.5cm]
% siehe tikz 3.6.1a §3.8
\node (a)               {$20$};
\node (d1) [right=of a] {$\to $};
\node (b) [right=of d1] {$4$};
\node (d2) [right=of b] {$\righttherefore $};
\node (c) [right=of d2] {$30$};
\node (d3) [right=of c] {$\to $};
\node (d) [right=of d3] {$4;225;36;25;9;100$.};
\node (s) [below=of b] {$z_1z_2$};
\node (t) [above=of c] {$z_1^2$};

% siehe tikz 3.6.1a p. 165/166
\draw (a) to [edge label'={$(z_1,z_2)/(2,10)$}] (s); 
\draw (c) to [edge label={$(z_1,z_2)/\begin{array}{l}
    (2,15);(15,2);(6,5);\\
    (5,6);(3,10);(10,3)
\end{array}$}] (s);
\draw (b) to [edge label={$(z_1,z_2)/(2,10)$}] (t);
\draw (d) to [edge label'={$(z_1,z_2)/\begin{array}{l}
    (2,15);(15,2);(6,5);\\
    (5,6);(3,10);(10,3)
\end{array}$}] (t);
\end{tikzpicture}
\end{center} The case $t(z_1,z_2)=z_2^2$ is analogous. %This case is interesting as $z_1z_2\to z_1^2$ is the first justification of $20\to 4\righttherefore 30\to 4$ which is not a justification of $20\to 4\righttherefore 30\to 6$. We will see later that this partly implies that 4 is a solution to \prettyref{equ:20-4_30-z}---it remains to check that there is no other number $d\in\{225,36,25,9,100\}$ such that $Jus_{\mathfrak M}(20\to 4\righttherefore 30\to 4)\subsetneq Jus_{\mathfrak M}(20\to 4\righttherefore 30\to d)$. This is shown by the following diagram (see explanation below):
The next case is:
\begin{center}
\begin{tikzpicture}[node distance=1cm and 0.5cm]
% siehe tikz 3.6.1a §3.8
\node (a)               {$20$};
\node (d1) [right=of a] {$\to $};
\node (b) [right=of d1] {$4$};
\node (d2) [right=of b] {$\righttherefore $};
\node (c) [right=of d2] {$30$};
\node (d3) [right=of c] {$\to $};
\node (d) [right=of d3] {4;\;30;\;12;\;10;\;6;\;20.};
\node (s) [below=of b] {$z_1z_2$};
\node (t) [above=of c] {$2z_1$};

% siehe tikz 3.6.1a p. 165/166
\draw (a) to [edge label'={$(z_1,z_2)/(2,10)$}] (s); 
\draw (c) to [edge label={$(z_1,z_2)/\begin{array}{l}
    (2,15);(15,2);(6,5);\\
    (5,6);(3,10);(10,3)
\end{array}$}] (s);
\draw (b) to [edge label={$(z_1,z_2)/(2,10)$}] (t);
\draw (d) to [edge label'={$(z_1,z_2)/\begin{array}{l}
    (2,15);(15,2);(6,5);\\
    (5,6);(3,10);(10,3)
\end{array}$}] (t);
\end{tikzpicture}
\end{center} %We see here that the solution 4 is justified by $z_1z_2\to 2z_1$, whereas 225, 36, 25, 9, and 100 are not. Moreover, we have seen above that 4 is justified by $z_1z_2\to z_1^2$, whereas 30, 12, 10, 6, and 20 are not. Hence, the two justifications $z_1z_2\to z_1^2$ and $z_1z_2\to 2z_1$ of $20\to 4\righttherefore 30\to 4$ show that $x=4$ is indeed a solution to \prettyref{equ:20-4_30-z}, that is,
% \begin{align*} \mathfrak M\models 20\to 4\righttherefore 30\to 4.
% \end{align*} 
The case $t(z_1,z_2)=2z_2$ is analogous.

\item $s(z)=5z$:
\begin{center}
\begin{tikzpicture}[node distance=1cm and 0.5cm]
% siehe tikz 3.6.1a §3.8
\node (a)               {$20$};
\node (d1) [right=of a] {$\to$};
\node (b) [right=of d1] {$4$};
\node (d2) [right=of b] {$\righttherefore $};
\node (c) [right=of d2] {$30$};
\node (d3) [right=of c] {$\to$};
\node (d) [right=of d3] {$6$.};
\node (s) [below=of b] {$5z$};
\node (t) [above=of c] {$z$};

% siehe tikz 3.6.1a p. 165/166
\draw (a) to [edge label'={$\text{(unique)}\quad z/4$}] (s); 
\draw (c) to [edge label={$z/6\quad\text{(unique)}$}] (s);
\draw (b) to [edge label={$z/4$}] (t);
\draw (d) to [edge label'={$z/6$}] (t);
\end{tikzpicture}
\end{center} Observe that transforming 20 into 4 means removing the prime factor 5 from 20 in \prettyref{equ:20_4_30}---transforming 30 `in the same way' therefore means here to remove the prime factor 5 from 30 in \prettyref{equ:20_4_30}, yielding the solution $x=6$. This shows that $5z\to z$ emulates the justification $z\to\frac z 5$ mentioned at the beginning of the example not available in $\mathfrak M$.

\item $s(z)=2z$: There is a unique $e_1=10$ with $20=s(e_1)$---there is no $t(z)\in gen_{\mathfrak M}(4)$ satisfying $4=t(10)$. This shows that there is no justification $2z\to t(z)$ of $20\to 4\righttherefore 30\to d$, for any $d\in\mathbb N_2$.

\item $s(z)=z$: There is a unique $e_1=20$ with $20=s(e_1)$---there is no $t(z)\in gen_{\mathfrak M}(4)$ satisfying $4=t(20)$. This shows that there is no justification $z\to t(z)$ of $20\to 4\righttherefore 30\to d$, for any $d\in\mathbb N_2$.
\end{enumerate}

To summarize, we have the following non-empty sets of justifications:
\begin{align*} Jus_{\mathfrak M}(20\to 4\righttherefore 30\to 6)&= \left\{
\begin{array}{l}
    10z\to 2z\\
    5z_1z_2\to 2z_1\\
    5z_1z_2\to 2z_2\\
    5z_1z_2\to z_1z_2\\
    2z_1z_2\to 2z_1\\
    2z_1z_2\to 2z_2\\
    z_1z_2z_3\to 2z_1\\
    z_1z_2z_3\to 2z_2\\
    z_1z_2z_3\to 2z_3\\
    z_1z_2\to z_1\\
    z_1z_2\to z_2\\
    z_1z_2\to 2z_1\\
    z_1z_2\to 2z_2\\
    5z\to z
\end{array}
\right\},
\end{align*}
\begin{align*} Jus_{\mathfrak M}(20\to 4\righttherefore 30\to 4)&= \left\{
\begin{array}{l}
    z\to 4\\
     5z_1z_2\to 2z_1\\
     5z_1z_2\to 2z_2\\
     5z_1z_2\to z_1^2\\
     5z_1z_2\to z_2^2\\
     z_1z_2z_3\to 2z_1\\
     z_1z_2z_3\to 2z_2\\
     z_1z_2z_3\to 2z_3\\
     z_1z_2z_3\to z_1^2\\
     z_1z_2z_3\to z_2^2\\
     z_1z_2z_3\to z_3^2\\
     z_1z_2\to z_1^2\\
     z_1z_2\to z_2^2\\
     z_1z_2\to 2z_1\\
     z_1z_2\to 2z_2
\end{array}
\right\},
\end{align*}
\begin{align*} Jus_\mathfrak M(20\to 4\righttherefore 30\to 9)&= \left\{
    \begin{array}{l}
        10z\to z^2\\
        5z_1z_2\to z_1^2\\
        5z_1z_2\to z_2^2\\
        2z_1z_2\to z_1^2\\
        2z_1z_2\to z_2^2\\
        z_1z_2z_3\to z_1^2\\
        z_1z_2z_3\to z_2^2\\
        z_1z_2z_3\to z_3^2\\
        z_1z_2\to z_1^2\\
        z_1z_2\to z_2^2
    \end{array}
    \right\},
\end{align*}
\begin{align*} Jus_\mathfrak M(20\to 4\righttherefore 30\to 25)&= \left\{
    \begin{array}{l}
        2z_1z_2\to z_1^2\\
        2z_1z_2\to z_2^2\\
        z_1z_2z_3\to z_1^2\\
        z_1z_2z_3\to z_2^2\\
        z_1z_2z_3\to z_3^2\\
        z_1z_2\to z_1^2\\
        z_1z_2\to z_2^2
    \end{array}
    \right\}\subsetneq Jus_\mathfrak M(20\to 4\righttherefore 30\to 9),
\end{align*}
\begin{align*} Jus_{\mathfrak M}(20\to 4\righttherefore 30\to 10)&= \left\{
\begin{array}{l}
    2z_1z_2\to 2z_1\\
    2z_1z_2\to 2z_2\\
    z_1z_2z_3\to 2z_1\\
    z_1z_2z_3\to 2z_2\\
    z_1z_2z_3\to 2z_3\\
    z_1z_2\to z_1\\
    z_1z_2\to z_2\\
    z_1z_2\to 2z_1\\
    z_1z_2\to 2z_2
\end{array}
\right\}\subsetneq Jus_{\mathfrak M}(20\to 4\righttherefore 30\to 6),
\end{align*}
\begin{align*} Jus_{\mathfrak M}(20\to 4\righttherefore 30\to d)&= \{z_1z_2\to 2z_1,\;z_1z_2\to 2z_2\}\quad\text{for all $d\in\{30,12,20\}$}\\
    &\subsetneq Jus_{\mathfrak M}(20\to 4\righttherefore 30\to 6),\\
Jus_{\mathfrak M}(20\to 4\righttherefore 30\to d)&=\{z_1z_2\to z_1^2,\;z_1z_2\to z_2^2\}\quad\text{for all $d\in\{225,36,100\}$}\\
    &\subsetneq Jus_{\mathfrak M}(20\to 4\righttherefore 30\to 4),\\
Jus_{\mathfrak M}(20\to 4\righttherefore 30\to 2)&=\{z_1z_2\to z_1,\;z_1z_2\to z_2\}\quad\text{for all $d\in\{2,15,5,3\}$}\\
    &\subsetneq Jus_{\mathfrak M}(20\to 4\righttherefore 30\to 6).
\end{align*} We see that only $x=4,6,9$ have maximal non-empty sets of justifications, proving
\begin{align*} Sol_{\mathfrak M}(20\to 4\righttherefore 30\to x)=\{4,6,9\}.
\end{align*} The solution $x=4$ follows easily with $z\to 4$ and \prettyref{rem:t}. The solution $x=6$ is intuitive and has, among others, the natural justifications $10z\to 2z$ and $5z\to z$ resembling $z\to\frac z 5$ in $(\mathbb Q,\cdot)$.
%The solution $x=4$, on the other hand, is more subtle---its justifications $z_1z_2\to z_1^2$ and $z_1z_2\to 2z_1$ are, in combination, unique to $x=4$ and prove that $Jus_{\mathfrak M}(20\to 4\righttherefore 30\to 4)$ is indeed non-empty and subset maximal with respect to 4. Intuitively, the justification $z_1z_2\to z_1^2$ means that we can factorize $20$ into $2\cdot 10$ and transform its factor $2$ into $2^2=4$---analogously, we can factorize $30$ into $2\cdot 15$ and again transform $2$ into $2^2=4$. It is important to emphasize, however, that $z_1z_2\to z_1^2$ alone is {\em not} enough to justify $20\to 4\righttherefore 30\to 4$, since $z_1z_2\to z_1^2$ also justifies $20\to 4\righttherefore 30\to d$, {\em for all} $d\in\{225,36,25,9,100\}$! This is the reason why we need the justification $z_1z_2\to 2z_1$ as well. 
Lastly, the solution $x=9$ has the unique justification $10z\to z^2$.

We now want to compute all solutions in $\mathfrak M$ to the directed analogical equation
\begin{align}\label{equ:20_4-30_x} 
    20:4\righttherefore 30:x.
\end{align} Recall from \prettyref{fact:Sol} that $d\in M$ is a solution to $20:4\righttherefore 30:x$ iff $d$ is a solution to $20\to 4\righttherefore 30\to x$ and $30$ is a solution to $4\to 20\righttherefore d\to x$ in $\mathfrak M$. We already know that 4, 6, and 9 are the only solutions to $20\to 4\righttherefore 30\to x$ in $\mathfrak M$. We check whether 4, 6, and 9 are solutions to \prettyref{equ:20_4-30_x} separately:
\begin{enumerate}
\item\label{ite:4-20-4-30} There can be no justification of $4\to 20\righttherefore 4\to 30$ in $\mathfrak M$---for example, given the $\mathfrak M$-generalization $2z$ of $4$ in $\mathfrak M$, there is no $\mathfrak M$-generalization $t$ of $20$ and $30$ such that $t(2)=20$ and $t(2)=30$, which can be depicted as follows:
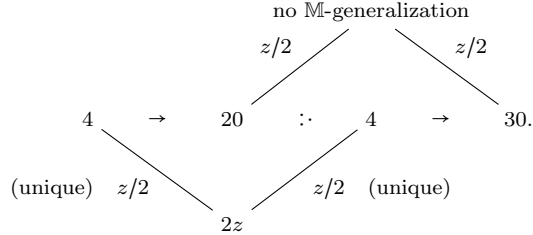
\begin{center}
\begin{tikzpicture}[node distance=1cm and 0.5cm]
% siehe tikz 3.6.1a §3.8
\node (a)               {$4$};
\node (d1) [right=of a] {$\to$};
\node (b) [right=of d1] {$20$};
\node (d2) [right=of b] {$\righttherefore $};
\node (c) [right=of d2] {$4$};
\node (d3) [right=of c] {$\to$};
\node (d) [right=of d3] {$30$.};
\node (s) [below=of b] {$2z$};
\node (t) [above=of c] {no $\mathfrak M$-generalization};

% siehe tikz 3.6.1a p. 165/166
\draw (a) to [edge label'={$\text{(unique)}\quad z/2$}] (s); 
\draw (c) to [edge label={$z/2\quad\text{(unique)}$}] (s);
\draw (b) to [edge label={$z/2$}] (t);
\draw (d) to [edge label'={$z/2$}] (t);
\end{tikzpicture}
\end{center} This shows
\begin{align*} \emptyset=Jus_\mathfrak M(4\to 20\righttherefore 4\to 30)\subsetneq Jus_\mathfrak M(4\to 20\righttherefore 4\to 20)=\{4\to 20,\ldots\},
\end{align*} which means that 30 is not a solution to $4\to 20\righttherefore 4\to x$ and, hence, 4 is not a solution to \prettyref{equ:20_4-30_x}:
\begin{align*} \mathfrak M\not\models 20:4\righttherefore 30:4.
\end{align*} 

\item On the other hand, the justifications
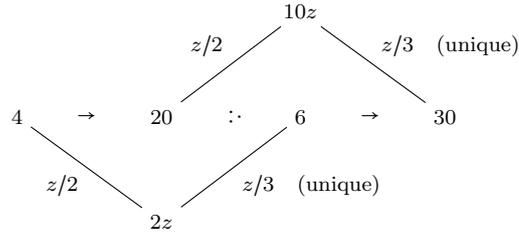
\begin{center}
\begin{tikzpicture}[node distance=1cm and 0.5cm]
% siehe tikz 3.6.1a §3.8
\node (a)               {$4$};
\node (d1) [right=of a] {$\to$};
\node (b) [right=of d1] {$20$};
\node (d2) [right=of b] {$\righttherefore $};
\node (c) [right=of d2] {$6$};
\node (d3) [right=of c] {$\to$};
\node (d) [right=of d3] {$30$};
\node (s) [below=of b] {$2z$};
\node (t) [above=of c] {$10z$};

% siehe tikz 3.6.1a p. 165/166
\draw (a) to [edge label'={$z/2$}] (s); 
\draw (c) to [edge label={$z/3$\quad\text{(unique)}}] (s);
\draw (b) to [edge label={$z/2$}] (t);
\draw (d) to [edge label'={$z/3$\quad\text{(unique)}}] (t);
\end{tikzpicture}
\end{center} and
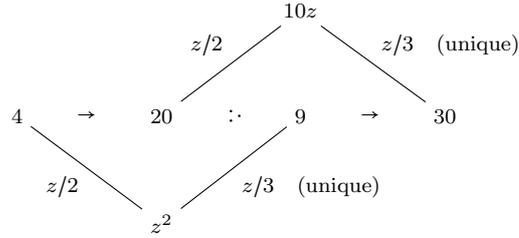
\begin{center}
\begin{tikzpicture}[node distance=1cm and 0.5cm]
% siehe tikz 3.6.1a §3.8
\node (a)               {$4$};
\node (d1) [right=of a] {$\to$};
\node (b) [right=of d1] {$20$};
\node (d2) [right=of b] {$\righttherefore $};
\node (c) [right=of d2] {$9$};
\node (d3) [right=of c] {$\to$};
\node (d) [right=of d3] {$30$};
\node (s) [below=of b] {$z^2$};
\node (t) [above=of c] {$10z$};

% siehe tikz 3.6.1a p. 165/166
\draw (a) to [edge label'={$z/2$}] (s); 
\draw (c) to [edge label={$z/3$\quad\text{(unique)}}] (s);
\draw (b) to [edge label={$z/2$}] (t);
\draw (d) to [edge label'={$z/3$\quad\text{(unique)}}] (t);
\end{tikzpicture}
\end{center} show that $30$ is a solution to $4\to 20\righttherefore 6\to x$ and $4\to 20\righttherefore 9\to x$ in $\mathfrak M$ (Uniqueness \prettyref{lem:UL}), thus proving
\begin{align*} \mathfrak M\models 20:4\righttherefore 30:6 \quad\text{and}\quad \mathfrak M\models 20:4\righttherefore 30:9.
\end{align*} This shows that $6$ and $9$ are the only solutions to $20:4\righttherefore 30:x$ in $\mathfrak M$, that is,
\begin{align*}  Sol_\mathfrak M(20:4\righttherefore 30:x)=\{6,9\}.
\end{align*}
\end{enumerate}

It remains to check whether $6$ and $9$ are solutions to $20:4::30:x$ in \prettyref{equ:20_4_30_x}:

\begin{enumerate}
\item We first verify that $6$ is a solution by showing the remaining relation
\begin{align*} \mathfrak M\models 30:6\righttherefore 20:4.
\end{align*} For this, we prove
\begin{align*} \mathfrak M\models 30\to 6\righttherefore 20\to 4 \quad\text{and}\quad \mathfrak M\models 6\to 30\righttherefore 4\to 20.
\end{align*} The first relation is justified by $5z\to z$, that is, $5z\to z$ is a justification of $30\to 6\righttherefore 20\to d$ in $\mathfrak M$, $d\in M$, iff there are $e_1,e_2\in M$ such that
\begin{align*} 30=5e_1 \quad\text{and}\quad 6=e_1 \quad\text{and}\quad 20=5e_2 \quad\text{and}\quad d=e_2,
\end{align*} which is equivalent to $d=4$. An analogous argument using the justification $z\to 5z$ proves the second relation.

\item Finally, we show that $9$ is a solution by showing the remaining relation
\begin{align*} \mathfrak M\models 30:9\righttherefore 20:4.
\end{align*} For this, we prove
\begin{align*} \mathfrak M\models 30\to 9\righttherefore 20\to 4 \quad\text{and}\quad \mathfrak M\models 9\to 30\righttherefore 4\to 20.
\end{align*} The first relation is justified by $10z\to z^2$ and the second by $z^2\to 10z$ by a similar argument as for $6$ in the previous item.
\end{enumerate} We have thus shown
\begin{align*} Sol_\mathfrak M(20:4::30:x)=\{6,9\}.
\end{align*}
\end{example}

\end{document}